\documentclass[11pt]{article}
\pdfoutput=1
\usepackage[T1]{fontenc}
\usepackage{lmodern}
\usepackage{slantsc}
\usepackage[protrusion=true,expansion=true]{microtype}
\usepackage{breakcites}
\usepackage{amsmath,amsfonts,amsthm}
\usepackage{subcaption}
\usepackage{graphicx}
\usepackage{fullpage}
\usepackage{setspace}
\usepackage[backref=page]{hyperref}
\usepackage{color}
\usepackage{wrapfig}
\usepackage{tikz}
\usetikzlibrary{decorations.pathreplacing}
\usepackage{algorithm}
\usepackage[noend]{algpseudocode}
\usepackage[framemethod=tikz]{mdframed}
\usepackage{xspace}
\usepackage{pgfplots}
\usepackage{framed}
\usepackage{thmtools}
\usepackage{thm-restate}
\usepackage{tabu}
\usepackage{fancyhdr}
\pgfplotsset{compat=1.5}

\newtheorem{theorem}{Theorem}[section]
\makeatletter
\@addtoreset{theorem}{section} 

\renewcommand{\thetheorem}{%
  \ifnum\value{subsection}=0 
    \thesection.\arabic{theorem}%
  \else
    \ifnum\value{subsubsection}=0 
      \thesubsection.\arabic{theorem}%
    \else 
      \thesubsubsection.\arabic{theorem}%
    \fi
  \fi
}
\makeatother

\newtheorem{corollary}[theorem]{Corollary}
\newtheorem{lemma}[theorem]{Lemma}
\newtheorem{proposition}[theorem]{Proposition}
\newtheorem{definition}[theorem]{Definition}

\newtheorem{invariant}[theorem]{Invariant}

\newtheorem{fact}[theorem]{Fact}

\newenvironment{proofof}[1]{\begin{trivlist} \item {\bf Proof
#1:~~}}
  {\qed\end{trivlist}}

\newcommand{\namedref}[2]{\hyperref[#2]{#1~\ref*{#2}}}
\newcommand{\thmlab}[1]{\label{thm:#1}}
\newcommand{\thmref}[1]{\namedref{Theorem}{thm:#1}}
\newcommand{\lemlab}[1]{\label{lem:#1}}
\newcommand{\lemref}[1]{\namedref{Lemma}{lem:#1}}

\newcommand{\invarlab}[1]{\label{invar:#1}}
\newcommand{\invarref}[1]{\namedref{Invariant}{invar:#1}}
\newcommand{\corlab}[1]{\label{cor:#1}}
\newcommand{\corref}[1]{\namedref{Corollary}{cor:#1}}
\newcommand{\seclab}[1]{\label{sec:#1}}
\newcommand{\secref}[1]{\namedref{Section}{sec:#1}}

\newcommand{\factlab}[1]{\label{fact:#1}}
\newcommand{\factref}[1]{\namedref{Fact}{fact:#1}}

\newcommand{\figlab}[1]{\label{fig:#1}}
\newcommand{\figref}[1]{\namedref{Figure}{fig:#1}}

\newcommand{\tablelab}[1]{\label{tab:#1}}
\newcommand{\tableref}[1]{\namedref{Table}{tab:#1}}
\newcommand{\deflab}[1]{\label{def:#1}}
\newcommand{\defref}[1]{\namedref{Definition}{def:#1}}

\newcommand{\propref}[1]{\namedref{Proposition}{prop:#1}}
\newcommand{\proplab}[1]{\label{prop:#1}}

\def \TVD    {\mdef{\mathsf{TVD}}}
\def \trunc {\textup{trunc}}

\newcommand\norm[1]{\left\lVert#1\right\rVert}
\newcommand{\PPr}[1]{\ensuremath{\mathbf{Pr}\left[#1\right]}}
\newcommand{\PPPr}[2]{\ensuremath{\underset{#1}{\mathbf{Pr}}\left[#2\right]}}
\newcommand{\Ex}[1]{\ensuremath{\mathbb{E}\left[#1\right]}}
\newcommand{\EEx}[2]{\ensuremath{\underset{#1}{\mathbb{E}}\left[#2\right]}}
\renewcommand{\O}[1]{\ensuremath{\mathcal{O}\left(#1\right)}}

\newcommand{\eps}{\varepsilon}

\def \calA    {\mdef{\mathcal{A}}}
\def \calB    {\mdef{\mathcal{B}}}
\def \calC    {\mdef{\mathcal{C}}}
\def \calD    {\mdef{\mathcal{D}}}
\def \calE    {\mdef{\mathcal{E}}}
\def \calF    {\mdef{\mathcal{F}}}
\def \calG    {\mdef{\mathcal{G}}}

\def \calI    {\mdef{\mathcal{I}}}

\def \calL    {\mdef{\mathcal{L}}}
\def \calM    {\mdef{\mathcal{M}}}
\def \calN    {\mdef{\mathcal{N}}}

\def \calP    {\mdef{\mathcal{P}}}

\def \bA    {\mdef{\mathbf{A}}}
\def \bB    {\mdef{\mathbf{B}}}

\def \bG    {\mdef{\mathbf{G}}}
\def \bH    {\mdef{\mathbf{H}}}

\def \bM    {\mdef{\mathbf{M}}}
\def \bP    {\mdef{\mathbf{P}}}
\def \bR    {\mdef{\mathbf{R}}}
\def \bS    {\mdef{\mathbf{S}}}

\def \bV    {\mdef{\mathbf{V}}}
\def \bU    {\mdef{\mathbf{U}}}
\def \bX    {\mdef{\mathbf{X}}}
\def \bSigma  {\mdef{\mathbf{\Sigma}}}
\def \bmu    {\mdef{\mathbf{\mu}}}

\def \bb    {\mdef{\mathbf{b}}}
\def \be    {\mdef{\mathbf{e}}}

\def \bq    {\mdef{\mathbf{q}}}

\def \bu    {\mdef{\mathbf{u}}}
\def \bw    {\mdef{\mathbf{w}}}
\def \bv    {\mdef{\mathbf{v}}}
\def \bx    {\mdef{\mathbf{x}}}
\def \by    {\mdef{\mathbf{y}}}
\def \bz    {\mdef{\mathbf{z}}}
\def \bg    {\mdef{\mathbf{g}}}
\def \bfEta    {\mdef{\mathbf{\eta}}}

\newcommand{\mdef}[1]{{\ensuremath{#1}}\xspace}  

\DeclareMathOperator*{\argmin}{argmin}
\DeclareMathOperator*{\argmax}{argmax}

\DeclareMathOperator*{\poly}{poly}
\DeclareMathOperator*{\supp}{supp}

\newcommand{\abs}[1]{\mdef{\left|#1\right|}}         

\newcommand{\ignore}[1]{}

\newif\ifnotes\notestrue 
\ifnotes
\newcommand{\samson}[1]{\textcolor{purple}{{\bf (Samson:} {#1}{\bf ) }} \marginpar{\tiny\bf
             \begin{minipage}[t]{0.5in}
               \raggedright S:
            \end{minipage}}}
\newcommand{\david}[1]{\textcolor{purple}{{\bf (David:} {#1}{\bf ) }} \marginpar{\tiny\bf
             \begin{minipage}[t]{0.5in}
               \raggedright D:
            \end{minipage}}} 

\newcommand{\elena}[1]{\textcolor{teal}{{\bf (Elena:} {#1}{\bf ) }} \marginpar{\tiny\bf
             \begin{minipage}[t]{0.5in}
               \raggedright S:
            \end{minipage}}}
\else
\newcommand{\samson}[1]{}
\newcommand{\david}[1]{}
\newcommand{\elena}[1]{}
\fi

\makeatletter
\renewcommand*{\@fnsymbol}[1]{\textcolor{mahogany}{\ensuremath{\ifcase#1\or *\or \dagger\or \ddagger\or
 \mathsection\or \triangledown\or \mathparagraph\or \|\or **\or \dagger\dagger
   \or \ddagger\ddagger \else\@ctrerr\fi}}}
\makeatother

\providecommand{\email}[1]{\href{mailto:#1}{\nolinkurl{#1}\xspace}}

\definecolor{mahogany}{rgb}{0.75, 0.25, 0.0}
\definecolor{darkblue}{rgb}{0.0, 0.0, 0.55}
\definecolor{darkpastelgreen}{rgb}{0.01, 0.75, 0.24}
\definecolor{darkgreen}{rgb}{0.0, 0.2, 0.13}
\definecolor{darkgoldenrod}{rgb}{0.72, 0.53, 0.04}
\definecolor{darkred}{rgb}{0.55, 0.0, 0.0}
\definecolor{forestgreenweb}{rgb}{0.13, 0.55, 0.13}
\definecolor{greencss}{rgb}{0.0, 0.5, 0.0}
\definecolor{bleudefrance}{rgb}{0.19, 0.55, 0.91}

\hypersetup{
     colorlinks   = true,
     citecolor    = mahogany,
	 linkcolor	  = darkpastelgreen
}

\fancypagestyle{pg}
{
\lhead{}
\rhead{}
\cfoot{--\ \thepage\ --}

}

\AtBeginDocument{%
  \DeclareFontShape{T1}{lmr}{m}{scit}{<->ssub*lmr/m/scsl}{}%
}

\allowdisplaybreaks
\begin{document}
\title{Lifting Linear Sketches: Optimal Bounds and Adversarial Robustness}
\author{Elena Gribelyuk \\ Princeton University \\ \email{eg5539@princeton.edu}  
\and 
Honghao Lin \\ Carnegie Mellon University \\ \email{honghaol@andrew.cmu.edu} 
\and 
David P. Woodruff \\ Carnegie Mellon University \\ \email{dwoodruf@andrew.cmu.edu}
\and
Huacheng Yu \\ Princeton University \\ \email{hy2@cs.princeton.edu}
\and
Samson Zhou \\ Texas A\&M University \\ \email{samsonzhou@gmail.com}}
\date{\today}
\allowdisplaybreaks
\maketitle
\begin{abstract}
We introduce a novel technique for ``lifting'' dimension lower bounds for linear sketches in the real-valued setting to dimension lower bounds for linear sketches with polynomially-bounded integer entries when the input is a polynomially-bounded integer vector. Using this technique, we obtain the first optimal sketching lower bounds for discrete inputs in a data stream, for classical problems such as approximating the frequency moments, estimating the operator norm, and compressed sensing. Additionally, we lift the adaptive attack of Hardt and Woodruff (STOC, 2013) for breaking any real-valued linear sketch via a sequence of real-valued queries, and show how to obtain an attack on any integer-valued linear sketch using integer-valued queries. This shows that there is no linear sketch in a data stream with insertions and deletions that is adversarially robust for approximating any $L_p$ norm of the input, resolving a central open question for adversarially robust streaming algorithms. To do so, we introduce a new pre-processing technique of independent interest which, given an integer-valued linear sketch, increases the dimension of the sketch by only a constant factor in order to make the orthogonal lattice to its row span smooth. This pre-processing then enables us to leverage results in lattice theory on discrete Gaussian distributions and reason that efficient discrete sketches imply efficient continuous sketches. Our work resolves open questions from the Banff '14 and '17 workshops on Communication Complexity and Applications, as well as the STOC '21 and FOCS '23 workshops on adaptivity and robustness. 
\end{abstract}
\thispagestyle{pg}
\setcounter{page}{1}

\section{Introduction}
The \emph{streaming model} has emerged as a popular large-scale computational model for the analysis of datasets too large to be stored in memory, such as the database logs generated from commercial transactions, financial markets, physical sensors, social network activity, scientific observations, and virtual traffic monitoring. 
To capture these applications, the one-pass streaming model implicitly defines an underlying dataset, given a sequence of updates that are each irrevocably discarded after processing. 
The goal is to approximate, compute, detect, or identify a predetermined property of the input dataset using memory (i.e., space complexity) that is sublinear in both the size of the input dataset and the length of the data stream. 

Linear sketches are one of the most common approaches for achieving sublinear space algorithms in the streaming model. 
Formally, if the data stream consists of updates to an underlying vector $\bx\in\mathbb{R}^n$, a linear sketch maintains a significantly smaller vector $\bA\bx\in\mathbb{R}^{r}$ with $r\ll n$, through a carefully designed matrix $\bA\in\mathbb{R}^{r\times n}$ that can be efficiently stored. 
At the end of the stream, the algorithm can then post-process $\bA\bx$ to approximately recover the desired properties of $\bx$. 
Due to their simplicity and versatility, linear sketches are widely used in streaming algorithms, such as CountSketch~\cite{CharikarCF04} and CountMin~\cite{CormodeM05} for frequency estimation, AMS~\cite{AlonMS99} for moment estimation, and Johnson-Lindenstrauss~\cite{johnson1984extensions} for distance estimation. 
All known turnstile streaming algorithms are linear sketches for sufficiently large stream lengths and in fact, it is known that the optimal turnstile streaming algorithms can be captured by linear sketches under certain conditions \cite{LiNW14,AiHLW16,KallaugherP20}. 

The study of lower bounds for linear sketches is fundamental to our understanding of the capabilities and limitations of streaming algorithms. 
A popular method to showing lower bounds for linear sketches is to define two ``hard'' input distributions $\calD_1$ and $\calD_2$ that exhibit a desired gap for the problem of interest, and then show that the total variation distance between $\bA\bx$ and $\bA\by$ is small for $\bx\sim\calD_1$ and $\by\sim\calD_2$. 
Because such lower bounds should hold for any sketching matrix $\bA\in\mathbb{R}^{r\times n}$ with $r\ll n$, the distributions $\calD_1$ and $\calD_2$ are often chosen to be multivariate Gaussians (or somewhat ``near'' Gaussian), so that by rotational invariance, $\bA\bx$ and $\bA\by$ are also multivariate Gaussians with the appropriate covariances. 

As a simple example, consider the problem of estimating $\|\bx\|_2^2$. We can assume, without loss of generality, that the rows of $\bA \in \mathbb{R}^{r \times n}$ are orthonormal, since we can always change the row basis by multiplying $\bA \bx$ by a change of basis matrix $\bR$ in post-processing, obtaining $\bR \bA \bx$. Then, if $\calD_1 \sim N(0, I_n)$ for a Gaussian distribution with mean zero and identity covariance, and $\calD_2 \sim N(0, (1+\eps)I_n)$, then if $\bx$ is drawn from $\calD_1$, then $\bA \bx$ is equal in distribution to $N(0, I_r)$, while if $\bx$ is drawn from $\calD_2$, then $\bA \bx$ is equal in distribution to $N(0, (1+\eps) I_r)$. Then using standard results on the number of samples needed to distinguish two normal distributions, we must have $r = \Omega(\eps^{-2} \log(1/\delta))$. 

The above technique has been used to prove lower bounds for $L_p$ estimation~\cite{GangulyW18}, compressed sensing~\cite{PriceW11,PriceW13}, eigenvalue estimation and PSD testing~\cite{NeedellSW22,SwartworthW23}, operator norm and Ky Fan norm~\cite{LiW16}, and norm estimation for adversarially robust streaming algorithms~\cite{HardtW13}. These works lower bound the sketching dimension when both the inputs and the sketch stored by the algorithm can contain arbitrary real numbers. 

Unfortunately, these lower bounds are unsatisfying in a number of ways. While the lower bounds hold for sketches even with real number inputs, they require the entries of the {\it input vector} $\bx$ to be real-valued as well. This is inherent: if $\bx$ has entries with finite bit complexity, we could use large enough precision entries in $\bA$ to exactly recover $\bx$ from $\bA$. However, in order to meaningfully discuss space complexity, the streaming model is defined on a stream of additive updates to $\bf x$ with finite precision. Without the finite precision bounds, one could obtain memory lower bounds for data stream algorithms that are arbitrarily large. One could try to discretize the input distribution to the above problems, but then the distribution is no longer rotationally invariant, and a priori it is not clear that information about the input is revealed by truncating low order bits. 
These issues mean that none of the above lower bounds actually apply to the data stream model! 
%
%
This phenomenon is real, as the data stream lower bounds we have for problems such as frequency moment estimation, which dates back to the original work of \cite{AlonMS99}, are worse than those that we have for sketching dimension lower bounds.
While this gap is at most logarithmic for the frequency moment estimation, for other problems such as approximating the operator norm of a matrix, the gap is polynomial; see the discussion following \thmref{thm:main:apps} for further details. 
Therefore, a long-standing open question has been:
\begin{quote}
{\it Is it possible to lift linear sketch lower bounds for continuous inputs to obtain linear sketch lower bounds for discrete inputs?} 
\end{quote}
There are multiple natural approaches to lift existing lower bounds, such as rounding the continuous inputs from hard distributions, or using discrete analogs of the continuous hard distributions. 
However, there are inherent shortcomings for each of these approaches, which we discuss in \secref{sec:overview}. Thus, the existence of lower bounds for discrete inputs (and implicitly the existence of lifting techniques) has been previously asked on multiple occasions~\cite{banff2014workshop,banff2017workshop}. 

Another line of work in the data stream literature that has also recently run into a similar issue is that of adversarially robust streaming algorithms \cite{Ben-EliezerJWY22}. 
In this model, a sequence of updates $u_1,\ldots,u_m$ is adaptively generated as an input stream to an algorithm. 
In particular, subsequent updates may be generated based on previous outputs of the algorithm, which is still required to correctly approximate or compute a fixed function at all times in the stream using space sublinear in the size of the input dataset. 
A major open question is whether there exists an adversarially robust streaming algorithm for approximating the squared Euclidean norm $F_2$ of $\bx$ in the presence of insertions and deletions to ${\bf x} \in \{-\textrm{poly}(n), \ldots, \textrm{poly}(n)\}^{n}$ and using space sublinear in $n$. 
The existence of an attack that uses \textit{discrete inputs} against adversarially robust linear sketches has recently been posed as an open question in multiple workshops~\cite{stoc2021workshop,focs2023workshop}. 
Just as we had for sketching lower bounds over the reals, there is an attack \cite{HardtW13} that breaks any linear sketch when the inputs in the stream are real-number valued. 
However, this attack crucially uses rotational invariance and fails if the inputs are required to be discrete. Thus we ask:

\begin{quote}
{\it Does there exist a sublinear space adversarially robust $F_2$-estimation linear sketch in a finite precision stream?} 
\end{quote}

\subsection{Our Contributions: Lifting Framework and Applications}
We give a technique for lifting linear sketch lower bounds for continuous inputs to achieve linear sketch lower bounds for discrete inputs, thereby answering the above open questions. 

\begin{theorem}[Informal, see \thmref{thm:lifting}]
\thmlab{thm:lifting:informal}
Let $\calI$ denote an arbitrary distribution on $\mathbb{Z}^n$, $\bS^\top\bS$ be a full-rank covariance matrix, and $\delta\ge\frac{1}{\poly(n)}$ be a failure parameter. 
Let $f$ be a $\frac{\delta}{3}$-smooth function on the distribution $X+Y$, where $X$ is drawn from the discrete Gaussian distribution $\calD(0,\bS^\top\bS)$ and $Y\sim\calI$ and suppose $f$ admits an integer sketch $\bA\in\mathbb{Z}^{r\times n}$ with $r\le n$ and an arbitrary post-processing function $g$, such that:
\begin{enumerate}
\item 
$\bA$ has polynomially-bounded integer entries.
\item
$g(\bA\bx)=f(\bx)$ with probability at least $1-\frac{\delta}{3}$ over $\bx\sim\calD(0,\bS^\top\bS)+\calI$. 
\item
The sketching matrix $\bA$ has a (not necessarily orthogonal) integer basis and the singular values of the covariance matrix $\bS^\top\bS$ are sufficiently large. 
\end{enumerate}
There exists a post-processing function $h$ with another sketching matrix $\bA' \in \mathbb{R}^{4r \times n}$ such that with probability at least $1-\delta$ over the distribution $\bx=X+Y$, where $X$ is drawn from the continuous Gaussian distribution $\calN(0,\bS^\top\bS)$ and $Y\sim\calI$, we have $h(\bA'\bx)=f(\bx)$. 
\end{theorem}
Intuitively, \thmref{thm:lifting:informal} shows that a linear sketch for a ``smooth'' function $f$ that succeeds on a distribution on discrete inputs of the form $X+Y$, where $X$ is drawn from a discrete Gaussian distribution (for a formal definition, see \defref{def:discrete:Gaussian}) and $Y$ is drawn from an arbitrary distribution over integer vectors, can be transformed into a linear sketch that succeeds on a distribution on real-valued inputs of the form $Z+Y$, where $Z$ is a continuous Gaussian distribution with the same covariance and $Y$ is the same distribution over integer vectors. 
Therefore, any lower bound for linear sketches on real-valued inputs based on distributions with such a form can be lifted to obtain a lower bound for linear sketches on discrete inputs. Our proof of the above theorem uses interesting results in lattice theory, which to the best of our knowledge have not been applied to the data stream literature. We first define the following notion:
\begin{definition}[Integer sketch]
Given a data stream of length at most $m = \poly(n)$ that implicitly defines an underlying vector of dimension $n$ with entries in $\{-\poly(n), \ldots, \poly(n)\}$, an \emph{integer sketch} is an algorithm that uses a linear sketch with integer entries of magnitude at most $M = \poly(n)$. In some applications below, our input is itself a matrix, which we also assume has entries in $\{-\poly(n), \ldots, \poly(n)\}$, and which we first vectorize and then we multiply it by our linear sketch. For example, if the matrix is $n \times n$ we flatten it to an $n^2$-dimensional vector that is the input to the linear sketch. 
\end{definition}
We use \thmref{thm:lifting:informal} to prove lower bounds for integer sketches for a number of applications. At a high-level, our approach shows that we can use discrete Gaussian queries to ``simulate'' the hard distribution from previously-known lower bounds which relied on Gaussian queries.
For example, for the problem of estimating the $L_2$ norm $\|\bx\|_2^2$, the hard distribution over the reals was to sample $\bx \sim N(0, I_n)$ or $\bx \sim N(0, (1+\epsilon) I_n)$. For the discrete attack, we sample $\bx \sim \calD(0, I_n)$ or $\bx \sim \calD(0, (1+\epsilon)I_n)$, and argue that for some small uniform ``noise'' $\eta$, $\bA\bx + \eta$ is statistically close to $\bA\bg$ for $\bg$ drawn from the corresponding continuous Gaussian. We refer the reader to \secref{sec:overview} for more details on how we simulate a continuous Gaussian via a discrete Gaussian. 

The first result below is known, though our framework gives a remarkably simple proof, so we include it as a warm-up. The remaining results provide substantial improvements over state-of-the-art and result in optimal bounds for fundamental problems, such as frequency moments for $p > 2$ and operator norm approximation:
\begin{theorem}
\thmlab{thm:main:apps}
\begin{enumerate}
\item 
($L_p$ norm estimation, $p\in[1,2]$) 
Given $\eps\in(0,1)$ and $p\in[1,2]$, any integer sketch that outputs a $(1+\eps)$-approximation to the $L_p$ norm of a vector  with probability $1-\delta$ has dimension $\Omega\left(\frac{1}{\eps^2}\log\frac{1}{\delta}\right)$. (See \lemref{lem:lp:small:lb}). 

\item 
($L_p$ norm estimation, $p>2$) 
For any constant $\eps\in(0,1)$ and $p>2$, any integer sketch that outputs a $(1+\eps)$-approximation to the $L_p$ norm of a vector with probability at least $\frac{2}{3}$ has dimension $\Omega\left(n^{1-2/p}\log n  \right)$. (See \lemref{lem:lp:large:lb}). 

\item 
(Accurate approximation to operator norm)
Given an approximation parameter $\eps\in\left(0,\frac{1}{3}\right)$, any integer sketch that outputs a $(1+\eps)$-approximation to the operator norm of a $d/\eps^2 \times d$ matrix with probability at least $\frac{5}{6}$ uses sketching dimension $\Omega\left(\frac{d^2}{\eps^2}\right)$. (See \lemref{lem:op:eps:lb}). 

\item
(Large approximation to operator norm) 
Given an approximation parameter $\alpha \ge 1 + c$ for an arbitrary small constant $c$, any integer sketch that estimates $\|X \|_{op}$ for an $n \times n$ matrix $\bX$ within a factor of $\alpha$ with error probability $\leq 1/6$ requires sketching dimension $\Omega(n^2/\alpha^4)$. (See~\corref{cor:op_lb}).

\item 
(Ky Fan norm) 
There exists an absolute constant $c > 0$ such that any integer linear sketch that estimates $\norm{\bX}_{F_S}$ for an $n \times n$ matrix $\bX$ and $s \le \O{\sqrt{n}}$ within a factor of $1 + c$ with error probability $1/6$ requires sketching dimension $\Omega(n^2 / s^2)$.  (See~\corref{cor:KyFan_lb}).

\item
(Eigenvalue estimation)
Given an approximation parameter $\eps\in\left(0,\frac{1}{3}\right)$, any integer sketch that outputs additive $\eps\cdot\|\bM\|_F$ approximations to the eigenvalues of a $d \times d$ matrix $\bM$ with probability at least $\frac{3}{4}$ uses sketching dimension $\Omega\left(\frac{1}{\eps^4}\right)$ for $d=\Omega\left(\frac{1}{\eps^2}\right)$. (See \thmref{thm:eigen:lb}). 

\item 
(PSD testing)
Given a distance parameter $\eps$, any integer sketch that reads $\bM\in\mathbb{Z}^{d\times d}$ and serves as a two-sided tester for whether $\bM$ is positive semidefinite or $\eps$-far from positive semidefinite in $\ell_p$ distance with probability at least $\frac{3}{4}$ requires (a) sketching dimension $\Omega\left(\frac{1}{\eps^{2p}}\right)$ for $p\in[1,2]$, (b) sketching dimension $\Omega\left(\frac{1}{\eps^4}d^{2-4/p}\right)$ for $p\in(2,\infty)$, and (c) sketching dimension $\Omega(d^2)$ for $p=\infty$. 
(See \thmref{thm:psd:lb}).
\item 
(Compressed sensing) Given an approximation parameter $\eps > \sqrt{(k \log n) / n}$, any integer sketch that reads $\bx \in \mathbb{Z}^{n}$ and outputs an $\O{\frac{\eps n}{\log n}}$-sparse $\bx'$ that $\norm{\bx - \bx'}_2 \le (1 + \eps) \min_{\text{$k$-sparse } \tilde{\bx}} \norm{\tilde{\bx} -  \bx}_2$ with high constant probability requires dimension $\Omega\left(\frac{k}{\eps}\log\frac{n}{k}\right)$. 
(See~\lemref{lem:compress_sensing_lb}). 
\end{enumerate}
\end{theorem}

We discuss the numerous implications of the results in \thmref{thm:main:apps}. 
For $L_p$ estimation with $p>2$, our discrete sketching dimension lower bound of $\Omega(n^{1-2/p}\log n)$ improves upon the existing discrete bounds of $\Omega(n^{1-2/p})$ by \cite{LiW13,WoodruffZ21b}, though it should be noted that the latter lower bound still holds for insertion-only streams. 
For operator norm, a folklore result shows that $\Omega(d)$ bits of space is necessary, due to a simple reduction from the set disjointness communication problem on the diagonal of an $d\times d$ matrix. 
It then follows that operator norm on discrete inputs and memory states requires sketching dimension $\Omega\left(\frac{d}{\log d}\right)$, as each dimension uses $\O{\log n}$ bits.  
This discrete lower bound further translates to an $\Omega\left(\frac{1}{\eps^2\log d}\right)$ lower bound for eigenvalue approximation and PSD testing since $d$ should be roughly $\frac{1}{\eps^2}$ to exhibit the desired gap for those problems. 
By comparison, our framework in \thmref{thm:lifting:informal} is able to lift existing continuous lower bounds to show a significantly stronger and optimal lower bound of sketching dimension $\Omega(d^2)$ for integer sketches for operator norm approximation, and $\Omega\left(\frac{1}{\eps^4}\right)$ for eigenvalue estimation and PSD testing. For sparse recovery, a bit complexity lower bound of $\Omega\left(\frac{k}{\eps}\log n\right)$ can be obtained using similar ideas to those in \cite{PriceW11}, by taking a direct sum of $k$ copies of the indexing communication problem on $\frac{1}{\eps}$ bits and then growing in $\log n$ scales with augmented indexing. 
This results in a discrete sketching dimension lower bound of $\Omega\left(\frac{k}{\eps}\right)$. 
In contrast, our results achieve an optimal lower bound of $\Omega\left(\frac{k}{\eps}\log\frac{n}{k}\right)$. 
We summarize our results in \tableref{table:apps}.

\begin{table}[!htb]
\centering
\resizebox{\columnwidth}{!}{
\begin{tabular}{c|c|c|c}
& Existing Real-Valued LB & Previous Discrete LB & Our Discrete LB \\\hline
$L_p$ Estimation, $p\in[1,2]$ & $\Omega\left(\frac{1}{\eps^2}\log\frac{1}{\delta}\right)$~\cite{GangulyW18} & $\Omega\left(\frac{1}{\eps^2}\log\frac{1}{\delta}\right)$~\cite{JayramW13} & $\Omega\left(\frac{1}{\eps^2}\log\frac{1}{\delta}\right)$ (\lemref{lem:lp:small:lb}) \\
$L_p$ Estimation, $p>2$ & $\Omega\left(n^{1-2/p}\log n\right)$~\cite{GangulyW18} & $\Omega\left(n^{1-2/p}\right)$~\cite{LiW13,WoodruffZ21b} & $\Omega\left(n^{1-2/p}\log n\right)$ (\lemref{lem:lp:large:lb}) \\
Operator Norm & $\Omega\left(\frac{d^2}{\eps^2}\right)$ \cite{LiW16} & $\Omega\left(\frac{d}{\log d}\right)$ (folklore) & $\Omega\left(\frac{d^2}{\eps^2}\right)$ (\lemref{lem:op:eps:lb})\\
Eigenvalue Estimation & $\Omega\left(\frac{1}{\eps^4}\right)$~\cite{NeedellSW22} & $\Omega\left(\frac{1}{\eps^2\log d}\right)$ (folklore) & $\Omega\left(\frac{1}{\eps^4}\right)$ (\thmref{thm:eigen:lb}) \\
PSD Testing & $\Omega\left(\frac{1}{\eps^4}\right)$~\cite{SwartworthW23} & $\Omega\left(\frac{1}{\eps^2\log d}\right)$ (folklore) & $\Omega\left(\frac{1}{\eps^4}\right)$ (\thmref{thm:psd:lb}) \\
Compressed Sensing & $\Omega\left(\frac{k}{\eps}\log\frac{n}{k}\right)$ \cite{PriceW11} & $\Omega\left(\frac{k}{\eps}\right)$ (folklore) & $\Omega\left(\frac{k}{\eps}\log\frac{n}{k}\right)$ (\lemref{lem:compress_sensing_lb}) \\
\end{tabular}
}
\caption{Applications of \thmref{thm:lifting:informal}. For $L_p$ estimation, $p>2$, we omit dependencies on $\eps$ and $\delta$.}
\tablelab{table:apps}
\end{table}

\subsection{Our Contributions: Streaming against Adaptive Adversaries}
We can apply \thmref{thm:lifting:informal} in a non-black box manner to achieve an important implication for adversarial robustness in the streaming model. 
In the adversarially robust streaming model~\cite{Ben-EliezerY20,HassidimKMMS20,AlonBDMNY21,BravermanHMSSZ21,KaplanMNS21,WoodruffZ21,BeimelKMNSS22,Ben-EliezerEO22,Ben-EliezerJWY22,ChakrabartiGS22,AvdiukhinMYZ19,AssadiCGS23,AttiasCSS23,CherapanamjeriS23,DinurSWZ23,GribelyukLWYZ24,WoodruffZ24}, a sequence of updates is adaptively generated as an input data stream to an algorithm. 
Specifically, a streaming algorithm $\mathcal{A}$ is \textit{adversarially robust} for some estimation function $g:\mathbb{Z}^n \rightarrow \mathbb{R}$ if $\mathcal{A}$ satisfies the following requirement.

\begin{definition}\cite{Ben-EliezerJWY22}
Let $g:\mathbb{Z}^n \rightarrow \mathbb{R}$ be a fixed function. 
Then, for any $\eps > 0$ and $\delta >0$, at each time $t \in [m]$ for $m = \poly(n)$, we require our algorithm $\mathcal{A}$ to return an estimate $z_t$ for $g(\bx^{(t)})$ such that 
\[\Pr\left[\left|z_t - g(\bx^{(t)})\right| \leq \eps g(\bx^{(t)})\right] \geq 1-\delta,\]
where query $\bx^{(t)}$ is adaptively chosen, i.e., it may depend on $\calA$'s previous responses $\{z_1,\ldots,z_{t-1}\}$.
\end{definition}

In fact, we may view the adversarial setting as a two-player game between a randomized streaming algorithm $\mathcal{A}$ and an unbounded adversary. By interacting with the streaming algorithm over many time steps, the adversary aims to construct a hard sequence of adaptive updates $\{u_1^*,\ldots, u_m^*\}$ such that any streaming algorithm $\mathcal{A}$ that produces $(\eps, \delta)$-approximate responses $\{z_t\}_{t = 1}^m$ will fail to estimate $g(\bx^{(t^*)})$ with high constant probability at some step $t^* \in [m]$ during the stream. 
More concretely, the game proceeds as follows:

\begin{enumerate}
\item
In each round $t \in [m]$, the adversary selects an update $u_t^*$ to append to the stream. This implicitly defines the underlying dataset $\bx^{(t)}$ at time $t$. 
Importantly, note that $\bx^{(t)}$ may depend on all previous updates $\{u_1^*,\ldots, u_{t-1}^*\}$, as well as the entire transcript of responses $\{z_1,\ldots, z_{t-1}\}$ of the streaming algorithm $\mathcal{A}$.
\item 
$\mathcal{A}$ receives an update $u_t^*$ and updates its internal state. 
\item 
Then, $\mathcal{A}$ returns an estimate $z_t(\bx^{(t)})$ for $g(\bx^{(t)})$ based on the stream observed until time $t$, and proceeds to the next round.
\end{enumerate}

In the adversarially robust insertion-only streaming model, where updates can only increase the coordinates of the underlying dataset, there is a large body of work showing that many fundamental streaming problems admit sublinear space algorithms, c.f.,~\cite{HassidimKMMS20,BravermanHMSSZ21,WoodruffZ21,AjtaiBJSSWZ22,Ben-EliezerJWY22,ChakrabartiGS22,BeimelKMNSS22,JiangPW23,AssadiCGS23,AttiasCSS24}. 
Specifically, \cite{Ben-EliezerJWY22,HassidimKMMS20} introduced frameworks for transforming streaming algorithms for well-behaved functions into adversarially robust streaming algorithms for insertion-only streams, while \cite{BravermanHMSSZ21} showed that adversarial robustness on insertion-only streams can be achieved for a wide class of problems such as clustering, subspace embeddings, linear regression, and graph sparsification using sampling techniques within the popular merge-and-reduce framework. 
In fact, \cite{WoodruffZ21} showed that for central problems such as norm estimation, distinct elements estimation, and heavy-hitters, achieving adversarial robustness on insertion-only streams only requires a small polylogarithmic overhead in space usage, compared to classic insertion-only streaming algorithms. 

On the other hand, for turnstile streams, where updates can either increase or decrease values of the underlying dataset, there are no known adversarially robust algorithms that use space sublinear in the size of the input dataset for these problems. 
Recently for $\ell_0$-estimation, \cite{GribelyukLWYZ24} showed that there can be no adversarially robust integer sketch with $n^{o(1)}$ memory for turnstile streams. 
However, their techniques use fingerprinting as a key technique, and this works precisely when any non-robust sketch uses some degree of sparsity. 
For estimating $\ell_p$ for $p > 0$, there are non-robust sketches that are dense, and the techniques in the above work do not apply.

For $\ell_p$ estimation for $p > 0$, \cite{HardtW13} showed that no linear sketch can approximate the $\ell_p$ norm within even a polynomial multiplicative factor against adaptive queries. 
However, their techniques require that both the input-stream is real-valued and that the linear sketch stored by the algorithm has arbitrary precision, which is unsatisfactory in the context of streaming algorithms. 
We show: 

\begin{restatable}{theorem}{thmadvmain}
\thmlab{thm:adv:main}
Let $B\ge 1$ be any desired accuracy parameter. 
For any integer sketch with $r$ rows, there exists a randomized attack algorithm that, with high probability, finds an integer-valued vector on which the integer sketch fails to output a $B$-approximation to the $L_2$ norm of the query. 
The adaptive attack uses $\poly(r \log n)$ adaptive queries to the integer sketch and has runtime $\poly(r\log n)$ across $r$ rounds of adaptivity and can be implemented as queries in a polynomially-bounded turnstile stream. 
\end{restatable}

In fact, the result holds even for sketch dimension $r = n - o(n)$.
By setting $B$ in the above theorem to be a large enough polynomial in $n$, we give an attack against linear sketches for any $p$-norm for constant $p > 1$:
\begin{corollary}[Adaptive attack for $L_p$ estimation for $p > 1$]
Let $B \geq 1$ be any desired accuracy parameter and $p > 1$. 
For any integer sketch with $r \leq n-o(n)$ rows, there exists a randomized attack algorithm that, with high probability, finds an integer-valued vector on which the integer sketch fails to output a $B$-approximation to the $L_p$ norm of the query. 
The adaptive attack uses $\poly(r\log n)$ adaptive queries to the integer sketch and has runtime $\poly(r \log n)$ across $\O{r \log n}$ rounds of adaptivity and can be implemented as queries in a polynomially-bounded stream. 
\end{corollary}

Finally, our result makes progress toward proving a lower bound for pseudodeterministic streaming algorithms for $\ell_p$ estimation, which was posed as an open problem in \cite{PDstreaming}. 
Since any pseudodeterministic streaming algorithm for $\ell_p$ estimation is also adversarially robust, our attack implies that 
any pseudodeterministic streaming algorithm which uses a finite-precision linear sketch with $r= n - o(n)$ rows can be made to fail after $\poly(r \log n)$ adaptive queries. 

\subsection{Technical Overview}
\seclab{sec:overview}
A continuous Gaussian $\calN(0,\sigma^2\cdot I_n)$ is often used to generate inputs $\bx$ for hard distributions in $\mathbb{R}^n$ for a real-valued linear sketch $\bA\in\mathbb{R}^{r\times n}$, 
as the image $\bA\bx$ is distributed as $\calN(0,\sigma^2\bA\bA^\top)$ due to the rotational invariance of Gaussians.
In other words, the image of a sketch on a Gaussian distribution is itself a Gaussian distribution, and so it can sometimes suffice to analyze the total variation distance between two Gaussian distributions of sufficiently small dimension to show lower bounds for real-valued sketches. 
A natural approach to adapting lower bounds for linear sketches on real-valued inputs would be to ``round'' or ``truncate'' hard distributions at the necessary granularity. 
For example, one could hope to generate a continuous Gaussian and then round to the nearest multiple of a small discretization, e.g., $\frac{1}{\poly(n)}$. 
Unfortunately, this rounding step loses rotational invariance, so analyzing the image of the resulting distribution of $\bA\bx$ becomes significantly more challenging, as the truncation of the input may reveal information about it. 

Another approach might be to use discrete Gaussians in place of continuous Gaussians. 
However, if we draw $\bx\sim\calD(0,\sigma^2 I_n)$, then it is not clear that the image $\bA\bx$ is also distributed like a discrete Gaussian $\calD(0,\sigma^2\bA\bA^\top)$ on the support $\bA\mathbb{Z}^n$. 
In fact, this is not true, and is the starting point of a body of work in lattice theory that analyzes discrete Gaussian distributions by bounding the so-called successive minima of the orthogonal lattice to the row span of $\bA$; this is a technique that we will use in our work as well.


Throughout the next few subsections, we let $\calD_{\calL, \bS}$ denote the discrete Gaussian distribution induced on the lattice $\calL$ with covariance $\Sigma = \bS^T \bS$. 
For notational convenience, we also denote the zero-centered discrete Gaussian distribution with covariance $\Sigma$ by $\calD(0, \Sigma)$ (mirroring standard notation for continuous Gaussian distributions).


\paragraph{Smoothing parameter and successive minima.} 
For a matrix $\bA \in \mathbb{Z}^{r\times n}$ with columns $\bA_i$, we consider the lattice 
\[\calL = \calL(\bA) = \left\{\bA \bz = \sum_{i = 1}^n z_i \bA_i \ : \ \bz \in \mathbb{Z}^n \right\}.\]
We consider the \textit{successive minima} $\lambda_i(\calL)$ of a lattice $\calL$, where $\lambda_i(\calL)$ is defined as the smallest value such that a ball of radius $\lambda_i(\calL)$ centered at the origin contains at least $i$ linearly independent lattice vectors. 
For notational convenience, we denote $\lambda_{\textrm{rank}(\bA)}(\calL(\bA))$ by $\lambda_{\max}(\calL(\bA))$.

In order to argue that the total variation distance between $\bA \bx$ and $\by$ is small for $\bx \sim \calD_{\mathbb{Z}^n, \bS}$ and $\by \sim \calD_{\bA\mathbb{Z}^n, \bS \bA^{\top}}$, we draw techniques from the literature on lattice theory. 
In prior works \cite{AggarwalR16, AgrawalGHS13}, a common technique is to upper bound the so-called \textit{smoothing parameter} of the lattice. 
This reduces to the problem of upper-bounding the successive minima of the \textit{orthogonal lattice} of $\calL(\bA)$, denoted by $\calL^{\perp}(\bA)$, which is the lattice formed by all integer vectors $\bx \in \mathbb{Z}^n$ that are orthogonal to all $r$ rows of $\bA$. 
So, the goal is to find a basis of $n-r$ linearly independent vectors such that their maximum length is as small as possible; concretely, this translates to bounding the last successive minimum $\lambda_{n-r}(\calL^{\perp}(\bA))$ of the orthogonal lattice. 
However, it is not true in general that an arbitrary sketching matrix $\bA \in \mathbb{Z}^{r \times n}$ induces a lattice $\calL^{\perp}(\bA)$ such that the orthogonal lattice $\calL(\bA)$ has a basis of short integer vectors (i.e., length $\leq \poly(n)$). 
Thus, in the next section, we discuss how we circumvent this obstacle by \textit{pre-processing} the sketching matrix.

\paragraph{Pre-processing to upper bound the successive minima.} 
Suppose $\bA \in \mathbb{Z}^{r \times n}$ is any sketching matrix with $r < n$ and integer entries bounded by $M \leq \poly(n)$. 
Let $\calL(\bA)$ denote the lattice induced by $\bA$, and let $\calL^{\perp}(\bA)$ be the orthogonal lattice. 
Motivated by the discussion above, we want to come up with a \textit{pre-processing} of $\bA$ so that the following properties hold: 
\begin{itemize}
\item 
After pre-processing $\bA$, this matrix satisfies $\lambda_{\max}(\calL^{\perp}(\bA)) \leq \poly(n)$. 
\item 
The resulting matrix $\bA \in \mathbb{Z}^{m \times n}$ has $m \leq n - \Omega(n)$ rows.
\item 
This pre-processing can be applied without loss of generality in our lower bound.
\end{itemize}
Note that we do not require the pre-processing to be constructed efficiently, as we merely want to ensure that there exists a pre-processing procedure that the algorithm can execute on sketch $\bA$ without loss of generality. 

Intuitively, we want to iteratively add rows to the rowspan of $\bA$ such that, after enough vectors have been added, there exists a basis of \textit{short} integer vectors for $\calL^{\perp}(\bA)$ such that the longest basis vector has length at most $\poly(n)$. 
For this, we draw inspiration from Siegel's lemma (\lemref{lem:siegel}), which guarantees that there exists a non-zero vector $\bx \in \mathbb{Z}^n$ of integers bounded by $(nM)^{r/(n-r)}$ such that $\bA \bx = 0^n$. 
Using this fact, we can now think of pre-processing $\bA$ iteratively, as follows: let $\bA_1 = \bA$, and at each step $t$, we find a vector $\bx_{t}$ such that $\bA_t \bx_{t} = 0^n$ and $\bx_{t}$ has entries bounded by $(nM)^{(r+t)/(n-r-t)}$, and add $\bx_{t}$ to the matrix $\bA_t$. 
We repeat this for $T= 0.49n - r$ iterations to obtain a set $\{\bx_{1},\ldots, \bx_{T} \}$. 
Then, at the end of this process, $\bA_{T}$ will have $r + T = 0.49n$ rows, so the final vector $\bx_T$ (as well as each intermediate vector $\bx_t$) will have all entries bounded by $(nM)^{0.49/0.51}\leq nM$. 
In particular, all vectors $\bx_t$ certainly have norm at most $(nM)^2$. 
We then continue to apply Siegel's lemma $0.51n$ more times to generate the remaining set of $0.51n$ linearly-independent rational vectors $\{\by_1,\ldots,\by_{0.51n} \}$. Finally, we add the rows $\{\by_1,\ldots,\by_{0.51n}\}$ to the \textit{original} matrix $\bA$ (which does not contain vectors $\bx_t$); call this new matrix $\bA'$. 
Observe that by construction, we immediately have that $\{\bx_1,\ldots, \bx_T\}$ is an integer basis for the orthogonal lattice $\calL^{\perp}(\bA')$ such that $\lambda_{\max}(\calL^{\perp}(\bA')) \leq \poly(n)$.

We note that since we only add rows to $\bA$ to form the new sketching matrix $\bA'$, we can without loss of generality assume that the actual sketching matrix is $\bA'$ instead of $\bA$, as it only makes the sketch more powerful. 

\paragraph{A better pre-processing.}
However, the above approach requires adding $\Theta(n)$ rows to the matrix $\bA$, which will stop us from achieving optimal sketching lower bounds for problems for which we want to prove a lower bound sublinear in $n$, e.g., $L_p$ estimation for $p \ge 2$. 
To get a better pre-processing, we consider the following alternative approach, which only requires adding $\O{r}$ rows to $\bA$. 
We first show that for any $\bA \in \mathbb{Z}^{r \times n}$ with entries bounded by $M \le\poly(n)$, there exist at least $n - 4r$ linearly independent integer vectors that are in $\calL^{\perp}(\bA)$ and each of them has length at most $\sqrt{nM^2}$. 
To prove this, suppose we have found $t$ such integer vectors, and let $\bB \in \mathbb{R}^{(n - t) \times n}$ denote the matrix whose rows form a basis of the orthogonal complement to the span of these $t$ vectors. 
We then choose the next integer vector by the following probabilistic argument:  we randomly pick $s=\Theta(M^{1.5r})$ random integer vectors $\bv^1, \bv^2, \dots, \bv^s$ with entries in $\{0, 1, \ldots, M - 1\}$. 
Note that since each coordinate of $\bA \bv^i$ is in the range of $\{-nM^2, \ldots, nM^2\}$, this means that there are at most $M^{3r}$ possible values of $\bA \bv^i$. 
Since we randomly sample $\Theta(M^{1.5r})$ vectors $\bv^i$, by the birthday paradox, we have that with high constant probability there exist indices $1 \le i < j \le s$, such that $\bA \bv^i = \bA \bv^j$. 

On the other hand, by writing $\bB$ in reduced row echelon form, we can see that for every $\bx^{ij} = \bv^i - \bv^j$, $\Pr [\bB \bx^{ij} = \mathbf{0}] \le \left(\frac{1}{M}\right)^{n - t}$. This implies when $t \le n - 4r$, after taking a union bound over all pairs of vectors $\bv^i,\bv^j$, we have with high constant probability, 
\begin{itemize}
\item There exists $1 \le i < j \le s$, such that $\bA \bv^i = \bA \bv^j$. 
\item For every $1 \le i' < j' \le s$, we have $\bB \bv^i \ne \bB \bv^j$.
\end{itemize}
From the above two conditions, we immediately have $\bv^i - \bv^j$ is in the kernel $\mathrm{ker} (\bA)$ and is linearly independent to the previous $t$ vectors (as we have $\bB(\bv^i - \bv^j) \ne 0$). This means that we can add the vector $\bv^i - \bv^j$ to the underlying set and repeat the process until $t = n - 4r$.

Now, after adding $n - 4r$ integer vectors in $\calL^{\perp}(\bA)$ via the process described above, we consider a similar strategy to our earlier pre-processing: apply Siegel's lemma iteratively to generate $3r$ integer vectors that are orthogonal to both the row span of $\bA$ and the $n - 4r$ integer vectors in $\calL^{\perp}(\bA)$. We then add them to the rows of $\bA$ and form a new sketch $\bA' \in \mathbb{Z}^{4r \times n}$. From the above discussion, it immediately follows that $\lambda_{\max}(\calL^{\perp}(\bA'))\leq \sqrt{nM^2}$. Since we consider $M = \poly(n)$, we have that $\lambda_{\max}(\calL^{\perp}(\bA)) \leq \poly(n)$, as desired.

\paragraph{Cell lemma.}\seclab{cell-lemma}
Suppose $\bA \in \mathbb{Z}^{m \times n}$ is the sketching matrix that we obtain after applying the pre-processing described above, i.e., $\lambda_{\max}(\calL^{\perp}(\bA)) \leq \poly(n)$. 
Fix the fundamental parallelepiped $\mathcal{F}$ of the lattice $\calL(\bA)$, which is defined as $\mathcal{F} = \{\bx \in \mathbb{R}^{r} \ | \ x_i \in [0, \ell_i)\}$, where the unit distance $\ell_i$ in the $i$-th coordinate is defined as the minimal non-zero value of $|(\bA \bx)_i|$ over $\bx \in \mathbb{Z}^n$. 
Observe that $\calF$ can be used to partition $\mathbb{R}^r$ into cells such that each lattice point $\bA \bz$ for an integer vector $\bz\in\mathbb{R}^n$ is associated to one \textit{cell} of $\calL(\bA)$. 
Moreover, without loss of generality, we may assume that $\bA$ has orthonormal rows by multiplying by a change-of-basis matrix in the post-processing step. Importantly, the fundamental parallelepiped induced by $\bA$ is still well-defined after orthonormalizing the rows: to see this, let $\bU^{\top}$ have orthonormal rows with the same rowspan $R(\bA)$ as $\bA$. 
Then, we can write the projection onto the rowspan $R(\bA)$ as $\bA^{\top}(\bA \bA^{\top})^{-1} \bA = \bU \bU^{\top}$. 
Since our pre-processing ensures that $\bA$ has an integer basis for $R(\bA)$, we have that $\bU \bU^{\top}$ has only rational entries, and $\bU \bU^{\top} \bx$ is lower bounded for any non-zero integer vector $\bx$. 
This means that the fundamental parallelepiped of $\bU^{\top}$ is well-defined. 

At this point, we are in good shape to obtain a point-wise bound between the probability mass functions of $\bA \bx$ and $\by$, where $\bx \sim \calD_{\mathbb{Z}^n, \bS}$ and $\by \sim \calD_{\bA \mathbb{Z}^n, \bS \bA^{\top}}$. Let $\rho_1$ be the probability mass function of $\bA \bx$ for $\bx \sim \calD_{\mathbb{Z}^n, \bS}$, and let $\rho_2$ be the probability mass function of $\by \sim \calD_{\bA \mathbb{Z}^{n}, \bS \bA^{\top}}$. 
Furthermore, suppose the covariance matrix $\Sigma = \bS^T \bS$ satisfies $\sigma_{n}(\Sigma) \geq \alpha $, where $\alpha = \poly(n) > \lambda_{n-m}(\calL(\bA^\perp))^2\cdot\frac{\ln(2n(1+1/\eps))}{\pi}$, where the second inequality holds as a result of our pre-processing procedure. 
Then, since the successive minima (and hence the smoothing parameter) are bounded, we can apply an existing result from lattice theory, e.g., \lemref{discrete_gaussian_tvd} (Lemma 4 of \cite{ AgrawalGHS13}), to conclude that 
\begin{align*}
    1-\frac{1}{\poly(n)} \leq \frac{\rho_1(\bx)}{\rho_2(\bx)} \leq  1 + \frac{1}{\poly(n)}.
\end{align*}
Notice that while the discrete Gaussian distribution and continuous Gaussian distribution have almost the same probability mass for each integer point $\bx \in \mathbb{Z}^n$, these distributions have large total variation distance inside of each unit cell of $\calL(\bA)$. 
In particular, there is no support inside each unit cell for the discrete Gaussian distribution. 
Thus, in order to ``lift'' lower bounds that rely on real-valued Gaussian inputs and obtain lower bounds for integer inputs, we need one more crucial observation: we can assume that for any input vector $\bx \in \mathbb{Z}^n$, the sketching algorithm actually observes $\bA \bx + \bfEta$ for $\bfEta$ drawn uniformly at random from the fundamental parallelepiped $\calF$ of $\bA$. 
This is without loss of generality, since the sketching algorithm can simply round $\bA \bx + \bfEta$ in $\calL(\bA)$ to recover $\bA \bx$. Given this, we show the following: let $p$ be the probability density function of $\bA \bx + \bfEta$ for $\bx \sim \calD_{\mathbb{Z}^n, \bS}$ and $\bfEta$ drawn uniformly from $\calF$. Let $q$ be the probability density function of $\calN(0, \bA\bS^{\top} \bS \bA^{\top})$. 
With high probability, we have 
\begin{align*}
    1-\frac{1}{\poly(n)} \leq \frac{p(\bx)}{q(\bx)} \leq  1 + \frac{1}{\poly(n)}.
\end{align*}
Importantly, this lemma allows us to ``simulate'' a continuous Gaussian distribution using $\bA \bx + \bfEta$, which is crucial to our lifting framework in \thmref{thm:lifting}. 
Our integer sketching lower bounds then follow by ``simulating'' the previously known attacks over the reals, which originally required continuous Gaussian queries.

Note that while we achieve point-wise closeness in the probability mass function distribution, we do not guarantee point-wise closeness here, i.e., we cannot sample a Gaussian $\bg$, round it in $\mathbb{Z}^n$, and look at $\bA \bg$, and expect this to be close in total variation distance to the rounding of $\bA \bg'$ for an independent Gaussian $\bg'$. 

\paragraph{Adversarial attack for $L_2$ estimation.}
In the adversarial streaming setting, the primary goal of the adaptive adversary is to iteratively learn the \textit{rowspace} $R(\bA)$ of the sketching matrix $\bA$ that is used by the $L_2$ estimation algorithm. In particular, given $R(\bA)$, the adversary can easily come up with a hard query distribution: with probability $1/2$, query (1) $\bx \in \mathrm{ker} (\bA)$, and otherwise query (2) $\bx = 0^n$. Since $\calA$ can only observe $\bA \bx$, the algorithm will necessarily fail to distinguish these two cases, and since the norm of $\bx$ is non-zero in case (1) and zero in case (2), $\calA$ will fail to provide a constant-factor approximation to $\|\bx\|_2^2$. However, since $\calA$ is a randomized streaming algorithm, the adversary cannot hope to learn $R(\bA)$ exactly; rather, it will try to learn $r - \O{1}$ orthonormal vectors that are \textit{approximately} in $R(\bA)$. 

Before we describe our adaptive attack against streaming algorithms that use \textit{integer} linear sketches for $L_2$ estimation, we give an overview of the real-valued attack against linear sketches that was given in \cite{HardtW13}; this is the attack that we will aim to ``lift'' to obtain an attack in the (discrete) streaming setting. Without loss of generality, we can assume that $\bA$ has orthonormal rows, since $\calA$ can always perform a change-of-basis. Furthermore, we may assume that the algorithm actually observes $\bA^{\top} \bA \bx = P_{\bA} \bx$ for each query $\bx \in \mathbb{R}^n$, since $\calA$ can recover $\bA \bx$ by multiplying by $\bA$ in the post-processing step. For notational convenience, we identify $\bA$ with the rowspace $R(\bA)$, and let $f:\bA \rightarrow \{0,1\}$ denote a linear sketch that takes any input $\bx \in \mathbb{R}^n$ and uses $r \times n$ sketching matrix $\bA$.

Then, the attack proceeds as follows: in the first round, the adversary samples queries from the continuous Gaussian distribution $\bx \sim \calN(0, \sigma^2 I_n)$. Note that since $\calN(0, \sigma^2 I_n)$ is spherically symmetric, $\calA$ will have to estimate $\|\bx\|_2^2$ solely based on the norm $\|P_{\bA}\bx \|_2^2$. However, since $R(\bA)$ is a proper $r$-dimensional subspace of $\mathbb{R}^n$, if query $\bx$ is more correlated with $R(\bA)$, then the algorithm will observe that $\bx$ has a slightly-increased projection onto $R(\bA)$, and may be ``tricked'' into thinking that $\bx$ has much larger norm than it does. Formally, this intuition is captured by the Conditional Expectation Lemma of \cite{HardtW13}: the authors prove that there exists $\sigma^2$ such that the difference $\mathbb{E}\left[\|P_{\bA}\bx\|_2^2\,\mid\,f(\bx) = 1 \right] - \mathbb{E}\left[ \|P_{\bA}\bx\|_2^2 \right] = \Delta/2$ is large. Consequently, it follows that the adversary can actually find $\ell = \poly(n)$ vectors $\bx_1,\ldots, \bx_{\ell}$ that have a slightly-increased expected projection onto $R(\bA)$. But, these $\bx_i$'s will only be slightly correlated with $R(\bA)$: to obtain a vector $\bx^*$ that is \textit{very} correlated with $R(\bA)$, the adversary arranges vectors $\bx_i$ such that $f(\bx_i) = 1$ as rows of a matrix $\bB$, and computes the top right singular vector $\bx^*$ of $\bB$. Indeed, the authors prove that $\|P_{\bA} \bx^* \|_2^2 \geq 1-\frac{1}{\poly(n)}$, i.e. the vector $\bx^*$ that is recovered by the adversary is very close to being contained in the rowspace.

At this point in the attack, the adversary has found a single vector $\bx^*$ that is highly correlated with $R(\bA)$. However, the algorithm $\calA$ may use different subsets of the rows of $\bA$ at different times to form its estimate, so this is not sufficient to ensure that $\calA$ fails with constant probability. To make progress in learning $R(\bA)$, in the next round the adversary makes queries sampled from the spherical Gaussian distribution inside of the subspace that is \textit{orthogonal} to the vector $\bx^*$ that it already found. This way, the adversary will iteratively find a highly-correlated vector to $R(\bA)$, which is orthogonal to all vectors that it found so far, and thus will eventually recover an $r$-dimensional space $V$ that is approximately-contained in $R(\bA)$.

Let us return to the problem of constructing an adaptive attack against linear sketches in the streaming setting, where $\bA \in \mathbb{Z}^{r \times n}$ and queries $\bx \in \mathbb{Z}^n$. There are several challenges in directly adapting the approach of \cite{HardtW13} for the streaming setting. For instance, if we na\"{i}vely modify the above attack strategy to instead sample from a discrete Gaussian $\bx \sim \calD(0, \sigma^2 I_n)$ at each step, then since queries $\bx$ are no longer spherically symmetric, it is no longer immediate that the response of $\calA$ will only depend on $\|P_{\bA} \bx \|_2^2$ and not on the direction of $\bx$. Furthermore, even if we are able to show some version of a conditional expectation lemma using our cell lemma, it is not immediately clear how to argue that the cell lemma will still hold in future rounds of the attack. 

To circumvent these challenges, we proceed as follows: we may assume that the streaming algorithm first \textit{pre-processes} the sketching matrix $\bA \in \mathbb{Z}^{r \times n}$ to obtain a new sketching matrix $\bA' \in \mathbb{Z}^{m \times n}$ with $m \leq 4r$ rows and $\lambda_{\max}(\calL(\bA'^{\perp})) \leq \poly(n)$. Note that this is without loss of generality, since our pre-processing only adds rows to $\bA$ and thus can only give $\calA$ more information about $\bx$. For notational convenience, we now write $\bA$ to denote this pre-processed matrix. Now, to carry out the attack, in each round $t$, the adversary will sample queries $\bx \sim \calD(0, \bSigma_t)$ such that $\bSigma_t$ is precisely the covariance matrix that would have been used by the real-valued attack \cite{HardtW13} in round $t$. In fact, by scaling $\bSigma_t$ such that $\sigma_n(\bSigma_t) \geq \alpha$ for $\alpha = \poly(n)$ much larger than $\lambda_{\max}(\calL^{\perp}(\bA))$, we can actually obtain a discrete version of the Conditional Expectation Lemma: namely, we prove that in each round $t$, there exists a scaling $c$ of the variance $\bSigma_t$ such that $\Delta = \mathbb{E}_{\bx \sim \calD(0, c \cdot \bSigma_{t})}\left[\|P_{\bA}\bx\|_2^2\,\mid\,f(\bx) = 1 \right] - \mathbb{E}_{\bx \sim \calD(0, c \cdot \bSigma_{t})}\left[\|P_{\bA}\bx \|_2^2 \right]$ is non-trivially large, for queries $\bx$ drawn from the discrete Gaussian distribution with variance $c \cdot \Sigma_t$. Interestingly, we no longer rely on showing that $\|P_{\bA}\bx\|_2^2$ is (approximately) a sufficient statistic, which was a crucial step in the argument of \cite{HardtW13}. Instead, we prove this claim by contradiction: in particular, suppose there exists an integer linear sketch $f:\bA \rightarrow \{0,1\}$ such that $f$ does not satisfy the Discrete Conditional Expectation Lemma. Then, we consider the real-valued linear sketch $f':\bA \rightarrow \{0,1\}$ which, for any continuous Gaussian query $\bg \sim \calN(0, \bSigma_t)$, first \textit{rounds} $\bA \bg$ to the nearest lattice point $\bA \bz$ for some $\bz \in \mathbb{Z}^n$, and then returns $f(\bz)$. By a similar argument to the cell lemma, we can show that the total variation distance between the distributions of $\bA \bz$ and $\bA \bx$ is at most $\frac{1}{\poly(n)}$, where $\bA \bz$ is the point obtained by rounding $\bA \bg$ for $\bg \sim \calN(0, \bSigma_t)$ and $\bx \sim \calD(0, \bSigma_t)$.  Thus, by conditioning on $f'(\bg) = 1$ or equivalently $f(\bz) = 1$, the probability distribution of $\|P_{\bA}\bg \|_2^2$ only changes by $\frac{1}{\poly(n)}$. Moreover, conditioning on the event $\|\bg \|_2^2 \leq \sigma^2 \log(n)$ we have that $\mathbb{E}[\|P_{\bA}\bg\|_2^2]$ only changes by $\frac{\sigma^2}{\poly(n)}$ for the distribution of $P_{\bA} \bz$. Putting it all together, this implies that the gap $\Delta$ in the (real-valued) Conditional Expectation Lemma is always small for $f'$, which gives a contradiction.

Finally, with our Discrete Conditional Expectation Lemma, we can now argue that in the first round, the adversary can find many discrete vectors $\bx_1,\ldots,\bx_{\ell}$ that have a slightly-increased expected projection onto $R(\bA)$, and obtain a strongly-correlated vector $\bx^*$ using the same technique as the real-valued attack. 

Importantly, unlike our other one-round lower bounds, the adaptive attack proceeds over many rounds, and so we need to be careful to ensure that we can apply this Conditional Expectation Lemma in future rounds: by \lemref{discrete_gaussian_tvd}, we need $\sigma_n(\Sigma_t) \geq \lambda_{\max}(\calL^{\perp}(\bA))^2 \cdot \frac{\ln(2n(1+1/\eps))}{\pi}$ to hold in each round $t$ of the attack. 
Fortunately, we constructed $\calD(0, \Sigma_t)$ such that $\Sigma_t$ is precisely the covariance matrix used by the real-valued attack, and one can check that $\sigma_{n}(\Sigma_t) \geq \poly(n)$ always holds for the real-valued attack. 
Therefore, we are able to argue that the Discrete Conditional Expectation Lemma holds for each round of our adaptive attack. 
We can then show that over the course of the attack, each vector which is identified by the adaptive attack has a large amount of correlation with the rowspan of the sketching matrix $\bA\in\mathbb{Z}^{r\times n}$ and is near-orthogonal to the previously identified vectors, thereby progressively growing the dimension of the subspace of the identified vectors. 
Hence after $r$ rounds, the subsequent orthogonal subspace is also nearly orthogonal to the sketch matrix $\bA$, and thus $\calA$ cannot correctly answer subsequent queries with high probability. 

\section{Preliminaries}
We use the notation $[n]$ to denote the integer set $\{1,2,\ldots,n\}$ for a positive integer $n>0$. 
We use $\poly(n)$ to denote any fixed polynomial in $n$ whose degree can be determined by adjusting the appropriate constants. 
We say that an event occurs with high probability if it occurs with probability $1-\frac{1}{\poly(n)}$. 
We generally use standard typeface to denote scalar variables and bold typeface to denote vectors or matrices. 

For a vector $\bv$, we define the entry-wise $L_p$ norm by $\|\bv\|_p=\left(v_1^p+\ldots+v_n^p\right)^{1/p}$. 
For a matrix $\bA\in\mathbb{R}^{n\times d}$, we define the operator norm $\|\bA\|_{op}=\max_{\bx\in\mathbb{R}^d}\frac{\|\bA\bx\|_2}{\|\bx\|_2}$. 
If $n\ge d$, then the matrix $\bA$ has $d$ singular values $\sigma_1(\bA),\ldots,\sigma_d(\bA)$ and we define the Schatten-$p$ norm of $\bA$ by $\left(\sigma_1(\bA)^p+\ldots+\sigma_d(\bA)^p\right)^{1/p}$. 
The Schatten-$p$ norm is defined analogously when $n<d$ for the $n$ singular values of $\bA$. 
Note that the operator norm of $\bA$ is equivalent to $\max_i\sigma_i(\bA)$, and thus also equivalent to the Schatten-$p$ norm for $p=\infty$. 
Similarly, the Schatten-$2$ norm corresponds to the Frobenius norm, denoted $\|\bA\|_F=\left(\sum_{i,j}A_{i,j}^2\right)^{1/2}$.

Let $P$ and $Q$ be probability mass functions. 
We define the \emph{total variation distance} between $P$ and $Q$ by 
\[\TVD(P,Q)=\frac{1}{2}\cdot\|P-Q\|_1.\]
For a convex function $\phi:\mathbb{R} \rightarrow \mathbb{R}$, we define the $\phi$-divergence 

$$D_{\phi}(P || Q) = \sum_{x} Q(x) \cdot \phi\left(\frac{P(x)}{Q(x)} \right)$$

For $\phi(x) = |x-1|$, $D_{\phi}(P || Q)$ is precisely the total variation distance between distributions $P$ and $Q$. Also, the \textit{$\chi^2$-divergence} between $P$ and $Q$ is defined by $D_{\phi}(P || Q)$ for $\phi(x) = x^2 - 1$ or $\phi(x) = (x-1)^2$. These two choices of $\phi$ result in the same value for $D_{\phi}(P || Q)$. However, note that in general, $D_{\phi}$ is not a distance since it is not necessarily symmetric. 

We have the following well-known relationship between the total variation distance and the $\chi^2$ divergence. 
\begin{proposition}
    $\TVD(P, Q) \leq \sqrt{\chi^2(P || Q)}$
\end{proposition}

\subsection{Continuous and Discrete Gaussians}
We first define our notation for continuous and discrete Gaussians. 

\begin{definition}[Continuous Gaussian]
For any $s >0$, the spherical Gaussian function centered at $\mu \in \mathbb{R}^n$ is given by $\rho_s(x) = \exp(-||\bx-\mu||^2 / 2s^2)$ for all $\bx \in \mathbb{R}^n$. Similarly, for any rank-$n$ $S \in \mathbb{R}^{m \times n}$, the ellipsoidal Gaussian function with mean $\mu \in \mathbb{R}^n$ and covariance matrix $\bSigma = 2 \cdot S^{\top} S$ is given by 
$$\rho_{S}(\bx) = \exp(- (\bx-\mu)^{\top}\bSigma^{-1} (\bx-\mu)),$$
for any $x \in \mathbb{R}^n$. 
\end{definition}



\begin{definition}[Discrete Gaussian]
\deflab{def:discrete:Gaussian}
For a lattice $\calL\subset\mathbb{R}^n$, $\textbf{c} \in \mathbb{R}^n$, and $\bx\in\calL + \textbf{c}$, we define the probability mass function for the Discrete Gaussian distribution with mean $\bmu\in\mathbb{R}^n$, covariance matrix $\bSigma = 2 \cdot (S^{\top}S) \in\mathbb{R}^{n\times n}$, and support $\calL + \textbf{c} $ by 
\[\mathcal{D}_{\calL+\textbf{c}, S}(\bx)=\frac{\rho_S(\bx)}{\rho_S(\calL + \textbf{c})},\]
where $\rho_S(A) = \sum_{\bx \in A} \rho_S(\bx)$ for any set $A$. In future sections, we will denote the discrete Gaussian distribution over $\calL = \mathbb{Z}^n$ with mean $0^n$ and covariance $\bSigma = \sigma^2 \cdot I_n$ by $\mathcal{D}(0, \sigma^2 I_n)$.
\end{definition}
We use the following fact about the normalization constant for the probability mass function of discrete Gaussians. 
\begin{fact}(Normalization constant)\factlab{fact:normalization}
    For all $\sigma \in \mathbb{R}$ and $\sigma > 0$, $$\max\{\sqrt{2\pi \sigma^2}, 1 \} \leq \sum_{z \in \mathbb{Z}} e^{-z^2/2\sigma^2} \leq \sqrt{2\pi \sigma^2} + 1.$$
\end{fact}
We also recall the following relationship relating the value of the probability mass function of a discrete Gaussian at a certain point $\bv$ with the value of the probability density function of a continuous Gaussian at $\bv$. 
The proof is completely standard, and only included for the sake of completeness. 
\begin{lemma} \lemlab{lem:pmf}
Let $\sigma^2 > n^{C+1}$ be large enough. Let $p$ denote the probability mass function for $\calD(0, \sigma^2 I_n)$, and let $q$ be the probability mass function for the output of sampling $x \sim \calN(0, \sigma^2 I_n)$ and then rounding each coordinate to the nearest integer. Then for any $\bv\in\mathbb{Z}^n$,  
$$\frac{p(\bv)}{q(\bv)} \in \left[1 - \frac{1}{n^C}, 1 + \frac{1}{n^C} \right].$$
\end{lemma}

\begin{proof}
    First, consider the probability mass function for the distribution $q$:

    $$q(\bv) = \frac{1}{(2\pi \sigma^2)^{n/2}}\int_{v_1}^{v_1 + 1} \cdots \int_{v_n}^{v_n + 1} e^{-\|x\|_2^2/2\sigma^2} dx_1 \ \ldots\ dx_n $$
    
    Observe that for each coordinate $i$, $e^{-(x_i+1)^2/2\sigma^2} \in \left[e^{-x_i^2/2\sigma^2} \cdot e^{-1/n^{C+1}} , e^{-x_i^2/2\sigma^2} \cdot e^{1/n^{C+1}} \right] $. Then, by approximating the integral above, it follows that 

    $$q(\bv) \in \left[\frac{e^{-\|\bv\|_2^2/2\sigma^2}}{(2\pi\sigma^2)^{n/2}} \cdot e^{-1/n^C},  \frac{e^{-\|\bv\|_2^2/2\sigma^2}}{(2\pi\sigma^2)^{n/2}} \cdot e^{1/n^C} \right]. $$
    
    Now, recall that for any $\bv \in \mathbb{Z}^n$, the probability mass function of $\calD(0, \sigma^2 I_n)$ is given by 

    $$p(\bv) = \frac{e^{-\|\bv\|_2^2/2\sigma^2}}{\left(\sum_{z \in \mathbb{Z}}e^{-z^2/2\sigma^2} \right)^n}.$$

    By \factref{fact:normalization}, we have $p(\bv) \in \left[\frac{e^{-\|\bv\|_2^2/2\sigma^2}}{(\sqrt{2\pi\sigma^2} + 1)^{n}}, \frac{e^{-\|\bv\|_2^2/2\sigma^2}}{(\sqrt{2\pi\sigma^2})^{n}} \right]$. Finally, the claim follows by combining the two inequalities above (with the assumption that $\sigma^2 > n^{C+1}$).
\end{proof}
We recall the definition of sub-Gaussian and sub-exponential norms:
\begin{definition}[Sub-Gaussian norm]
We define the sub-Gaussian norm of a random variable $X$ by
\[\|X\|_{\psi_2}=\inf\{t>0\,\mid\,\Ex{\exp(X^2/t^2)}\le 2\}.\]
A random vector $\bX\in\mathbb{R}^n$ is sub-Gaussian if all one-dimensional marginals $\langle\bX,\bv\rangle$ are sub-Gaussian random variables for all $\bv\in\mathbb{R}^n$. 
Then we define the sub-Gaussian norm $\|\bX\|_{\psi_2}:=\sup_{\bv\in S^{n-1}}\|\langle\bX,\bv\rangle\|_{\psi_2}$, for the unit sphere $S^{n-1}$. 
\end{definition}

\begin{definition}[Sub-exponential norm]
We define the sub-exponential norm of a random variable $X$ by
\[\|X\|_{\psi_1}=\inf\{t>0\,\mid\,\Ex{\exp(|X|/t)}\le 2\}.\]
\end{definition}
We have the following relationship between a sub-Gaussian random variable and its square:
\begin{lemma}[Lemma 2.7.6 and Exercise 2.7.10 in \cite{Vershynin18}]
\lemlab{lem:sub:props}
A random variable $X$ is sub-Gaussian if and only if $X^2$ is sub-exponential. 
In particular,
\[\|X\|_{\psi_2}^2=\|X^2\|_{\psi_1}.\]
Also, for a sub-exponential random variable $X$, there exists an absolute constant $C>0$ such that
\[\|X-\Ex{X}\|_{\psi_1}\le C\|X\|_{\psi_1}.\]
\end{lemma}
We recall the following concentration inequality for sub-exponential random variables:
\begin{lemma}[Bernstein's inequality for sub-exponential random variables]
\lemlab{lem:bernstein:subexp}
Let $X_1,\ldots,X_n$ be independent, mean zero, sub-exponential random variables and let $K=\max_i\|X_i\|_{\psi_1}$. 
Then there exists an absolute constant $c>0$ such that for all $t\ge0$,
\[\PPr{\left\lvert\frac{1}{n}\sum_{i=1}^n X_i\right\rvert>t}\le2\exp\left(-c\cdot\min\left(\frac{t^2}{K^2},\frac{t}{K}\right)\cdot n\right).\]
\end{lemma}


Although it is well-known that continuous Gaussian random variables are sub-Gaussian, we recall that discrete Gaussian random variables are also sub-Gaussian:
\begin{lemma}[Corollary 17 in \cite{CanonneKS20}]
\lemlab{lem:disc:subgauss}
Let $X\sim\calD(0,\sigma^2)$. 
Then there is an absolute constant $C>0$ such that $\|X\|_{\psi_2}\le C\sigma$. 
\end{lemma}
We next observe that the singular values of a random rectangular discrete Gaussian matrix are well-concentrated. 
The proof generalizes Theorem 4.6.1 in \cite{Vershynin18} and is straightforward:
\begin{lemma}
\lemlab{lem:two:sided:oper}
Let $\bA\in\mathbb{R}^{m\times n}$ with entries that are independently drawn from $\calD(0,N^2)$. 
Then there exist absolute constants $C_1,C_2>0$, such that for any $t\ge0$ and any $i\in\min(n,m)$,
\[C_1N\sqrt{m}-C_2N(\sqrt{n}+t)\le\sigma_i\le C_1N\sqrt{m}+C_2N(\sqrt{n}+t),\]
with probability at least $1-2\exp(-t^2)$. 
In particular, the above inequality implies
\[C_1N\sqrt{m}-C_2N(\sqrt{n}+t)\le\|\bA\|_{op}\le C_1N\sqrt{m}+C_2N(\sqrt{n}+t).\]
\end{lemma}
\begin{proof}
We will show
\[\left\|\frac{1}{m}\bA^\top\bA-N^2\cdot I_n\right\|_{op}\le C_1^2N^2\cdot\max(\tau,\tau^2),\]
where $\tau=C_2\cdot\left(\sqrt{\frac{n}{m}}+\frac{t}{\sqrt{m}}\right)$. 

Consider a unit vector $\bx\in\mathbb{R}^n$. 
For each $i\in[m]$, let $X_i=\langle A_i,\bx\rangle$. 
Then we have
\[\|\bA\bx\|_2^2=\sum_{i=1}^m\langle A_i,\bx\rangle^2=\sum_{i=1}^m X_i^2.\]
Since each entry of $A_i$ is independently drawn from $\calD(0,N^2)$, then we have $\Ex{X_i^2}=\gamma_1^2\cdot N^2$ for some absolute constant $\gamma_1>0$, as well as $\|X_i\|_{\psi_2}\le\gamma_2\cdot N$ by \lemref{lem:disc:subgauss}. 
Since $X_i$ is sub-Gaussian, then by \lemref{lem:sub:props}, we have that 
\[\|X_i^2-\gamma_1^2\cdot N^2\|_{\psi_1}\le\gamma_3\cdot N^2,\]
for some absolute constant $\gamma_3>0$. 
By Bernstein's inequality, c.f., \lemref{lem:bernstein:subexp}, there exists an absolute constant $\gamma_4>0$ such that
\begin{align*}
\PPr{\left\lvert\frac{1}{m}\|\bA\bx\|_2^2-\gamma_1^2\cdot N^2\right\rvert\ge\frac{t}{2}}&=\PPr{\left\lvert\frac{1}{m}\sum_{i=1}^m X_i^2-\gamma_1^2\cdot N^2\right\rvert\ge\frac{t}{2}}.
\end{align*}
In particular, if we set $t=N^2\cdot\max(\tau,\tau^2)$, then we have
\begin{align*}
\PPr{\left\lvert\frac{1}{m}\|\bA\bx\|_2^2-\gamma_1^2\cdot N^2\right\rvert\ge\frac{t}{2}}&\le 2\exp\left(-\gamma_4\cdot\min\left(\frac{t^2}{N^4},\frac{t}{N^2}\right)m\right)\\
&=2\exp\left(-\gamma_4\cdot\tau^2 m\right)\\
&\le2\exp\left(-\gamma_4\cdot C_2^2(n+t^2)\right),
\end{align*}
by the definition of $\tau$. 

Consider a $\frac{1}{4}$-net $\calN$ of the unit sphere $S^{n-1}$, so that $|\calN|\le\exp(\O{n})$. 
Because $\calN$ is a $\frac{1}{4}$-net of $S^{n-1}$, then
\[\left\|\frac{1}{m}\bA^\top\bA-N^2\cdot I_n\right\|_{op}\le 2\cdot\max_{\bx\in\calN}\left\lvert\left\langle\left(\frac{1}{m}\bA^\top\bA-N^2\cdot I_n\right)\bx,\bx\right\rangle\right\rvert=2\max_{\bx\in\calN}\left\lvert\frac{1}{m}\|\bA\bx\|_2^2-\gamma_1^2\cdot N^2\right\rvert.\]
Thus the desired claim follows from setting $C_2$ sufficiently large and then union bounding over all vectors $\bx\in\calN$. 
\end{proof}

In our analysis, we will make use of the following proposition about the $\chi^2$-divergence between a standard Gaussian vector $v \sim \calN(0,\sigma^2 I_n)$ and a standard discrete Gaussian mixture $\pi = v+ \mu$ for $v \sim \calD(0, \sigma^2 I_n)$ and unknown mean vector $\mu$.

\begin{proposition}[\cite{ingster2012nonparametric}, Proposition 2 in \cite{LiW16}]
\proplab{prop:tvd:chi}
$$\chi^2\left(\calN(0,I_n) * \mu || \calN(0,I_n) \right) \leq \mathbb{E}\left(e^{\langle z, z' \rangle/\sigma^2} \right) - 1$$ where $z, z' \sim \mu$ independently.
\end{proposition}

\subsection{Lattice Theory}
We recall a number of preliminaries from lattice theory, including the definitions of a lattice, dual lattice, orthogonal lattice, smoothing parameter, and successive minima. 
\begin{definition}[Lattice]
    A lattice is a discrete additive subgroup of $\mathbb {R}^m$: for any basis of linearly independent vectors $\bB = \{\bb_1, \ldots, \bb_n\} \subset \mathbb{R}^m$ for $m \geq n \geq 1$, define the lattice $\mathcal{L}$ generated by basis $\bB$ by

    $$\mathcal{L} = \calL(\bB) = \left\{\bB \bz = \sum_{i = 1}^n z_i \bb_i \ : \ \bz \in \mathbb{Z}^n \right\} $$
\end{definition}

\begin{definition}[Successive minima]
For $i\in[n]$, the successive minima $\lambda_i(\calL)$ of a lattice $\calL$ is the smallest radius $r$ such that a ball of radius $r$ centered at the origin contains at least $i$ linearly independent lattice vectors. 
\end{definition}

\begin{definition}[Orthogonal lattice]
Given a lattice $\calL\subset\mathbb{Z}^n$, we define the orthogonal lattice $\calL^{\perp}$ by
\[\calL^{\perp}=\{\by\in\mathbb{Z}^n\,:\,\langle\by,\bx\rangle=0\,\ \ \forall\,\bx\in\calL\}.\]
\end{definition}

\begin{definition}[Dual lattice \& smoothing parameter]
    For any lattice $\calL$, let  $\calL^*$ denote the dual lattice of $\calL$, defined as $\calL = \{\bx \in \mathbb{R}^n : \langle \mathcal{L}, \bx \rangle \subseteq \mathbb{Z}\}$. Moreover, for any lattice $\calL$ and $\eps > 0$, the smoothing parameter $\eta_{\eps}(\calL)$ is the smallest real $s > 0$ such that $\rho_{1/s}(\calL^* \setminus \{ 0\}) \leq \eps$.
\end{definition}

For a full row-rank matrix $\bA\in\mathbb{Z}^{m\times n}$, a full row-rank matrix $\bM\in\mathbb{R}^{m\times n}$, and a vector $\bv\in\mathbb{R}^m$, we define 
\[\calE_{\bA,\bM,\bv}=\{\bA\cdot\bx:\bx\sim\calD_{\mathbb{Z}^n+\bv,\bM}\}.\]
Note that $\calE_{\bA,\bM,\bv}$ is defined by sampling a discrete Gaussian vector $\bx$ and then returning $\bA\bx$, which is a different process than drawing a discrete Gaussian from the support $\bA\mathbb{Z}^n$ with the appropriate covariance. 
However, the following theorem states that these two processes are point-wise close in distribution under certain conditions for the successive minima of $\bA$: 
\begin{theorem} \lemlab{discrete_gaussian_tvd}
[Lemma 4 in \cite{AgrawalGHS13}]
For any rank $n$ lattice $\calL$, $\eps\in(0,1)$, vector $\bv\in\mathbb{R}^n$ and full-rank matrix $\bS\in\mathbb{R}^{n\times n}$ such that $\sigma_n(\bS)\ge\eta_\eps(\calL)$,
\[\rho_{\bS}(\calL+\bv)\in\left[\frac{1-\eps}{1+\eps},1\right]\cdot\rho_{\bS}(\calL).\]
\end{theorem}
We state the following result which is contained within the proof of Lemma 3.3 in \cite{AggarwalR16}, who show a total variation distance bound between $\calE_{\bA,\bS,\bv}$ and $\calD_{\bA\mathbb{Z}^n+\bA\bv,\bS\bA^\top}$ using the following stronger point-wise bound. 
\begin{theorem}
[Lemma 3.3 in \cite{AggarwalR16}]
\thmlab{thm:disc:tvd}
Let $\bA\in\mathbb{R}^{r\times n}$ be a full-rank matrix with $r<0.75n$ and let $\calL^\perp(\bA)$ be the orthogonal lattice to $\bA$. 
Let $\bv\in\mathbb{R}^n$ and $\bS\in\mathbb{R}^{n\times n}$ be any full-rank matrix such that the minimum singular value $\sigma_n(\bS)$ of $\bS$ satisfies
\[\sigma_n(\bS)>\lambda_{n - r}(\calL^\perp(\bA))\cdot\sqrt{\frac{\ln(2n(1+1/\eps))}{\pi}}.\]
For any vector $\bx\in\bA\mathbb{Z}^n+\bA\bv$, let $\rho_1(\bx)$ be the probability mass function of $\calE_{\bA,\bS,\bv}$ and let $\rho_2(\bx)$ be the probability mass function of $\calD_{\bA\mathbb{Z}^n+\bA\bv,\bS\bA^\top}$. 
Then we have 
\[\frac{1-\eps}{1+\eps}\le\frac{\rho_1(\bx)}{\rho_2(\bx)}\le\frac{1+\eps}{1-\eps}.\]
\end{theorem}
We also recall Siegel's lemma, which bounds the entries of vectors in the kernel of a matrix $\bA$. 
\begin{lemma}[Siegel's lemma]
\lemlab{lem:siegel}
\cite{silverman2000diophantine}
Let $\bA\in\mathbb{Z}^{r\times n}$ be a nonzero integer matrix with $r<n$ and entries bounded by $M$. 
Then there exists a nonzero vector $\bx\in\mathbb{Z}^n$ of integers bounded by $(nM)^{r/(n-r)}$ such that $\bA\bx=0^r$. 
\end{lemma}

\section{Lifting Framework}

In this section, we present our lifting framework and give our \thmref{thm:lifting}.  We first give the following matrix pre-processing lemma, which shows that given a matrix $\bA \in \mathbb{Z}^{r \times n}$, we can add at most $3r$ integer vectors to the rows of $\bA$ to form a new matrix $\bA'$ which satisfies $\lambda_{n - r}(\calL^\perp(\bA'))\le \sqrt{nM^2}$. 


\begin{lemma}[Matrix pre-processing lemma] 
\lemlab{lem:preprocessing}
Let $\bA\in\mathbb{Z}^{r\times n}$ for $r\le 0.25 n$ be a full rank integer matrix with entries bounded by $M$ for some $M \ge n$. 
There exists a pre-processing that adds rows to $\bA$ to form an  matrix $\bA'\in\mathbb{Z}^{m\times n}$ for $m\le 4r$, such that $\lambda_{n - r}(\calL^\perp(\bA'))\le \ell_{\bA}$, where $\ell_{\bA} = \sqrt{nM^2}$. 
\end{lemma}

\begin{proof}
    We first show that we can find $n - 4r$ integer vectors with entries bounded by $M$ in the kernel space of $\bA$. We use the probabilistic method to construct the vector sets sequentially. Suppose that we have chosen $t \le  n - 4r$ linear independent vectors in the kernel of $\bA$. Let $\bB \in \mathbb{R}^{(n - t) \times n}$ denote the matrix whose rows form a basis of the orthogonal complement to the span of these $t$ vectors. Without loss of generality, we can have $\bB$ contain $n - t$ columns to form a diagonal matrix with non-zero diagonal entries as we can choose the row span of $\bB$ arbitrarily.

    Next, we consider the following procedure where we randomly pick $\Theta(M^{1.5r})$ random integer vectors $\bv^1, \bv^2, \dots, \bv^s$ with entries in $\{0, 1, \ldots, M - 1\}$. Let $\mathcal{E}_1$ denote the event that there exists $1 \le i < j \le s$, such that $\bA \bv^i = \bA \bv^j$. Note that since each coordinate of $\bA \bv^i$ is in the range of $\{-nM^2, \ldots, nM^2\}$, which means that they are at most $M^{3r}$ possible values of $\bA \bv^i$. Since we randomly sample $\Theta(M^{1.5r})$ vectors $\bv^i$, by birthday paradox we have $\Pr [\mathcal{E}_1] \ge 0.9$. 

    Let $\bx = \bv^i - \bv^j$ be a random integer vector with coordinates in $\{-(M - 1), \ldots, M - 1\}$. Since we have $\bB$ contain $n - t$ columns to form a diagonal matrix with non-zero diagonal entries, we can get that $\Pr [\bB \bx = \mathbf{0}] \le \left(\frac{1}{M}\right)^{n - t}$ (as when the remaining $t$ coordinates are fixed, there will be unique assignment for these $n - t$ coordinates to make $\bB \bx = 0$).
    Let $\mathcal{E}_2$ denote the event that for every $1 \le i < j \le s$, we have $\bB \bv^i \ne \bB \bv^j$. Since $n - t \ge 4r $, after taking a union bound we have $\Pr [\mathcal{E}_2 ] \ge 0.9$.

    From the above discussion, after taking a union bound, we have $\Pr [\mathcal{E}_1 \cap \mathcal{E}_2] \ge 0.8$. This means that there exists $\by = \bx^i - \bx^j$ such that(1) $\bA \by = \mathbf{0}$,  and (2)$\bB \by \ne \mathbf{0}$. Recall that $\bB$ is the matrix whose rows form a basis of the orthogonal complement to span of these $t$ vectors we have chosen. Hence from~(2) we have $\by$ is linearly independent to the $t$ vectors we have chosen so far, and hence we can choose $\by$ as the $(t + 1)$-th integer vectors in the kernel of $\bA$. 

    Now after we have chosen $n - 4r$ such linear independent integer vectors with each coordinate in $\{-(M - 1), \ldots, M - 1\}$. Let $V$ be the subspace spanned by these vectors. Then we have the dimension of $\mathrm{ker}(\bA) \setminus V$ is $3r$. Therefore, we can iteratively apply Siegel’s Lemma to generate a set of 
$3r$ vectors as a basis of $\mathrm{ker}(\bA) \setminus V$ and add them to the rows of $\bA$ to form a new matrix $\bA' \in \mathbb{Z}^{4r \times n}$. Note that the integer vectors we choose are still orthogonal to the rows of $\bA'$, which means that $\lambda_{\max}(\calL^\perp(\bA'))\le \sqrt{nM^2}$.  
\end{proof}

\begin{lemma}
\lemlab{lem:main:disc:to:cont}
Let $C>0$ be any fixed constant and let $\bS\in\mathbb{R}^{n\times n}$ be any full rank matrix. 
Let $\bA\in\mathbb{R}^{r\times n}$ be an orthonormal matrix with $r\le n$, such that:
\begin{enumerate}
\item 
There exists a matrix $\bM\in\mathbb{R}^{r\times n}$ with rational entries and the same rowspan as $\bA$. 
\item
$\lambda_{\max}(\calL^\bot(\bA))\le\frac{\sigma_n(\bS)}{10C\log n}$.
\item
$n^{6C}\ge\sigma_1(\bS)\ge\sigma_n(\bS)\ge n^{5C+3}$.
\end{enumerate}
Suppose $\bx\sim\calD(0,\bS^\top\bS)$ and $\bfEta$ be drawn uniformly at random from the fundamental parallelepiped of $\calL(\bA)$. 
Let $p$ be the probability density function of $\bA\bx+\bfEta$ and let $q$ be the probability density function of $\calN(0,\bA\bS^\top\bS\bA^\top)$. 
Then with probability $1-\frac{1}{\poly(n)}$ for $\bv=\bA\bx\in\mathbb{R}^r$ over $\bx\sim\calD(0,\bS^\top\bS)$, we have
\[\frac{p(\bv)}{q(\bv)}\in\left[1-\frac{1}{n^C},1+\frac{1}{n^C}\right].\]
\end{lemma}
\begin{proof}
Let $p'$ be the probability density function of $\by+\eta$, where $\by\sim\calD_{\bA\mathbb{Z}^n,\bS\bA^\top}$. 
Under the assumption that $\lambda_{\max}(\calL^\bot(\bA))\le\frac{\sigma_n(\bS)}{10C\log n}$, we apply \thmref{thm:disc:tvd}, so that
\[\frac{p(\bv)}{p'(\bv)}\in\left[1-\frac{1}{n^{2C}},1+\frac{1}{n^{2C}}\right]\]
for all vectors $\bv\in\mathbb{R}^r$. 

We first claim that the unit cell is well-defined, i.e., that there exists a fundamental parallelepiped in the resulting lattice. 
Suppose that there exists a matrix $\bM\in\mathbb{R}^{r\times n}$ with rational entries and the same rowspan as $\bA$, so that the fundamental parallelepiped of $\bM$ is well-defined and non-zero. 
Now we can write the projection matrix of $\bM$ as $\bM^\top(\bM\bM^\top)^{-1}\bM$, which has the same row span as $\bM$. 
Furthermore, because $\bM$ only has rational entries, then $\bM^\top(\bM\bM^\top)^{-1}\bM$ only has rational entries. 
We write $\bU\bU^\top=\bM^\top(\bM\bM^\top)^{-1}\bM$, so that $\bU$ is orthonormal and has the same rowspan as $\bM$. 
Therefore, since the entries of $\bx$ are integers, then the nonzero entries of $\bU\bU^\top\bx$ is lower bounded by some rational number and thus $\bU$ has a fundamental parallelepiped. 
Since $\bA$ is a rotation of $\bU$, then it follows that $\bA$ also has a fundamental parallelepiped, and so the unit cell of $\bA$ is well-defined. 

Now, consider the largest possible distance between vectors $\bv$ and $\bv'$ in the same unit cell of $\bA$. 
Let $A^{(i_1)},\ldots,A^{(i_r)}$ be $r$ linearly independent columns of $A$ and let $\calC$ be the unit cell defined by these $r$ columns. 
Note that any additional columns of $\bA$ can only decrease the size of the unit cell. 
Hence, the length of the diagonal of $\calC$ is an upper bound on the largest possible distance between $\bv$ and $\bv'$. 
Therefore, the length of the diagonal is at most $\|A^{(i_1)}\|_2+\ldots+\|A^{(i_r)}\|_2\le r\sqrt{r}$ since all the rows of $\bA$ are orthonormal. 

Hence, for all vectors $\bv$ and $\bv'$ in a unit cell, we can write $\bv'=\bv+\bw$ for a vector $\bw$ with $\|\bw\|_2\le r\sqrt{r}$. 
Thus,
\begin{align*}
\bv'(\bA\bS^\top\bS\bA^\top)^{-1}(\bv')^\top&=\bv'\bA(\bS^\top\bS)^{-1}\bA^\top(\bv')^\top\\
&=\left(\bv+\bw\right)\bA(\bS^\top\bS)^{-1}\bA^\top\left(\bv+\bw\right)^\top\\
&=\bv\bA(\bS^\top\bS)^{-1}\bA^\top\bv^\top+2\bv\bA(\bS^\top\bS)^{-1}\bA^\top\bw^\top+\bw\bA(\bS^\top\bS)^{-1}\bA^\top\bw^\top.
\end{align*}
Moreover, for $\bx\sim\calD(0,\bS^\top\bS)$, we have $\|\bx\|_2=\O{n^{6C+1}}$ with high probability since $\sigma_1(\bS)\le n^{6C}$. 
Since $\bv=\bA\bx$ and $\bA$ has orthonormal rows, then we also have $\|\bv\|_2=2n^{6C+1}$ with high probability. 
We also have $\|\bw\|_2\le r\sqrt{r}$ and $\sigma_n(\bS)\ge n^{5C+3}$. 
Therefore,
\[\left\lvert\bv'(\bA\bS^\top\bS\bA^\top)^{-1}(\bv')^\top-\bv\bA(\bS^\top\bS)^{-1}\bA^\top\bv^\top\right\rvert\le 6r\sqrt{r}n^{6C+1}\cdot\frac{1}{n^{10C+6}}\le\frac{1}{n^{4C}},\]
for sufficiently large $n$. 
Hence, we have
\[\frac{\rho_{\bA^\top}(\bv)}{\rho_{\bA^\top}(\bv')}\in\left[1-\frac{1}{n^{2C}},1+\frac{1}{n^{2C}}\right].\]
On the other hand, all vectors $\bv$ and $\bv'$ in a unit cell have the same probability density under $p'$. 
Therefore, we have 
\[\frac{p'(\bv)}{q(\bv)}\in\left[1-\frac{1}{n^{2C}},1+\frac{1}{n^{2C}}\right]\]
for all vectors $\bv\in\mathbb{R}^r$. 
Hence, it follows that
\[\frac{p(\bv)}{q(\bv)}\in\left[1-\frac{1}{n^C},1+\frac{1}{n^C}\right].\]
\end{proof}
We define the following notion of smoothness to capture functions that, informally, are not sensitive under small perturbations. 
\begin{definition}[Smoothness]
Let $\calI$ denote an arbitrary distribution on $\mathbb{Z}^n$. 
A function $f:\mathbb{R}^n \to S$ where $S$ is an finite set in $\mathbb{Z}$ is $\delta$-smooth with respect to a distribution $\calD(\mu,\bSigma)+\calI$ if 
\[
\underset{\substack{X\sim(\calD(\mu,\bSigma)+\calI),\,Y\sim(\calN(\mu,\bSigma)+\calI)}}\Pr\left[f(X) \ne f(Y)\right]\le\delta.
\]
\end{definition}
For example, some function that exhibits drastically different behavior on the integers versus the real-numbers, e.g., $f(x)=1$ if $x\in\mathbb{Z}$ and $f(x)=0$ if $x\in\mathbb{R}\setminus\mathbb{Z}$ is not smooth. 
These functions are difficult to capture by our framework because a small perturbation in the fundamental parallelepiped within the cell of a lattice of the sketch matrix can result in a completely different value after post-processing. 

We are now ready to prove our main theorem in this section.
\begin{restatable}{theorem}{thmlifting}
\thmlab{thm:lifting}
Let $\calI$ denote an arbitrary distribution on $\mathbb{Z}^n$ and let $\delta\ge\frac{1}{\poly(n)}$ be a failure parameter. 
Let $f$ be a $\frac{\delta}{3}$-smooth function on a distribution $\calD(0,\bS^\top\bS)+\calI$ for a full-rank covariance matrix $\bS^\top\bS\in\mathbb{R}^{n\times n}$ and suppose $f$ admits an orthonormal linear sketch $\bA\in\mathbb{R}^{r\times n}$ with $r\le n$ and an arbitrary post-processing function $g$, such that:
\begin{enumerate}
\item 
There exists a matrix $\bM\in\mathbb{R}^{r\times n}$ with rational entries and the same rowspan as $\bA$. 
\item
$g(\bA\bx)=f(\bx)$ with probability at least $1-\frac{\delta}{3}$ over $\bx\sim\calD(0,\bS^\top\bS)+\calI$ 
\item
$\lambda_{\max}(\calL^\bot(\bA))\le\frac{\sigma_n(\bS)}{10C\log n}$.
\item
$n^{6C}\ge\sigma_1(\bS)\ge\sigma_n(\bS)\ge n^{5C+3}$.
\end{enumerate}
There there exists a post-processing function $h$ with another sketching matrix $\bA' \in \mathbb{R}^{4r \times n} $  such that with probability at least $1-\delta$ over $\bx\sim\calN(0,\bSigma)+\calI$, we have $h(\bA'\bx)=f(\bx)$. 
\end{restatable}
\begin{proof}
Let $\bSigma=\bS^\top\bS$ be the covariance matrix. 
Let $\phi:\mathbb{R}^r\to\mathbb{R}^r$ map any vector $\bv\in\mathbb{R}^r$ to the lattice point in the unit cell of $\bv$ induced by $\calL(\bA)$. 
Given $\bx\sim\calN(0,\bSigma)+\calI$, consider the post-processing function $h(\bA\bx)=g(\phi(\bA\bx))$. 
We show that $h$ is correct with probability at least $1-\delta$ over $\bx\sim\calN(0,\bSigma)+\calI$. 
The latter two assumptions satisfy the conditions of \lemref{lem:main:disc:to:cont}. 
Therefore, we have that $\bA\bx$ for $\bx\sim\calN(0,\bSigma)$ has total variation distance $\frac{1}{\poly(n)}$ with the distribution that first draws $\bx\sim\calD(0,\bSigma)$ and then computes $\bA\bx+\bfEta$ for $\bfEta$ uniformly at random from the fundamental parallelepiped of $\calL(\bA)$. 
Since $\calI$ only has support on the lattice points, then it follows that $\bA\bx$ for $\bx\sim(\calN(0,\bSigma)+\calI)$ also has total variation distance $\frac{1}{\poly(n)}$ with the distribution $\overline{\calD}$ that first draws $\bx\sim(\calD(0,\bSigma)+\calI)$ and then computes $\bA\bx+\bfEta$ for $\bfEta$ uniformly at random from the fundamental parallelepiped of $\calL(\bA)$. 

Observe that the distribution of $h(\bA\bx)$ for $\bx\sim\overline{\calD}$ is the same as the distribution of $g(\bA\bx)$ for $\bx\sim\calD(0,\bSigma)+\calI$. 
By assumption, $g(\bA\bx)$ is correct with probability $1 - \frac{\delta}{3}$ over $\bx\sim\calD(0,\bSigma)+\calI$. 
Moreover, because $f$ is $\frac{\delta}{3}$-smooth on a distribution $\calD(0,\bSigma)+\calI$, then $\bx$ being drawn from $\bx\sim(\calN(0,\bSigma)+\calI)$ rather than $\bx\sim(\calD(0,\bSigma)+\calI)$ only affects the outcome with probability $\frac{\delta}{3}$. 
Hence, $h(\bA\bx)=f(\bx)$ for $\bx\sim(\calN(0,\bSigma)+\calI)$ with probability at least $1-\frac{\delta}{3}-\frac{1}{\poly(n)}-\frac{\delta}{3}\ge 1-\delta$. 
\end{proof}

We also show the following useful statement, which informally shows that the image of a rounded continuous Gaussian is close in distribution to the image of a discrete Gaussian. 
\begin{restatable}{lemma}{lemdistroundclose}
\lemlab{thm:dist:round:close}
Let $\bS^\top\bS\in\mathbb{R}^{n\times n}$ be a full-rank covariance matrix and suppose $\bA\in\mathbb{R}^{r\times n}$ with $r\le n$ such that:
\begin{enumerate}
\item 
There exists a matrix $\bM\in\mathbb{R}^{r\times n}$ with rational entries and the same rowspan as $\bA$. 
\item
$\lambda_{\max}(\calL^\bot(\bA))\le\frac{\sigma_n(\bS)}{10C\log n}$.
\item
$n^{6C}\ge\sigma_1(\bS)\ge\sigma_n(\bS)\ge n^{5C+3}$.
\end{enumerate}
Let $\calD_1$ be the distribution of $\bA\bx$, where $\bx\sim\calD(0,\bS^\top\bS)$ and let $\calD_2$ be the distribution of $\by$, where $\by$ is obtained by rounding $\bA\bx$ for a sample $\bx\sim\calN(0,\bS^\top\bS)$ to the lattice point of its unit cell in $\bA\cdot\mathbb{Z}^n$. 
Then $\TVD(\calD_1,\calD_2)\le\frac{1}{n^C}$. 
\end{restatable}
\begin{proof}
First we note that by the same argument as the proof of \lemref{lem:main:disc:to:cont}, the unit cell of $\bA$ is well-defined, due to the first assumption. As before, we may assume that $\bA$ has orthonormal rows, without loss of generality.

For a vector $\bv$ in the integer rowspan of $\bA$, let $p(\bv)$ be the probability mass function of $\calD_1$ and $q(\bv)$ be the probability mass function of $\calD_2$. 
Let $\bSigma=\bS^\top\bS$ be the covariance matrix and let $\phi:\mathbb{R}^r\to\mathbb{R}^r$ map any vector $\bv\in\mathbb{R}^r$ to the lattice point in the unit cell of $\bv$ induced by $\calL(\bA)$. 
Let $p'$ be the probability density function of $\by$, where $\by\sim\calD_{\bA\mathbb{Z}^n,\bS\bA^\top}$. 
Under the assumption that $\lambda_{\max}(\calL^\bot(\bA))\le\frac{\sigma_n(\bS)}{10C\log n}$, we apply \thmref{thm:disc:tvd}, so that
\[\frac{p(\bv)}{p'(\bv)}\in\left[1-\frac{1}{n^{2C}},1+\frac{1}{n^{2C}}\right]\]
for all vectors $\bv\in\mathbb{R}^r$. 

Note that by rotational invariance of Gaussian distributions, we have that $\bA\bx\sim\calN(0,\bA\bA^\top)$. 
Now by the same argument as the proof of \lemref{lem:main:disc:to:cont}, we have that for $\bv=\bA\bx$, then with high probability over $\bx\sim\calD(0,\Sigma)$
\[\frac{\rho_{\bA^\top}(\bv)}{\rho_{\bA^\top}(\bv')}\in\left[1-\frac{1}{n^{2C}},1+\frac{1}{n^{2C}}\right],\]
for any two vectors $\bv,\bv'$ in the same unit cell of $\bA$. 
Therefore, it follows that
\[\frac{p'(\bv)}{q(\bv)}\in\left[1-\frac{1}{n^{2C}},1+\frac{1}{n^{2C}}\right]\]
and thus $\TVD(\calD_1,\calD_2)\le\frac{1}{n^C}$. 
\end{proof}

\section{Adversarially Robust Streaming: \texorpdfstring{$L_p$}{Lp} Estimation using Linear Sketches}

We consider the problem  of adversarially robust $\ell_p$ estimation for $p \geq 1$. In particular, \cite{HardtW13} gave a real-valued adaptive attack against linear sketches for $\ell_p$ estimation that crucially relied on the adaptive sequence of inputs being drawn from continuous Gaussian distributions. 
In this section, we show an attack against linear sketches in the streaming model, where all integer query vectors are drawn from appropriate discrete Gaussian distributions which, with probability $ 1 - \exp(-cn)$, have entries bounded by $\poly(n)$.

\subsection{Preliminaries for Adversarial Attack}

\paragraph{Notation.}
Recall that a linear sketch is given by a distribution $\mathcal{M}$ over $r \times n$ matrices, together with an estimator $F:\mathbb{R}^{r \times n} \times \mathbb{R}^{n} \rightarrow \{0,1\}$, which may be arbitrary. 
Then, the streaming algorithm $\calA$ initially samples a matrix $\bA \sim \calM$, and for each query $\bx \in \mathbb{Z}^n$, $\calA$ returns $F(\bA, \bA \bx)$. 
Without loss of generality, we assume that $\bA$ has full row rank, and let $R(\bA)$ denote the rowspace of $\bA$. For convenience, we will identify the sketching matrix $\bA$ with its rowspace, which is an $r$-dimensional subspace of $\mathbb{R}^n$. 
Therefore, we can express the sketch as a mapping $f:\mathbb{Z}^n \rightarrow \{0,1\}$ such that $f(\bx) = f(P_{\bA}\bx)$, where $P_{\bU}$ denotes the orthogonal projection operator onto the space $\bU$. 
We construct our attack for the following promise problem.

\begin{definition}[$\textrm{GapNorm}(B, \alpha)$ promise problem] For any $B \geq 8$ and $\alpha = \poly(n)$ that are fixed in advance, we say an algorithm $\mathcal{A}$ solves $\textrm{GapNorm}(B, \alpha)$ if the following holds: for any input vector $\bx \in \mathbb{Z}^n$, $\mathcal{A}$ returns $1$ if $\|\bx\|_2^2 \geq \alpha \cdot B$, and returns $0$ if $\|\bx\|_2^2 \leq \alpha$. Note that if $\|\bx\|_2^2 \in (\alpha, \alpha \cdot B)$, $\mathcal{A}$ may return either $0$ or $1$.
\end{definition}

For convenience in the analysis of our discrete attack, we define the family of distributions $G(V^{\perp}, \sigma^2)$ which is used for the real-valued attack of \cite{HardtW13}.

\begin{definition}[Distribution $G(V^{\perp}, \sigma^2)$ \cite{HardtW13}]
Given a subspace $V\subseteq\bA$ of dimension $t\le r-1$, let $d=r-t$. 
We define the distribution $G(V^\perp,\sigma^2)$ implicitly as the distribution induced by first independently sampling $\bg_1\sim \calN(0,\frac{3\sigma^2}{4})^n$, $\bg_2\sim \calN\left(0,\frac{\sigma^2}{4}\right)^n$ and then setting $\bg=P_{V^\perp}\bg_1+\bg_2$. 
    
\end{definition}

Now, we define the family of discrete Gaussian distributions from which queries are drawn in the attack.

\begin{definition}[Distribution $D(V^{\perp}, \sigma^2)$]
Given a subspace $V\subseteq\bA$ of dimension $t\le r-1$, let $d=r-t$. 
We define the distribution $D(V^\perp,\sigma^2)$ implicitly as $\calD_{\mathbb{Z}^n}(0, \Sigma_{\sigma^2}))$, where $\Sigma_{\sigma^2}$ is the covariance matrix from $G(V^{\perp}, \sigma^2)$, i.e. $\Sigma_{\sigma^2} = \frac{3\sigma^2}{4}P^\top_{V^\perp}P_{V^\perp}+\frac{\sigma^2}{4}\cdot I_n$.
\end{definition}

\subsection{Discrete Conditional Expectation Lemma}

In this part, we prove a discrete version of the Conditional Expectation Lemma of \cite{HardtW13}. First, we define the following notion to intuitively capture the desired behavior of an estimator $f$ on inputs drawn from the distribution $D(V^\perp,\sigma^2)$. We state this definition generally for any distribution $\calD$ with parameter $\sigma^2$: in later parts, this $\calD$ may be $D(V^\perp,\sigma^2)$ for our discrete attack or $G(V^\perp,\sigma^2)$ from the continuous attack. 

We first formally define ``correctness'' of a function $f$ with respect to the GapNorm promise problem. 
\begin{definition}[Correctness]
    A function $f:\bA\to\{0,1\}$ is $(\delta,\alpha, B)$-correct on $V^\perp$ for inputs drawn from a distribution family $\calD$ with $B \geq 4$ and $d=\dim(V^\perp\cap \bA)$ if:
\begin{enumerate}
\item
For all $\sigma^2\in[\alpha \cdot B/2,2\alpha\cdot B]$ and $\bg\sim \calD$, we have $\PPr{f(\bg)=1}\ge 1-\delta$.
\item
For all $\sigma^2\in[\alpha,2\alpha]$ and $\bg\sim \calD$, we have $\PPr{f(\bg)=1}\le\delta$.
\end{enumerate}

\end{definition}

We next recall the Conditional Expectation Lemma of \cite{HardtW13}, which notes that there exist certain directions $\bu$ that are more aligned with a sketch matrix $\bA$, thus resulting in a small bias in the expected value of $\langle\bu,\bg\rangle^2$ for a random Gaussian $\bg$, provided that $f(\bg)=1$. 
\begin{lemma}[Conditional Expectation Lemma \cite{HardtW13}]\lemlab{lem:cond:exp}
    Let $\bA \subseteq \mathbb{R}^n$ be a subspace of dimension $\dim(\bA) = r \leq n - d_0$ for sufficiently large constant $d_0$. Let $V \subseteq \bA$ be a subspace of $\bA$ of dimension $t \leq r$. Suppose that $f: \bA \rightarrow \{0,1\}$ is $\left(\frac{1}{5(d_0B)^2}, \alpha, B\right)$-correct on $V^{\perp}$ for inputs drawn from $G(V^{\perp}, \sigma^2)$. Then, there exists a scalar $\sigma^2 \in [\alpha, B\alpha]$ and a vector $\bu \in A \cap V^{\perp}$ such that for $\bg \sim G(V^{\perp}, \sigma^2)$, we have that for $d = \max(r-d, d_0)$:

    \begin{enumerate}
        \item $\mathbb{E}\left[\langle \bu, \bg\rangle^2\,\mid\,f(\bg) = 1 \right] \geq \mathbb{E}\left[\langle \bu, \bg \rangle^2 \right] + \frac{\alpha}{10 B d}$
        \item $\Pr\left[f(\bg) = 1 \right] \geq \frac{1}{40 B^2 d}$.
    \end{enumerate}
\end{lemma}

Let $\Sigma$ denote the covariance of the distribution family $D(V^{\perp}, \sigma^2)$. Note that $\Sigma$ depends on parameter $\sigma^2$ implicitly.
We show an analog of the Conditional Expectation Lemma (\lemref{lem:cond:exp}) when the input $\bx$ is instead drawn from a discrete Gaussian. 
We remark that such a statement was previously a significant challenge and illustrates the power of \lemref{thm:dist:round:close}, since the continuous analog in \cite{HardtW13} crucially used the rotational invariance of continuous Gaussians. 

\begin{lemma}[Discrete Conditional Expectation Lemma]\lemlab{lem:disc_cond:exp}

Let $\bA$ be a subspace of dimension $\dim(\bA)=r\le n-d_0$ for some sufficiently large constant $d_0$.
Let $V\subseteq\bA$ be a subspace of $\bA$ of dimension $t\le r$. Furthermore, suppose $\sigma_n(\Sigma) \geq \lambda_{\max}(\calL^{\perp}(\bA)^2 \cdot \frac{\ln(2n(1+1/\eps))}{\pi}$.
Suppose $f:\bA\to\{0,1\}$ is $\left(\frac{1}{10B^2d_0^2} + \frac{1}{\poly(n)},\alpha, B\right)$-correct on $V^\perp$. 
Then, there exists a scalar $\sigma^2\in\left[\alpha, \alpha B\right]$ and a vector $\bu\in\bA\cap V^\perp$ such that for $\bx\sim D(V^\perp,\sigma^2)$, we have that for $d=\max(r-d_0,d_0)$:
\begin{enumerate}
\item
$\Ex{\langle\bu,\bx\rangle^2\,\mid\,f(\bx)=1}\ge\Ex{\langle\bu,\bx\rangle^2}+\frac{\alpha}{4Bd}$
\item
$\PPr{f(\bx)=1}\ge\frac{1}{40B^2 d}$.
\end{enumerate}
\end{lemma}
\begin{proof}
We will prove the discrete Conditional Expectation Lemma by contradiction: suppose that there exists a discrete linear sketch $f:\bA \rightarrow \{0,1\}$ such that for every $\sigma^2 \in [\alpha, \alpha B]$, there is no vector $\bu \in \bA \cap V^{\perp}$ such that the statements in (1) and (2) hold for $\bx \sim D(V^{\perp}, \sigma^2)$. 
    
Before we proceed, define a partition of $\calL(\bA)$ as follows: fix a fundamental parallelepiped $\calF$ with diameter $\poly(r)$ in $\calL(\bA) \subseteq \mathbb{R}^r$. 
Observe that this fundamental parallelepiped can be used to partition $\mathbb{R}^r$ into cells such that each lattice point $y$ (such that $\bz = \bA \by$ for some $\by \in \mathbb{Z}^n$) is associated to one cell. 
Let $\mathcal{P}(\calL(\bA))$ denote the partition of $\calL(\bA)$ into such parallelepipeds. 

Let $\phi:\mathbb{R}^r\to\mathbb{R}^r$ be the function that maps each vector $\bv$ to the canonical lattice point of the unit cell of $\bv$. 
Let $p$ be the probability mass function of $\phi(\bA\bg)$ for $\bg\sim G(V^\perp,\sigma^2)$, and let $q$ be the probability mass function of $\bA\by$ for $\by\sim D(V^{\perp}, \sigma^2)$. 
Since $\sigma_n(\Sigma) \geq \lambda_{\max}(\calL^{\perp}(\bA)^2 \cdot \frac{\ln(2n(1+1/\eps))}{\pi}$, by \lemref{thm:dist:round:close} the distribution of $\phi(\bA\bg)$ has total variation distance at most $\frac{1}{\poly(n)}$ to $\bA\by$ for $\by\sim D(V^\perp,\sigma^2)$. 
%

Suppose that $f: \bA \rightarrow \{0,1\}$ is a $\left(\frac{1}{10(d_0B)^2}, \alpha, B\right)$-correct discrete-valued linear sketch from above. 
We construct a real-valued linear sketch $f' : \bA \rightarrow \{0,1\}$ that, for any query $\bg \sim G(V^{\perp}, \sigma^2)$, observes $\bA\bg$ and then computes $\by = \phi(\bA\bg)$ by rounding the point $\bA\bg$ to the integer lattice point $\bA \by$ induced by partition $\mathcal{P}(\calL(\bA))$; finally, $f'$ outputs $f'(\bg) = f(\by)$. 
Since $\bA \bg$ for $\bg \sim G(V^{\perp}, \sigma^2)$ is $\frac{1}{\poly(n)}$-close to $\bA \by + \bfEta$ for $\by \sim D(V^{\perp}, \sigma^2)$, and without loss of generality any sketch can round $\bA \by + \bfEta$ to recover $\bA \by$, it follows that $f'$ is $\left(\frac{2}{10(d_0B)^2} + \frac{1}{\poly(n)}, \alpha, B\right)$-correct for $\bg \sim G(V^{\perp}, \sigma^2)$. 

Now, by \lemref{lem:cond:exp}, there exists $\bu \in \bA \cap V^{\perp}$ such that the conclusion of the real-valued Conditional Expectation Lemma holds for $\bg \sim G(V^{\perp}, \sigma^2)$ for some $\sigma^2 \in [\alpha, B \alpha]$. 
Observe that by conditioning on the event $\calE_1 = \{f'(\bg) = 1\}$ versus $\calE_2 = \{f(\by) = 1\}$, the probability distribution of $\langle \bu,  \bg \rangle^2$ changes by only an additive $\frac{1}{\poly(n)}$. 
Also, $\mathbb{E}\left[\langle \bu, \bg \rangle^2\right] \leq \sigma^2$ and $\mathbb{E}\left[\langle \bu, \by \rangle^2\right] \leq \sigma^2$, so if we condition on $\|\bg\|_2^2 \leq \sigma^2 \log(n)$, it follows that the probability distribution of $\langle \bu, \by \rangle^2$ conditioned on $\calE_2$ is changed from the probability distribution of $\langle \bu, \bg \rangle^2$ conditioned on $\calE_2$ by at most $\frac{\sigma^2}{\poly(n)}$. 
Thus, we get that
\[
\mathbb{E}\left[\langle \bu, \bg \rangle^2 | f'(\bg) = 1 \right] - \mathbb{E}\left[\langle \bu, \bg \rangle^2 \right] < \frac{\alpha}{10Bd} + \frac{\sigma^2}{\poly(n)} + \frac{1}{\poly(n)} \leq \frac{\alpha}{4Bd}.
\]
However, since $f'$ is a $\left(\frac{2}{10(d_0B)^2} + \frac{1}{\poly(n)}, \alpha, B \right)$-correct linear sketch for $G(V^{\perp}, \sigma^2)$, this contradicts the Conditional Expectation Lemma of \cite{HardtW13}. 
\end{proof}

\subsection{Algorithm for Adversarial Attack}
In this section, we describe our adversarial attack for the GapNorm promise problem, which is summarized in \figref{fig:attack}. 
\begin{figure*}[!htb]
\begin{mdframed}
\textbf{Input}: Oracle $\calA$ providing access to a function $f:\mathbb{R}^n\to\{0,1\}$, parameters $B\ge 4$, and sufficiently large $\alpha = \poly(n)$ satisfying $\alpha \geq \ell_{\bA}^2 \cdot \frac{\ln(2n(1+1/\eps))}{\pi}$  after pre-processing via \lemref{lem:preprocessing}, for all possible integer matrices $\bA \in \mathbb{Z}^{r \times n}$ initially with $\poly(n)$-bounded entries.
\newline\noindent
\textbf{Attack}: Let $V_0=\emptyset$, $m=\O{B^{13}n^{11}\log^{15}(n)}$, $S=\left[\alpha, \alpha \cdot B\right]\cap\zeta\mathbb{Z}$ where $\zeta=\frac{1}{20(Bn)^2\log(Bn)}$.
\newline\noindent
\textbf{For $t\in[r+1]$:}
\begin{enumerate}
\item
For each $\sigma^2\in S$:
\begin{enumerate}
\item
Sample $\bx_1,\ldots,\bx_m\sim D(V^\perp,\sigma^2)$. 
Query $\calA$ on each $\bx_i$ and let $a_i=\calA(\bx_i)$.
\item
Let $s(t,\sigma^2)=\frac{1}{m}\sum_{i=1}^m a_i$ denote the fraction of samples that are positively labeled.
\begin{enumerate}
\item
If either (1) $\sigma^2\ge \alpha \cdot B/2$ and $s(t,\sigma^2)\le 1-\zeta$ or (2) $\sigma^2\le 2 \cdot \alpha$ and $s(t,\sigma^2)\ge\zeta$, then terminate and \textbf{return} $(V_t^\perp,\sigma^2)$.
\item
Else let $\bx'_1,\ldots,\bx'_{m'}$ be the vectors such that $\calA(\bx'_i)=1$ for all $i\in[m']$. 
\end{enumerate}
\item
If $m'<\frac{m}{100B^2n}$, increment $\sigma^2$. Else, compute $v_{\sigma}=\argmax_{\bv\in\mathbb{R}^n}z(\bv)$ for $z(\bv)=\frac{1}{m'}\sum_{i=1}^{m'}\langle\bv,\bx'_i\rangle^2$
\end{enumerate}
\item
Let $\bv'$ represent the first vector $\bv_\sigma$ with $z(\bv_\sigma)\ge\sigma^2+\frac{\sigma^2}{4}+\frac{1}{14Br}$.
\begin{enumerate}
\item
If no such $\bv_\sigma$ was found, set $V_{t+1}=V_t$ and proceed to the next round.
\item 
Otherwise, let $\bv^* = \bv'$. 
Compute $\bv_t=\bv^*-\frac{\sum_{\bv\in V_t}\bv\langle\bv,\bv^*\rangle}{\left\|\sum_{\bv\in V_t}\bv\langle\bv,\bv^*\rangle\right\|_2}$ and set $V_{t+1}=V_t\cup\{\bv_t\}$. 
\end{enumerate}
\end{enumerate}
\end{mdframed}
\caption{Algorithm that creates an adaptive attack using a turnstile data stream.}
\figlab{fig:attack}
\end{figure*}

Ultimately, the goal of the attack is to find the following notion of a failure certificate:
\begin{definition}[Failure certificate]
Let $B\ge8$, $\sigma^2\in[\alpha ,2\alpha \cdot B]$, and $f:\mathbb{R}^n\to\{0,1\}$.  
We say a pair $(V,\sigma^2)$ is a $d$-dimensional failure certificate for $f$ if $V\subseteq\mathbb{R}^n$ is a $d$-dimensional subspace such that there exists a constant $C>0$ for which $n\ge d+10C\log(Bn)$ and:
\begin{itemize}
\item 
Either $\sigma^2\in[\alpha \cdot B/2,50 \alpha \cdot B]$ and $\PPPr{\bg\sim D(V^\perp,\sigma^2)}{f(\bg)=1}\le1-(Bn)^{-C}$,
\item 
or $\sigma^2 \in \left[\alpha, 2 \alpha \right]$ and $\PPPr{\bg\sim D(V^\perp,\sigma^2)}{f(\bg)=1}\ge n^{-C}$.
\end{itemize}
\end{definition}
We first show that a failure certificate for $f$ can be used to find a discrete-valued input vector $\bx$ for which $f$ fails on $\bx$, similar to Fact 5.2 in \cite{HardtW13}, which only requires $f$ to fail on a real-valued vector $\bx$. 
\begin{theorem}
Let $\eta=\frac{1}{\poly(n)}$ be a fixed parameter. Let $\alpha = \poly(n)$ large enough. Given a $d$-dimensional failure certificate for $f$, there exists an algorithm that uses $\poly(Bn)$ non-adaptive queries and with probability at least $\frac{2}{3}$, outputs a vector $\bx\in(\eta\cdot\mathbb{Z})^n$ such that either $\|\bx\|_2>\frac{\alpha B(n-d)}{3}$ and $f(\bx)=0$ or $\|\bx\|_2<3\alpha(n-d)$ and $f(\bx)=1$. 
\end{theorem}

\begin{proof}
Suppose we sample a set $X$ of $\O{(\alpha Bn)^C}$ samples from $D(V^\perp,\sigma^2)$ with $\alpha \le \sigma^2\le 2\alpha$. 
Since $n-d\ge d$, then by sub-Gaussian concentration and a union bound, we have that with high probability simultaneously for all samples $\bx\in X$, we have $\|\bx\|_2^2<3\alpha (n-d)$. 
On the other hand, since $D(V^\perp,\sigma^2)$ is a failure certificate, then with high probability, there exists $\bx$ such that $f(\bA\bx)=1$ for some $\bx\in X$. 

The case where $\sigma^2\ge\frac{B\alpha }{2}$ and $\|\bx\|_2>\frac{\alpha B(n-d)}{3}$ is symmetric. 
Namely, suppose we sample a set $X$ of $\O{(\alpha Bn)^C}$ samples from $D(V^\perp,\sigma^2)$ with $\sigma^2\ge\frac{\alpha B}{2}$. 
Because $n-d\ge d$, then by sub-Gaussian concentration and a union bound, with high probability, we have $\|\bx\|_2^2>\frac{\alpha B(n-d)}{3}$ simultaneously for all samples $\bx\in X$. 
Now because $D(V^\perp,\sigma^2)$ is a failure certificate, then with high probability, there exists $\bx$ such that $f(\bA\bx)=0$ for some $\bx\in X$. 
\end{proof}
We now give the guarantees of our adaptive attack in \figref{fig:attack}, showing that it indeed finds a failure certificate. 
Since the proof requires several structural properties which we will prove in the next few subsections, we defer the proof of \thmref{thm:adaptive-attack} to \secref{sec:attack:final}. 
\begin{restatable}{theorem}{thmadaptiveattack}[Adaptive attack against integer linear sketches]
\thmlab{thm:adaptive-attack}
Let $8 \leq B \leq \poly(n)$ and $\alpha \geq 4 \ell_{\bA}$, where $\lambda_{\textrm{max}}(\calL^{\perp}(\bA)) \leq \ell_{\bA} = \sqrt{nM^2}$. Let $\bA \subseteq \mathbb{R}^n$ be an $r$-dimensional subspace of $\mathbb{R}^n$ where $n \geq  4r + 90 \log(Bm)$. Let $f:\bA \rightarrow \{0,1\}$ be a linear sketch for all $x \in \mathbb{Z}^n$. Then, there is an attack algorithm that, given only oracle access to $f$, finds a failure certificate for $f$ with probability $9/10$. The time and query complexity of the algorithm are bounded by $\poly(r\log n)$. Furthermore, all queries that the attack algorithm makes are sampled from $D(V^{\perp}, \sigma^2)$ for some $V \subseteq \mathbb{R}^n$ and $\sigma^2 \in [\alpha, B\alpha]$.  
\end{restatable}


\subsection{Distance Between Subspaces}
In this section, we briefly introduce the distance between subspaces, for the purposes of bounding the distance between the subspace spanned by the set of vectors identified by the attack and the closest subspace in the sketch matrix $\bA$ in \secref{sec:progress}. 
This property will in turn ultimately allow us to show an invariant that holds over the duration of the attack. 
Formally, we have the following definition for the distance between subspace.

\begin{definition}[Distance between subspaces]
We define the distance between subspaces $V,W\subseteq\mathbb{R}^n$ by
\[d(V,W)=\|P_V-P_W\|_2:=\sup_{\bv\in\mathbb{R}^n}\frac{\|P_V\bv-P_W\bv\|_2}{\|\bv\|_2}.\]
\end{definition}
We give the following structural property upper bounding the total variation distance between two subspace Gaussian families by the distance between their orthogonal subspaces, analogous to Lemma 4.14 in~\cite{HardtW13}, but for two discrete Gaussian distributions. 

\begin{lemma}
\lemlab{lem:tvd:subspaces}
For every $\sigma^2\in(\alpha,B \alpha]$,
\[\TVD(D(V^\perp,\sigma^2),D(W^\perp,\sigma^2))\le40\sqrt{ Bn\log ( Bn)}\cdot d(V,W)+\frac{1}{(Bn)^4}.\]
\end{lemma}
\begin{proof}
Sample $\bx_1 \sim D(V^{\perp}, \sigma^2)$ and $\bx_2 \sim D(W^{\perp}, \sigma^2)$ independently. 
By setting $\bA$ to be the identity matrix in \lemref{thm:dist:round:close}, we have point-wise closeness of $D(V^{\perp},\sigma^2)$ and the process of rounding a sample from $G(V^{\perp},\sigma^2)$ to its canonical point in $\mathbb{Z}^n$ up to $1\pm\frac{1}{\poly(n)}$ multiplicative factor, for sufficiently large $\alpha$. 
By applying a similar argument to $D(W^{\perp},\sigma^2)$ and $G(W^{\perp},\sigma^2)$, then it follows that
\[\TVD(D(V^{\perp}, \sigma^2), D(W^{\perp}, \sigma^2))\le\TVD(G(V^{\perp}, \sigma^2), G(W^{\perp}, \sigma^2))+\frac{1}{\poly(n)}.\]
Recall that from Lemma 4.14 of \cite{HardtW13} for $\sigma$ set to be in the range $\sigma^2 \in [\alpha, \alpha B]$, we have that 
\[\TVD(G(V^{\perp}, \sigma^2), G(W^{\perp}, \sigma^2)) \leq 20 \sqrt{Bn \log(Bn)} \cdot d(V, W) + \frac{1}{(Bn)^5}.\]
Thus, by triangle inequality, the result follows, up to a loss of $\frac{1}{\poly(n)}$ in total variation distance.
\end{proof}

\subsection{Progress Lemma}
\seclab{sec:progress}
We now show that each round of the adaptive attack makes progress toward identifying a subspace that is close to the rowspan of the sketch matrix $\bA$. 
For each $1 \leq t \leq r$, let $W_t \subseteq \bA$ be the closest $(t-1)$-dimensional subspace to $V_t$ that is contained in $\bA$. Specifically, let $W_t$ be such that 

\begin{align*}
d(V_t, W_t) = \min\{d(V_t, W): \dim(W) = t-1, W \subseteq \bA \}
\end{align*}

We identify $V_t$ with the subspace $\textrm{span}(V_t)$. 
To show the correctness of our attack, we argue that the attack algorithm maintains the following invariant with high probability: 

\begin{invariant}
\invarlab{invar:step}
At each step $t\in[r+1]$, $\dim(V_t)=t-1$ and $d(V_t,W_t)\le\frac{t}{40(Bn)^{3.5}\log^{2.5}(Bn)}$.
\end{invariant}

We note that if \invarref{invar:step} holds at step $t$, then we have an upper bounds on the total variation distance between the two subspace Gaussians. 
\begin{lemma}
\lemlab{lem:invar:tvd}
Suppose \invarref{invar:step} holds at step $t$. 
Then 
\[\TVD(D(V_t^\perp,\sigma^2),D(W_t^\perp,\sigma^2))\le\frac{1}{B^3n^2\log^2(Bn)}.\]
\end{lemma}
\begin{proof}
Under the assumption that \invarref{invar:step} holds at step $t$, then we have 
\[d(V_t,W_t)\le\frac{t}{40(Bn)^{3.5}\log^{2.5}(\alpha Bn)}\le\frac{1}{40 B^{3.5}n^{2.5}\log^{2.5}(Bn)}\]
and thus by \lemref{lem:tvd:subspaces}, 
\[\TVD(D(V_t^\perp,\sigma^2),D(W_t^\perp,\sigma^2))\le\frac{1}{B^3n^2\log^2( Bn)}.\]
\end{proof}

We next show that under \invarref{invar:step}, correctness on the subspace $V_t^\perp$ implies correctness on the subspace $W_t^\perp$ with some small loss in the probability, similar to Lemma 5.4 in \cite{HardtW13}, which gave a similar statement when inputs are generated from a continuous Gaussian distribution. 
\begin{lemma}
\lemlab{lem:correct:before:after}
Suppose \invarref{invar:step} holds at step $t$ and let $\zeta=\frac{1}{20(Bn)^2\log(Bn)}$. 
Then $f$ being $(\alpha,B)$-correct on $V_t^\perp$ implies that $f$ is $(\alpha+\zeta,B)$-correct on $W_t^\perp$. 
\end{lemma}
\begin{proof}
Under the assumption that \invarref{invar:step} holds at step $t$, then by \lemref{lem:invar:tvd}, we have 
\[\TVD(D(V_t^\perp,\sigma^2),D(W_t^\perp,\sigma^2))\le\frac{1}{B^3n^2\log^2( Bn)}.\]
Since $\zeta=\frac{1}{20(Bn)^2\log(Bn)}$, then in particular, $\TVD(D(V_t^\perp,\sigma^2), D(W_t^\perp,\sigma^2))\le\zeta$. 
Thus since $f$ is $(\alpha,B)$-correct on $V_t^\perp$, then $f$ is $(\alpha+\zeta,B)$-correct on $W_t^\perp$, since there is at most an additive $\zeta$ loss in the failure probabilities. 
\end{proof}
We next show that the empirical fraction $s(t,\sigma^2)$ is a good estimator of the true probability that $f(\bg)=1$ for a discrete Gaussian $\bg$, similar to Lemma 5.5 in \cite{HardtW13}, which considers the accuracy of this empirical estimate as a proxy for the probability that $f(\bg)=1$ for a continuous Gaussian $\bg$. 
\begin{lemma}
\lemlab{lem:event:whp}
Let $\eta>0$ and let $\calE$ denote the event that for all $t\in[r+1]$ and for all $p\in S$,
\[\bigg|s(t,\sigma^2)-\PPPr{\bg\sim D(V_t^\perp,\sigma^2)}{f(\bg)=1}\bigg|\le\frac{1}{20(Bn)^2\log(Bn)}.\]
Then $\PPr{\calE}\ge 1-\exp(-n)$. 
\end{lemma}
\begin{proof}
Let $\zeta=\frac{1}{20( Bn)^2\log( Bn)}$. 
Since the number of samples satisfies $m\gg\left(\frac{\alpha Bn}{\zeta}\right)^2$, the claim follows from standard Chernoff bounds. 
\end{proof}

Next, we show that in each round of the attack, either the algorithm terminates and outputs a failure certificate, or conditioned on the invariant holding in the next round, the function $f$ is correct on the orthogonal subspace. 
This statement is  analogous to Lemma 5.6 in~\cite{HardtW13}. 
Note that we have not yet claimed that the invariant holds in the next round: we will prove this later in the ``Progress Lemma'', c.f., \lemref{lem:progress:lem}.
\begin{lemma}
\lemlab{lem:end:or:invar}
Suppose the event $\calE$ occurs. 
Then either:
\begin{itemize}
\item 
the algorithm terminates in some round $t$ and outputs a failure certificate $D(V_t^\perp,\sigma^2)$ for $f$
\item
the algorithm does not terminate in round $t$ and if the invariant holds in round $t$, then $f$ is $(\alpha, B)$-correct on $W_t^\perp$. 
\end{itemize}
\end{lemma}
\begin{proof}
The first claim follows directly from the definition of a failure certificate and the condition that the empirical error given by $s(t,p)$ is $\zeta$-close to the actual error. 

Moreover, since the algorithm did not terminate at time $t$, we have that $f$ must be $(2\zeta,B)$-correct on $V_t^\perp$. 
Then by \lemref{lem:correct:before:after}, conditioned on the event $\calE$, we have that $f$ must be $(3\zeta,B)$-correct on $W_t^\perp$. 
Since $\zeta=\frac{1}{20( Bn)^2\log( Bn)}$, then $3\zeta\le\frac{1}{10(Bn)^2}$ for sufficiently large $n$ and so $f$ is $(\alpha, B)$-correct on $W^\perp$, as desired for the second claim.
\end{proof}
On the other hand, we show that if the attack lasts until round $r+1$, then $f$ cannot be correct on the orthogonal subspace afterwards, similar to Lemma 5.8 in \cite{HardtW13}. 
\begin{lemma}
\lemlab{lem:invar:end:wrong}
Suppose \invarref{invar:step} holds for $t=r+1$. 
Then $f$ is not $(\alpha, B)$-correct on $W_{r+1}$.
\end{lemma}
\begin{proof}
Note that $\dim(V_{r+1})=\dim(W_{r+1})=r$. 
Since $W_{r+1}\subseteq\bA$ and $\dim(\bA)=r$ then $W_{r+1}=\bA$. 
Note that $f$ cannot distinguish between $\bg$ for $\bg\sim D(W_{r+1}^\perp,2\alpha)$ and $\bg\sim D(W_{r+1}^\perp,B\alpha)$. 
Thus $f$ must be incorrect with constant probability on inputs $\bg$ generated from one of the distributions. 
\end{proof}

To ultimately prove the progress lemma in \lemref{lem:progress:lem}, we first show that the Conditional Expectation Lemma still holds for a discretization of the variance in the set $S$ of the attack, analogous to Lemma 5.9 in \cite{HardtW13}. 
\begin{lemma}
\lemlab{lem:exist:var}
Suppose $f$ is correct on $W_t^\perp$. 
Then there exists $\widetilde{\sigma}^2\in S$, $\Delta\ge\frac{\alpha}{7Br}$ and $\bu\in V^\perp_t\cap\bA$ such that for $g\sim D(V^\perp_t,\widetilde{\sigma}^2)$, we have $\PPr{f(\bg)=1}\ge\frac{1}{60B^2 r}$ and
\[\Ex{\langle\bu,\bg\rangle^2\,\mid\,f(\bg)=1}\ge\Ex{\langle\bu,\bg\rangle^2}+\Delta.\]
\end{lemma}
\begin{proof}
Since $f$ is correct on $W_t^\perp$, then by the Conditional Expectation Lemma, i.e., \lemref{lem:cond:exp}, there exists $\bu\in U=W_t^\perp\cap\bA$ and $\sigma\in\left[\alpha,\alpha B\right]$ such that
\[\Ex{\langle\bu,\bg\rangle^2\,\mid\,f(\bg)=1}\ge\Ex{\langle\bu,\bg\rangle^2}+\frac{\alpha}{4Br}\]
and $\PPr{f(\bg)=1}\ge\frac{1}{40B^2 r}$. 
Since \invarref{invar:step} holds at step $t$, then by \lemref{lem:invar:tvd}, we have that
\[\TVD(D(V_t^\perp,\sigma^2),D(W_t^\perp,\sigma^2))\le\frac{1}{B^3n^2\log^2(Bn)}.\]
Thus we have $\PPr{f(\bg)=1}\ge\frac{1}{50 B^2r}$. 
Moreover by the definition of conditioning, the total variation distance between the two distributions under the condition that $f(\bg)=1$ can increase by at most a factor of $50 B^2r$. 
Thus for any function $F:\mathbb{R}^n\to[0,M]$, we have
\[\left|\EEx{\bg\sim D(W_t^\perp,\sigma^2)}{F(\bg)\,|\,f(\bg)=1}\right|-\left|\EEx{\bg\sim D(V_t^\perp,\sigma^2)}{F(\bg)\,|\,f(\bg)=1}\right|\le\frac{M(50B^2r)}{B^3n^2\log^2(Bn)}.\]
By sub-Gaussian concentration, we have that $\PPr{\langle\bu,\bg\rangle^2>10\beta C}\le\exp(-\beta)$ for any $\beta>0$ and thus we can truncate $\langle\bu,\bg\rangle^2$ at $M=10B\log(rB)$ while only affecting either expectation by at most $o\left(\frac{\alpha}{Br}\right)$. 
Therefore, we have
\[\left|\EEx{\bg\sim D(W_t^\perp,\sigma^2)}{F(\bg)\,|\,f(\bg)=1}\right|-\left|\EEx{\bg\sim D(V_t^\perp,\sigma^2)}{F(\bg)\,|\,f(\bg)=1}\right|\le o\left(\frac{\alpha}{Br}\right).\]
Note that the same claim holds for the rounding of $\sigma^2$ to some $(\widetilde{\sigma})^2$ that is a multiple of $\zeta=\frac{1}{20(Bn)^2\log(Bn)}$ within the interval $S=[\alpha,B\alpha]\cap\zeta\mathbb{Z}$. 
Finally, note that we have $\|P_{V_t}\bu\|_2\ge 1-\frac{1}{B^2n^2}$, since $\bu\in W_t\cap\bA$ and the invariant holds for $(V_t,W_t)$. 
Hence, there exists $\bu\in V^\perp_t\cap\bA$ such that
\begin{align*}
\Ex{\langle\bu,\bg\rangle^2\,\mid\,f(\bg)=1}&\ge\Ex{\langle\bu,\bg\rangle^2}+\frac{\alpha}{4Br}-o\left(\frac{1}{Br}\right)-\frac{1}{B^2n^2}\\
&\ge\Ex{\langle\bu,\bg\rangle^2}+\Delta,  
\end{align*}
for any $\Delta\ge\frac{\alpha}{7Br}$, as desired. 
\end{proof}

We next lower bound the norm of the output vector $\bv^*$ after projecting onto $V_t^\perp\cap\bA$, adapting Lemma 5.10 in \cite{HardtW13} to handle discrete Gaussians in the computation of the maximizer $\bv^*$. 
Within the confines of the proof of \lemref{lem:found:vec}, we require a statement (\lemref{lem:top:vector}) that shows the vector we find in each round has a large projection onto the rowspan of $\bA$ orthogonal to the rows we have already found; we defer the proof of this statement to \secref{sec:attack:top:singular}. 
\begin{lemma}
\lemlab{lem:found:vec}
Let $N>0$ be sufficiently large and suppose  $\sigma_n(\Sigma_t)>\lambda_{\max}(\calL^{\perp}(\bA))^2 \cdot \frac{\ln(2n(1+1/\eps))}{\pi}$. 
Then with probability $1-\frac{1}{\poly(n)}$, the vector $\bv^*$ found in step $t$ satisfies
\[\|P_{V^\perp_t\cap\bA}\bv^*\|_2^2\ge 1-\frac{1}{200(Bn)^{3.5}\log^4(Bn)}.\]
\end{lemma}
\begin{proof}
Consider the setting of $\widetilde{\sigma}^2$ in \lemref{lem:exist:var}. 
For $g\sim D(V^\perp_t,\widetilde{\sigma}^2)$, we have that by \lemref{lem:exist:var}, $\PPr{f(\bg)=1}\ge\frac{1}{60B^2r}$. 
Consider the conditional distribution of $g\sim D(V^\perp_t,\widetilde{\sigma}^2)$, conditioned on the event that $f(\bg)=1$. 
Observe that $\bg'_1,\ldots,\bg'_{m'}$ are sampled uniformly from this distribution. 
Since $\calA$ outputs $1$ on each sample with probability at least $\frac{1}{70B^2r}$, then for $\Delta=\O{\frac{1}{Br}}$, $\xi^2=\O{B^2\log^2 n}$, and $\gamma=\O{\frac{1}{(Br)^{3.5}\log^4(Bn)}}$, we have
\[m'\ge\frac{m}{700B^2r}=\Omega(B^{11}n^{10}\log^{15}(r))\]
samples from $D(V^\perp_t,\widetilde{\sigma}^2)$ conditioned on the event that $f(\trunc_\eta(\bg))=1$. 

Define $\bv^*$ to be the first vector $\bv_\sigma=\argmax_{\bv\in\mathbb{R}^n}z(\bv)$ with $z(\bv)\ge\sigma^2+\frac{\sigma^2}{4}+\frac{1}{14Br}$ for $z(\bv)=\frac{1}{m'}\sum_{i=1}^{m'}\langle\bv,\bg'_i\rangle^2$. 
We show that the prerequisites of \lemref{lem:top:vector} apply to conclude the mass of the projection of $\bv^*$ onto $V_t^\perp$, for the setting of $\gamma=\frac{1}{(Bn)^{3.5}\log^4(Bn)}$. 
In particular, we set $V=V_t^\perp\cap\bA$, $W=V_t+\bA^\perp$, $\tau=\widetilde{\sigma}^2+\frac{\widetilde{\sigma}^2}{4}+\frac{1}{\poly(n)}$, and $\Delta$ be the parameter from \lemref{lem:exist:var}. 
\begin{enumerate}
\item 
We first show the second condition. 
Let $\bw\in W$ be an arbitrary unit vector and consider the decomposition $\bw=\alpha\bw_1+\beta\bw_2$, where $\bw_1\in V_t$ and $\bw_2'\in\bA^\perp$, orthogonal to $\bw_1$, are unit vectors. 
We have $\sigma_n(\Sigma_t)>\lambda_{\max}(\calL^{\perp}(\bA)) \cdot \sqrt{\frac{\ln(2n(1+1/\eps))}{\pi}}$. 
Hence, up to a $\frac{1}{\poly(n)}$ additive variation distance, the distribution of $\bg$ can be written as a product distribution $\bg=\bg_1+\bg_2$ for independent continuous Gaussians $\bg_1$ and $\bg_2$ in $\bA$ and $\bA^\perp$ respectively, and then rounding the subsequent value, by an argument similar to \lemref{thm:dist:round:close}.  
Therefore, we have $\Ex{\langle\bw_1,\bg\rangle\langle\bw_2,\bg\rangle}\le\frac{1}{\poly(n)}$. 
We also observe that with high probability, $|\langle\bw_1,\bg\rangle|\le\O{\sigma\log n}$ and $|\langle\bw_2,\bg\rangle|\le\O{\log n}$. 
Hence, we have $\Ex{\langle\bw_1,\bg\rangle\langle\bw_2,\bg\rangle\,\mid\,f(\bg)=1}\le\frac{1}{\poly(n)}+\O{\sigma\log^2 n}$. 
\item
We next show the first condition. 
Let $\bw\in W$ be any fixed unit vector and again consider the decomposition $\bw=\alpha\bw_1+\beta\bw_2$, where $\bw_1\in V_t$ and $\bw_2\in\bA^\perp$, orthogonal to $\bw_1$, are unit vectors. 
We seek to upper bound $\EEx{\bg}{\langle\bw_2,\bg\rangle^2\,\mid\,P_A\bg=\bz}$ for some fixed $\bz$ in the rowspan of $\bA$. 
Consider another vector $\bz'$ in the rowspan of $\bA$ such that $\left\lvert\|\bz\|_2^2-\|\bz'\|^2\right\rvert\le\O{\sigma\log n}$. 
Note that the ratio of the probability mass function assigned to $\bz$ and $\bz'$ is $\exp((\bz+\bq^\top)\bSigma^{-1}(\bz+\bq) - (\bz'+\bq)^\top \bSigma^{-1}(\bz'+\bq))$ across all vectors $\bq$ orthogonal to $\bA$. 
Therefore, since $\left\lvert\|\bz\|_2^2-\|\bz'\|_2^2\right\rvert\le\O{\sigma\log n}$, then we have that $\exp((\bz+\bq)^\top\bSigma^{-1}(\bz) - (\bz'+\bq)^\top \bSigma^{-1}(\bz'+\bq))\le\left(1+\frac{1}{\poly(n)}\right)$ since $\sigma_{\min}(\bSigma^{-1})\ge\frac{\sigma^2}{\poly(n)}\gg\sigma\log n$. 
Moreover, we have $\langle\bw,\bq+\bz\rangle^2-\langle\bw,\bq+\bz'\rangle^2\le\O{\sigma\log n}$ since $\bq$ is orthogonal to $\bw$ and has variance at most $\O{1}$. 
Given the tail bounds for sub-Gaussian vectors, we thus have that
\[\EEx{\bg}{\langle\bw_2,\bg\rangle^2\,\mid\,P_A\bg=\bz}-\EEx{\bg}{\langle\bw_2,\bg\rangle^2\,\mid\,P_A\bg=\bz'}\le\O{\sigma\log n}.\]
Therefore, we have
\begin{align*}
\EEx{\bg}{\langle\bw_2,\bg\rangle^2\,\mid\,f(\bg)=1}&=\sum_{\bz|f(\bz) = 1}\EEx{\bg}{\langle\bw_2,\bg\rangle^2\,\mid\,P_A\bg=\bz}\cdot\frac{\PPr{P_A\bg=\bz}}{\PPr{f(\bz)=1}}\\
&\le\O{\sigma\log n}+\sum_{\bz|f(\bz) = 1}\EEx{\bg}{\langle\bw_2,\bg\rangle^2}\cdot\frac{\PPr{P_A\bg=\bz}}{\PPr{f(\bz)=1}}\\
&\le\O{\sigma\log n}+\EEx{\bg}{\langle\bw_2,\bg\rangle^2}.
\end{align*}
On the other hand, we have $\EEx{\bg}{\langle\bw_2,\bg\rangle^2}\le n+\sigma^2$ by the fact that $\EEx{\bg\sim G(V_t^\perp,\sigma^2)}{\langle\bw_2,\bg\rangle^2}\le\sigma^2$ and the difference between the probability density function and probability mass function for continuous and discrete Gaussian distributions. 
By a similar argument, we have $\EEx{\bg}{\langle\bw_1,\bg\rangle^2}\ll\O{n}$, since $\bw_1$ is contained in $V_t$. 
Finally, by the above conditioned, we have that $\Ex{\langle\bw_1,\bg\rangle\langle\bw_2,\bg\rangle\,\mid\,f(\bg)=1}\le\frac{1}{\poly(n)}+\O{\sigma\log^2 n}$. 
Putting these together, we have that $\EEx{\bg}{\langle\bw,\bg\rangle^2}\le\tau$, as desired.
\item 
We next prove the third condition. 
By \lemref{lem:exist:var}, there exists $\bv\in V$ such that $\Ex{\langle\bv,\bg\rangle}^2\ge\tau+\frac{\Delta}{2}$ for $\Delta\ge\frac{1}{7Br}$. 
\item 
We next show that the fourth condition holds. 
For every unit vector $\bu\in\mathbb{R}^n$, observe that $\xi^2=\mathbb{V}(\langle\bu,\bg\rangle^2)\le\O{B^2\log^2 n}$ by standard sub-Gaussian concentration bounds.  
\item 
Finally, we show that we have sufficiently large number of samples. 
Note that $m'=\Omega(B^{11}n^{10}\log^{15}(r))=\Omega\left(\frac{n\log^2(n)\xi^2}{\gamma^2\Delta^2}\right)$, for $\Delta=\Omega\left(\frac{1}{Br}\right)$, $\gamma=\Omega\left(\frac{1}{(Br)^{3.5}\log^4(Bn)}\right)$, and $\xi^2=\O{B^2\log^2 n}$. 
\end{enumerate}
Therefore by \lemref{lem:top:vector}, we have that with probability at least $1-\exp(-n)$, both $z(\bv)\ge\sigma^2+\frac{\sigma^2}{14Br}$ and 
\[\|P_{V^\perp_t\cap\bA}\bv^*\|_2^2\ge 1-\frac{1}{200(Bn)^{3.5}\log^4(Bn)},\]
for the variance $\widetilde{\sigma^2}$ and the vector $\bv^*$ the maximizes the quantity $z(\bv)$. 

Moreover, we call a variance $\sigma^2\in S$ \emph{bad} if for every unit vector $\bv\in V_t^\perp\cap\bA$, we have for $\bg\sim D(V_t^\perp,\sigma^2)$,
\[\Ex{\langle\bu,\bg\rangle^2\,|\,f(\bg)=1}\le\Ex{\langle\bu,\bg\rangle^2}+\frac{\Delta}{20}.\]
Then by a standard concentration inequality and a union bound over a net, similar to the argument in \lemref{lem:top:vector}, we have with probability at least $1-\exp(-n)$ that for all bad $\sigma^2$ and all vectors $\bv$, $z(\bv)<\sigma^2+\frac{\sigma^2}{18Br}$. 
Therefore, the vector $\bv^*$ from $\widetilde{\sigma^2}$ will be the first vector that achieves sufficiently high objective and will be output by step $t$. 
\end{proof}

We finally show that \invarref{invar:step} holding at round $t$ implies \invarref{invar:step} holds at round $t+1$ if the algorithm does not terminate, similar to Lemma 5.11 in \cite{HardtW13}. 
This will allow us to finish the proof of the progress lemma. 
\begin{lemma}
\lemlab{lem:invar:next:step}
Let $\eta=\frac{1}{\poly(n)}$ be a fixed parameter. Suppose  $\sigma_n(\Sigma_t)>\lambda_{\max}(\calL^{\perp}(\bA))^2 \cdot \frac{\ln(2n(1+1/\eps))}{\pi}$. Suppose $(V_t,W_t)$ satisfy \invarref{invar:step} and suppose $\bv^*$ computed in round $t$ satisfies 
\[\|P_{V^\perp_t\cap\bA}\bv^*\|_2^2\ge 1-\frac{1}{200(Bn)^{3.5}\log^4(Bn)}.\]
Then \invarref{invar:step} holds for $(V_{t+1},W_{t+1})$. 
\end{lemma}
\begin{proof}
Let $\gamma=\frac{1}{20(Bn)^{3.5}\log^4(Bn)}$. 
By \lemref{lem:found:vec}, we have
\[\|P_{V^\perp_t\cap\bA}\bv^*\|_2^2\ge 1-\frac{\gamma}{10}.\]
Therefore, we have
\[\|P_{\bA}\bv^*\|_2^2\ge 1-\frac{\gamma}{10},\qquad\|P_{V_t}\bv^*\|_2^2\le\frac{\gamma}{10}.\]
Thus since $\bv_t=\trunc_{\eta}\left(\bv^*-\frac{\sum_{\bv\in V_t}\bv\langle\bv,\bv^*\rangle}{\left\|\sum_{\bv\in V_t}\bv\langle\bv,\bv^*\rangle\right\|_2}\right)$, then by the Pythagorean theorem for sufficiently small $\eta$, we have 
\[\|P_{\bA}\bv_t\|_2^2\ge 1-\frac{\gamma}{4}.\]
Thus we have that $P_{W_{t+1}}=P_{W_t}+\bw_t\bw_t^\top$ for some unit vector $\bw_t$ orthogonal to $W_t$ such that $\|\bv_t-\bw_t\|_2\le\frac{\gamma}{4}$. 

Observe that $P_{V_{t+1}}=P_{V_t}+\bv_t\bv_t^\top$. 
Therefore, by the triangle inequality, we have
\begin{align*}
d(V_{t+1},W_{t+1})&=\|P_{V_{t+1}}-P_{W_{t+1}}\|_2\\
&\le\|P_{V_t}-P_{W_t}\|_2+\|\bv_t\bv_t^\top-\bw_t\bw_t^\top\|_2\\
&\le d(V_t,W_t)+\|\bv_t\bv_t^\top-\bv_t\bw_t^\top\|_2+\|\bv_t\bw_t^\top-\bw_t\bw_t^\top\|_2.
\end{align*}
By submultiplicativity of norms and the fact that $\bv_t$ and $\bw_t$ are unit vectors, we thus have
\begin{align*}
d(V_{t+1},W_{t+1})&\le d(V_t,W_t)+\|\bv_t\|_2\cdot\|\bv_t^\top-\bw_t^\top\|_2+\|\bw_t^\top\|_2\cdot\|\bv_t-\bw_t\|_2\\
&\le d(V_t,W_t)+\frac{\gamma}{2}.
\end{align*}
By the assumption of \invarref{invar:step}, $d(V_t,W_t)\le\frac{t}{20(Bn)^{3.5}\log^{2.5}(Bn)}$.
Therefore,
\[d(V_{t+1},W_{t+1})\le\frac{t}{20(Bn)^{3.5}\log^{2.5}(Bn)}+\frac{1}{40(Bn)^{3.5}\log^4(Bn)}\le\frac{t+1}{20(Bn)^{3.5}\log^{2.5}(Bn)},\]
and so the invariant holds for $(V_{t+1},W_{t+1})$.
\end{proof}

We now put everything together to achieve the progress lemma, analogous to Lemma 5.7 in \cite{HardtW13}, which allows entire real vectors to be given as input to $f$, rather than truncated vectors. 
\begin{lemma}[Progress lemma]
\lemlab{lem:progress:lem}
Let $\eta=\frac{1}{\poly(n)}$ be a fixed parameter. Suppose  $\sigma_n(\Sigma_t)>\lambda_{\max}(\calL^{\perp}(\bA))^2 \cdot \frac{\ln(2n(1+1/\eps))}{\pi}$. Let $t\in[r]$ and suppose \invarref{invar:step} holds in round $t$ and $f$ is $B$-correct on $W^\perp_t$. 
Then with probability at least $1-\frac{1}{n^2}$, \invarref{invar:step} holds in round $t+1$. 
\end{lemma}
\begin{proof}
The proof follows immediately from \lemref{lem:exist:var}, \lemref{lem:found:vec}, and \lemref{lem:invar:next:step}. 
\end{proof}

\subsection{Top Right Singular Vector of Biased Discrete Gaussian Matrices}
\seclab{sec:attack:top:singular}
In this section, we prove \lemref{lem:top:vector}, which states that the matrix consisting of all of the vectors that have a slight correlation with the rowspan of $\bA$ has a top right singular vector that is significantly correlated with the rowspan of $\bA$. 
Recall that \lemref{lem:top:vector} was used in the proof of \lemref{lem:found:vec}, which ultimately culminated in the Progress Lemma in \lemref{lem:progress:lem}. 
\begin{theorem}
\thmlab{thm:chernoff:hoeffding}
For independent random variables $X_1,\ldots,X_m$, let $X=X_1+\ldots+X_m$ and $\xi^2=\mathbb{V}[X]$. 
Then for any $t>0$,
\[\PPr{\left\lvert X-\mathbb{E}[X]\right\rvert>t}\le\exp\left(-\frac{t^2}{4\xi^2}\right).\]
\end{theorem}
We prove the following statement about the top right singular vector of biased discrete Gaussian matrices. 
In the context of \lemref{lem:found:vec}, the subspace $V$ will correspond to the subspace of $\bA$ that has \emph{not} been found by the adversarial attack so far, the value $\tau$ shall correspond to the expected length $\EEx{\bg\sim D}{\langle\bw,\bg\rangle^2}$, and $\Delta$ will correspond to the gap in the Conditional Expectation Lemma. 
\begin{lemma}
\lemlab{lem:top:vector}
Let $\tau\ge 0$ be a threshold, $\eta$ be a negligible term, and $V$ be a subspace of $\mathbb{R}^n$. 
Let $D$ be a distribution over $\left(\eta\cdot\mathbb{Z}\right)^n$ such that for $\bg\sim D$,
\begin{enumerate}
\item
For every unit vector $\bw\in V^\perp$, we have $\mathbb{E}_{\bg\sim D}\langle\bw,\bg\rangle^2\le\tau$.
\item
For every two unit vectors $\bv\in V$ and $\bw\in V^\perp$, we have $\left\lvert\mathbb{E}_{\bg\sim D}\langle\bv,\bg\rangle\cdot\langle\bw,\bg\rangle\right\rvert\le\eta$.
\item
There exists a unit vector $\bv\in V\cap\left(\eta\cdot\mathbb{Z}\right)^n$ such that $\mathbb{E}_{\bg\sim D}\langle\bv,\bg\rangle^2\ge\tau+\Delta$, for some $\Delta>\frac{1}{\poly(n)}$.
\item
For every unit vector $\bu\in\mathbb{R}^n$, we have $\mathbb{V}_{\bg\sim D}\langle\bu,\bg\rangle^2\le\xi^2$.
\end{enumerate}
Let $\gamma\in\left(\frac{1}{\poly(n)},\frac{1}{2n}\right)$ and $m=\O{\frac{n\log^2(n)\xi^2}{\gamma^2\Delta^2}}$. 
Let $\bg_1,\ldots,\bg_m\sim D$ and $\bu^*=\arg\max_{\|\bu\|_2\le 1+n\eta}\sum_{i=1}^m\langle\bg_i,\bu\rangle^2$.
Then with probability $1-\exp(-n\log^2 n)$, we have 
\[\|P_V\bu^*\|_2^2\ge1-\gamma,\qquad\frac{1}{m}\sum_{i=1}^m\langle\bg_i,\bu^*\rangle^2\ge\tau+\frac{\Delta}{2}.\]
\end{lemma}
\begin{proof}
Consider the vector $\bv\in V\cap\left(\eta\cdot\mathbb{Z}\right)^n$ with $|\|\bv\|_2^2-1|\le n\eta$ such that $\mathbb{E}_{\bg\sim D}\langle\bv,\bg\rangle^2\ge\tau-n^2\eta+\Delta$, for some $\Delta>\frac{1}{\poly(n)}$, as guaranteed by our assumptions. 
Let $X=\sum_{i\in[m]}\langle\bv,\bg_i\rangle^2$, so that $\Ex{X}\ge\tau m+n^2\eta m+\Delta m$ and $\mathbb{V}_{\bg\sim D}[X]\le m\xi^2$. 
By \thmref{thm:chernoff:hoeffding} and the setting of $\eta\ll\frac{\tau}{n^2}$, 
\[\PPr{X\le\left(\tau-n^2\eta+\Delta\right)m-\frac{\gamma m\Delta}{4}}\le\exp\left(-\frac{\Delta^2\gamma^2 m^2}{\O{\xi^2 m}}\right)\le\exp(-\Omega(n\log^2 n)).\]
Hence, there exists a near-unit vector that is well-aligned with $V$ and with high probability, will have ``large'' sum of squared dot products with the random vectors $\bg_1,\ldots,\bg_m$. 

We now show that vectors that are not well-aligned with $V$ do not have large sum of squared dot products with the random vectors $\bg_1,\ldots,\bg_m$. 
Let $\bu=\alpha\bv+\beta\bw$ with be any truncated unit vector so that $\alpha^2+\beta^2\le1+n\eta$, where $\bv\in V\cap\left(\eta\cdot\mathbb{Z}\right)^n$, $\bw\in V^\perp\cap\left(\eta\cdot\mathbb{Z}\right)^n$, and both $|\|\bv\|_2^2-1|\le n\eta$ and $|\|\bw\|_2^2-1|\le n\eta$. 
Moreover, we consider the case where $\alpha^2<1-\gamma$ and set $Y=\sum_{i=1}^m\langle\bu,\bg_i\rangle^2$. 
Then we have
\[\mathbb{E}_{\bg\sim D}\langle\bu,\bg\rangle^2=\alpha^2\cdot\mathbb{E}_{\bg\sim D}\langle\bv,\bg\rangle^2+\beta^2\cdot\mathbb{E}_{\bg\sim D}\langle\bw,\bg\rangle^2+2\alpha\beta\cdot\mathbb{E}_{\bg\sim D}\langle\bv,\bg\rangle\cdot\langle\bw,\bg\rangle.\]
Since $\left\lvert\mathbb{E}_{\bg\sim D}\langle\bv,\bg\rangle\cdot\langle\bw,\bg\rangle\right\rvert\le n^2\eta$ by assumption, then
\begin{align*}
\mathbb{E}_{\bg\sim D}\langle\bu,\bg\rangle^2&=\alpha^2\cdot\mathbb{E}_{\bg\sim D}\langle\bv,\bg\rangle^2+\beta^2\cdot\mathbb{E}_{\bg\sim D}\langle\bw,\bg\rangle^2+4n^2\eta\\
&\le(1-\gamma)(\tau+n^2\eta+\Delta)+\frac{\tau+n^2\eta}{2n}\\
&\le\tau+(1-\gamma)\Delta.   
\end{align*}
Thus, we have $\mathbb{E}[Y]\le(\tau+(1-\gamma)\Delta)m\le(\tau+\Delta) m$, so that by \thmref{thm:chernoff:hoeffding},
\[\PPr{Y\ge\left(\tau+\Delta\right)m-\frac{3\gamma m\Delta}{4}}\le\exp\left(-\frac{\Delta^2\gamma^2 m^2}{\O{\xi^2 m}}\right)\le\exp(-\Omega(n\log^2 n)).\]
In particular, since $\eta n^2\ll\gamma m\Delta$, there exists a gap of $\frac{\gamma m\Delta}{4}$ between the high-probability lower bound on $X$ and the high-probability upper bound on $Y$. 
Let $M=\{\bu:\|P_V\bu^*\|_2^2\ge1-\gamma\}$ and let $N$ be a $\frac{\gamma\Delta}{8}$-net of $M$, so that $|N|\le\exp(\O{n\log n})$. 
Then by a union bound, we have
\[\PPr{\max_{\bu\in N}\sum_{i=1}^m\langle\bu,\bg_i\rangle^2\ge\left(\tau+\Delta\right)m-\frac{3\gamma m\Delta}{4}}\le|N|\cdot\exp(-\Omega(n\log^2 n))\le\exp(-n\log^2 n),\]
where the last inequality holds for a sufficiently large constant in the number of samples $m$. 
Moreover, we have
\[\max_{\bu\in M}\sum_{i=1}^m\langle\bu,\bg_i\rangle^2\le\max_{\bu\in N}\sum_{i=1}^m\langle\bu,\bg_i\rangle^2+\frac{\gamma\Delta m}{8},\]
since each summand can differ from the closest point in the net by at most $\frac{\gamma\Delta m}{8}$ and there are $m$ summands. 
Thus it follows that with probability $1-\exp(-n\log^2 n)$, we have 
\[\|P_V\bu^*\|_2^2\ge1-\gamma,\qquad\frac{1}{m}\sum_{i=1}^m\langle\bg_i,\bu^*\rangle^2\ge\tau+\frac{\Delta}{2}.\]
\end{proof}
 
\subsection{Putting It All Together}
\seclab{sec:attack:final}
We conclude with the proof of \thmref{thm:adaptive-attack}, which summarizes the guarantees of our adaptive attack. 
\thmadaptiveattack*
\begin{proof}
First, for each row $\bA_k$ of the pre-processed $\bA$, we consider writing each entry in its binary representation using $\O{\log n}$ bits, and replace $\bA_k$ by rows $B_i = [b_{1i},\ldots, b_{ni}]$ for each bit-significance $i \in [\O{\log n}]$ -- this is without loss of generality, since $\calA$ can still recover $\sum_{i = 0}^{\O{\log n}} 10^i \cdot B_i = \bA_k$. Next, we apply the pre-processing in \lemref{lem:preprocessing} to the sketching matrix $\bA$ without loss of generality, as adding additional rows can only make the sketch stronger. 
At this point, $\bA$ has $m \leq 4r \log n$ rows and $\lambda_{\max}(\calL^{\perp}(\bA)) \leq \sqrt{n}$. 
Without loss of generality, we assume that $n = 4m + 90 \log(Bm)$ by working with the first $4m + 90 \log(Bm)$ coordinates of $\mathbb{R}^n$; this way, a polynomial dependence on $n$ in our algorithm is also a polynomial dependence on $r\log n$.

Now, we verify that the invariant $\sigma_n(\Sigma_{\sigma^2}^t) \geq \lambda_{\max}(\calL^{\perp}(\bA))^2  \cdot \frac{\ln(2n(1+1/\eps))}{\pi}$ is satisfied at each step of the attack. 
Let $\Sigma^{t}_{\sigma^2} = \frac{3\sigma^2}{4}P^{\top}_{V^{\perp}}P_{V^{\perp}} + \frac{\sigma^2}{4} \cdot I_n$ be the covariance matrix in round $t$ of the attack. 
Since $\sigma^2 \in [\alpha, B \alpha]$ in the attack, it follows that $\sigma_n(\Sigma_{\sigma^2}^{t}) \geq \frac{\alpha}{4} > \lambda_{\max}(\calL^{\perp}(\bA))^2  \cdot \frac{\ln(2n(1+1/\eps))}{\pi}$, directly by construction.

For each $t\in[r]$, let $W_t\subseteq\bA$ be the closest $(t-1)$-dimensional subspace to $V_t$ contained in $\bA$, so that
\[W_t=\argmin\{d(V_t,W)\,\mid\,\dim(W)=t-1, W\subseteq\bA\}.\]
We claim that \invarref{invar:step} holds at all steps in the attack. 
Observe that \invarref{invar:step} vacuously holds for $t=1$ since $V_0=\{0\}\subseteq\bA$. 
We define $\calE$ to be the event that the empirical estimate $s(t,\sigma^2)$ is accurate at all steps of the algorithm. 
By \lemref{lem:event:whp}, we have that $\calE$ occurs with probability $1-\exp(-n)$. 

Now for the purposes of induction, suppose that \invarref{invar:step} holds for a time $t-1$. 
By \lemref{lem:end:or:invar}, then if the algorithm terminates in some round $t$, then it outputs a failure certificate $G(V_t^\perp,\sigma^2)$ for $f$. 
Otherwise by \lemref{lem:end:or:invar}, if the algorithm does not terminate in round $t$ and if the invariant holds in round $t$, then $f$ is $B$-correct on $W_t^\perp$. 
By the progress lemma, we have that \invarref{invar:step} holds for time $t+1$ with probability at least $1-\frac{1}{n^2}$. 

Thus by a union bound, either the algorithm finds a failure certificate in some round $t$ or \invarref{invar:step} holds at round $r+1$ with probability at least $1-\frac{1}{n}$. 
By \lemref{lem:invar:end:wrong}, $W_{r+1}$ is not correct for $f$. 
Hence by \lemref{lem:end:or:invar}, the algorithm must terminate in round $r+1$ and output a failure certificate with probability $1-\exp(-n)$. 
Therefore, with probability at least $1-\frac{2}{n}$, the algorithm finds a failure certificate for $f$. 

Finally, observe that the query complexity is polynomially bounded in $\alpha$, $n$, and $B$, where $\alpha \geq \max_{\bA \in \mathbb{Z}^{r \times n}}(\ell_{\bA})$ for $\ell_{\bA} = \poly(n)$ is an upper bound on $\lambda_{n-r}(\calL^{\perp}(\bA))$ for any $\bA$ after applying the pre-processing in \lemref{lem:preprocessing}. By the arguments above, it follows that the attack makes $\poly(\alpha, B, m) = \poly(r \log n)$ queries.
In terms of runtime complexity, the vector that maximizes $z(\bv)$ can be found using singular vector computation, which also uses $\poly(r\log n)$ time.
\end{proof}

\section{Applications to Sketching Lower Bounds}
In this section, we will show that our lifting framework in \thmref{thm:lifting} can be used to obtain lower bounds on the sketching dimension for discrete inputs.
We note that for most of the problems, we sample our input vector from a discrete Gaussian distribution, whose entries are bounded by $\poly(n)$ with probability $1-\exp(-cn)$.

As mentioned in \lemref{lem:preprocessing}, without loss of generality, we assume our sketching matrix $\bA \in \mathbb{Z}^{r \times n}$ satisfies
$\lambda_{n - r}(\calL^\perp(\bA)) \le \poly(n)$, as the pre-processing procedure adds $\O{r}$ rows to $\bA$ and thus can only make the sketch more powerful. With this in mind, the lower bounds in this section proceed as follows: for each problem $\calP$, we choose the hard distribution for the linear sketch, i.e. to distinguish between $\bA \bx_1$ and $\bA \bx_2$ for $\bx_1, \bx_2$ drawn from appropriate discrete Gaussian distributions. We note that the hard distribution for $\calP$ will be precisely the \textit{discrete analogue} of the hard distribution from the lower bound for $\calP$ over the reals. Then, by invoking \thmref{thm:lifting}, we show that since $\calP(\bx_1)$ and $\calP(\bx_2)$ have very different values with high probability, the linear sketch $\bA$ must use at least as many rows as in the continuous case.

\subsection{\texorpdfstring{$L_p$}{Lp} Estimation, \texorpdfstring{$p\in[1,2]$}{p in [1,2]}}
\seclab{sec:lp:small:lb}
Recall that in the $L_p$ estimation problem, the input is a vector $\bx\in\mathbb{R}^n$ and an approximation parameter $\eps\in(0,1)$, and to goal is to output a multiplicative $(1+\eps)$-approximation to $\|\bx\|_p=\left(x_1^p+\ldots+x_n^p\right)^{1/p}$. 

In this section, we consider the case $p\le 2$. 
Let $n=\Theta\left(\frac{1}{\eps^2}\log\frac{1}{\delta}\right)$ for $\delta\ge\frac{1}{\poly(n)}$ and define $\tau:=\EEx{\bu\sim\calN(0,N^2\cdot I_n)}{\|\bu\|_p}$ for a sufficiently large parameter $N>0$. 
We define $f:\mathbb{R}^n\to\{0,1,\bot\}$ so that 
\[f(\bx)=\begin{cases}
0,\qquad&\text{if } \|\bx\|_p\le(1+\eps)\tau\\
1,\qquad&\text{if } \|\bx\|_p\ge(1+3\eps)\tau\\
\bot,\qquad&\text{otherwise}.
\end{cases}\] 
We claim that $f(\bu)=0$ with probability at least $1-\frac{\delta}{4}$ for $\bu\sim\calN(0,N^2\cdot I_n)$ and $f(\bv)=1$ with probability at least $1-\frac{\delta}{4}$ for $\bu\sim\calN(0,N^2\cdot I_n)$ for $\bv\sim\calN(0,(1+4\eps)^2\cdot N^2\cdot I_n)$. 

We first recall the following concentration inequality for Lipschitz functions on vectors with entries that are independently generated from sub-Gaussian distributions. 

\begin{theorem}
\thmlab{thm:lipschitz:conc}
For each $i\in[n]$, let $x_i$ be an independent, mean-zero, sub-Gaussian random variable such that $\PPr{|x_i|\ge t}\le2\exp\left(-\frac{t^2}{2\sigma^2}\right)$ for all $t>0$. 
If $f:\mathbb{R}^n\to\mathbb{R}$ is a Lipschitz function with constant $L$, then
\[\PPr{\left\lvert f(\bx)-\Ex{f(\bx)}\right\rvert\ge t}\le 2\exp\left(-\frac{t^2}{2\sigma^2 L^2}\right).\]
\end{theorem}
We now give a lower bound for the sketching dimension of any integer sketch for $L_p$ estimation on turnstile streams. 
\begin{lemma}
\lemlab{lem:lp:small:lb}
Given an approximation parameter $\eps\in(0,1)$ and $p\in[1,2]$, any integer sketch that outputs a $(1+\eps)$-approximation for $L_p$ estimation with probability at least $1-\delta$ uses sketching dimension $\Omega\left(\frac{1}{\eps^2}\log\frac{1}{\delta}\right)$. 
\end{lemma}
\begin{proof}
We first set $n=\Theta\left(\frac{1}{\eps^2}\log\frac{1}{\delta}\right)$ to have a sufficiently large constant and then set $N$ to be a sufficiently large polynomial of $n$. 
Since $\EEx{x\sim\calN(0,N^2)}{|x|^p}=\Theta(N^p)$, then we have $\Ex{\|\bx\|_p^p}=\Theta(N^pn)$. 
By Jensen's inequality, we thus have $\Ex{\|\bx\|_p}=\O{Nn^{1/p}}$. 
Furthermore, note that from a standard calculation, $\PPPr{x\sim\calN(0,N^2)}{|x|^p}\le\frac{N}{1000}$ with probability at most $\frac{1}{3}$. 
For $i\in[n]$, let $Y_i$ be the indicator random variable such that $Y_i=1$ if $|x_i|^p\le\frac{N}{1000}$ and $Y_i=0$ otherwise, so that $\PPr{Y_i=1}\le\frac{1}{3}$. 
Then $\PPr{\|\bx\|_p^p\le\frac{Nn}{2000}}\le\PPr{Y_1+\ldots+Y_n\ge\frac{n}{2}}$. 
Thus by standard Chernoff bounds, we have that $\|\bx\|_p^p\ge\frac{Nn}{2000}$ with probability $1-\frac{1}{\poly(n)}$, so that $\Ex{\|\bx\|_p}=\Omega(Nn^{1/p})$. 

By triangle inequality, we have
\[\left\lvert\|\bx\|_p-\|\by\|_p\right\rvert\le\|\bx-\by\|_p\le n^{1/p-1/2}\cdot\|\bx-\by\|_2,\]
so that the $L_p$ norm function is $n^{1/p-1/2}$-Lipschitz with respect to the Euclidean norm. 
By definition, we have $\tau=\EEx{\bx\sim\calN(0,N^2\cdot I_n)}{\|\bx\|_p}$. 
Thus by Lipschitz concentration of sub-Gaussian random variables, c.f., \thmref{thm:lipschitz:conc}, we have
\[\PPr{\left\lvert\|\bx\|_p-\tau\right\rvert\ge\eps\tau}\le 2\exp\left(-\frac{\eps^2\tau^2}{2N(n^{2/p-1})}\right).\]
Since $\tau=\Theta(Nn^{1/p})$, we have $\PPPr{\bx\sim\calN(0,N^2\cdot I_n)}{\|\bx\|_p>(1+\eps)\tau}\le\frac{\delta}{2}$. 
Similarly, we have 
\[\PPPr{\bx\sim\calN(0,(1+4\eps)^2N^2\cdot I_n)}{\|\bx\|_p<(1+3\eps)\tau}\le\frac{\delta}{2}.\] 

Next, we can note that from the probability mass function of the discrete Gaussian $\calD(0,N^2\cdot I_n)$ and the probability density function $\calN(0,N^2\cdot I_n)$, we have 
\[\EEx{\bx\sim\calD(0,N^2\cdot I_n)}{\|\bx\|_p}=\EEx{\bx\sim\calN(0,N^2\cdot I_n)}{\|\bx\|_p}+\frac{1}{\poly(n)},\]
for sufficiently large $N$. 
We thus note that by \lemref{lem:disc:subgauss} and the same Lipschitz concentration of sub-Gaussian random variables, c.f., \thmref{thm:lipschitz:conc}, that $
\PPPr{\bx\sim\mathcal{D}(0,N^2\cdot I_n)}{\|\bx\|_p>(1+\eps)\tau}\le\frac{\delta}{2}$. 
Similarly, we have for sufficiently large $n$, $\PPPr{\bx\sim\mathcal{D}(0,(1+4\eps)^2N^2\cdot I_n)}{\|\bx\|_p<(1+3\eps)\tau}\le\frac{\delta}{2}$. 

Finally, since Section 2 in \cite{GangulyW18} has proven an $\Omega\left(\frac{1}{\eps^2} \log(1/\delta)\right)$ sketching dimension lower bound to distinguish the two Gaussian distribution $\bu\sim\calN(0,N^2\cdot I_n)$ for $\bv\sim\calN(0,(1+4\eps)^2\cdot N^2\cdot I_n$, by~\thmref{thm:lifting} with sketching matrix $\bA$, function $f$ and covariance matrix $\bS = N^2 I_n$ and $(1 + 4\eps)^2 N^2 I_n$, we get the desired lower bound.
\end{proof}

\subsection{\texorpdfstring{$L_p$}{Lp} Estimation, \texorpdfstring{$p>2$}{p>2}}
\seclab{sec:lp:big:lb}
In this section, we consider the $L_p$ estimation problem as in \secref{sec:lp:small:lb}, but for the case $p>2$, rather than $p\in[1,2]$. 
We define $E_n=\EEx{\bx\in\sim\calN(0,I_n)}{\|\bx\|_p}$. 
By standard concentration of sub-Gaussian random variables, we have $E_n=\Theta(n^{1/p})$. 

We consider the following construction, which is inspired by the work of~\cite{ANPW13, GangulyW18}. We first draw $\by\sim\calG(0,N^2\cdot I_n)$.  
We then sample a vector $\bx\sim\calG(0,N^2\cdot I_n)$ and let $T$ be a sample of $t:=\log_3\frac{1}{\sqrt{\delta}}$ coordinates of $[n]$. 
We then set $\bz=\bx+\sum_{i\in[T]}\frac{C\eps^{1/p} N\cdot E_{n-t}}{t^{1/p}}\cdot\be_i$.  ~

\begin{lemma}[Appendix D.2.2 in \cite{GangulyW18}]
With probability at least $1-\frac{\delta}{2}$, we have 
\begin{align*}
\|\by\|_p\le (1+2\eps)N\cdot E_{n-t}^p,\qquad\|\bz\|_p\ge(1+4\eps)N\cdot E_{n-t}^p. 
\end{align*}
\end{lemma}

We now claim the same bounds hold if we use discrete Gaussians instead of continuous Gaussians. 

\begin{lemma}
\lemlab{lem:lp:large:lb}
Given an approximation parameter $\eps\in(0,1)$ and $p>2$, any integer sketch that outputs a $(1+\eps)$-approximation for $L_p$ estimation with probability at least $1-\delta$ uses sketching dimension $\Omega\left(n^{1-2/p}\eps^{-2/p}\log n\log^{2/p}\frac{1}{\delta}\right)$. 
\end{lemma}
\begin{proof}
We first define distribution $\calD_1:=\calD(0,N^2\cdot I_n)$. 
Consider $\by\sim\calD_1$ and let $T$ be a sample of $t:=\log_3\frac{1}{\sqrt{\delta}}$ coordinates of $[n]$. 
We define $\bx$ to be $\by$ restricted to the coordinates of $T$ and $\overline{\by}$ to be the vector $\by$ with the coordinates in $T$ set to zero, so that $\by=\overline{\by}+\bx$. 
By the triangle inequality, we have $\|\by\|_p\le\|\overline{\by}\|_p+\|\bx\|_p$. 
By concentration of Lipschitz functions for sub-Gaussian random variables, we have $\|\overline{\by}\|_p\le\Ex{\|\overline{\by}\|_p}+10N\sqrt{\log\frac{1}{\delta}}$ with probability at least $1-\frac{\delta}{10}$. 
Similarly, we have $\|\bx\|_p\le\Ex{\|\bx\|_p}+10N\sqrt{\log\frac{1}{\delta}}$. 
Observe that for sufficiently large $N>0$, we have $\Ex{\|\overline{\by}\|_p}=\EEx{\bv\sim\calG(0,N^2\cdot I_{n-t})}{\|\bv\|_p}\pm\frac{1}{\poly(n)}$ and $\Ex{\|\bx\|_p}=\EEx{\bv\sim\calG(0,N^2\cdot I_t)}{\|\bv\|_p}\pm\frac{1}{\poly(n)}$. 
Since we have $\sqrt{\log\frac{1}{\delta}}\le\sqrt{C_1}\eps n^{1/p}$ for a sufficiently small constant $C_1>0$ and $E_{n-t}=\Theta(n^{1/p})$, then we have
\[\|\by\|_p^p\le (1+2\eps)\cdot N^p\cdot E_{n-t}^p,\]
with probability at least $1-\frac{\delta}{5}$. 

We similarly define distribution $\calD_2$ as follows. 
Sample a vector $\bx\sim\calG(0,N^2\cdot I_n)$ and again let $T$ be a sample of $t:=\log_3\frac{1}{\sqrt{\delta}}$ coordinates of $[n]$. 
We then set $\bz=\bx+\sum_{i\in[T]}\frac{C\eps^{1/p} N\cdot E_{n-t}}{t^{1/p}}\cdot\be_i$. 
Let $\calE$ be the event that $|x_i|\le\frac{\eps^{1/p}n^{1/p}N}{pt^{1/p}}$ for all $i\in T$. 
Then for $\log\frac{1}{\delta}\le\frac{C\eps^{2/p}n^{2/p}}{\log^{2/p}\frac{1}{\delta}}$, we have
\[\PPr{\calE}\ge1-t\cdot2\exp\left(-\frac{\eps^{2/p}n^{2/p}}{2p^2t^{2/p}}\right)\ge1-\frac{\delta}{10}.\]
Let $\overline{\bx}$ be the vector $\bx$ with the indices in $T$ set to zero. 
Conditioned on $\calE$, we have 
\[\|\bz\|_p^p\ge t\cdot\frac{\eps\cdot C_2\cdot E_{n-t}^p\cdot N^p}{t}\cdot\left(1-\frac{1}{p}\right)^p+\|\overline{\bx}\|_p^p.\]
By concentration of Lipschitz functions for sub-Gaussian random variables, we have $\|\overline{\by}\|_p\ge\Ex{\|\overline{\by}\|_p}-10N\sqrt{\log\frac{1}{\delta}}$ with probability at least $1-\frac{\delta}{10}$. 
Since we have $\sqrt{\log\frac{1}{\delta}}\le\sqrt{C_1}\eps n^{1/p}$ for a sufficiently small constant $C_1>0$ and $E_{n-t}=\Theta(n^{1/p})$, then we have for a sufficiently large constant $C_2$,
\begin{align*}
\|\bz\|_p^p&\ge\eps\cdot C_2\cdot N^p\cdot E_{n-t}^p\cdot\left(1-\frac{1}{p}\right)^p+\left(1-\frac{\eps}{p}\right)^p\cdot E_{n-t}^p\\
&\ge(1+4\eps)\cdot N^p\cdot E_{n-t}^p
\end{align*}
with probability at least $1-\frac{\delta}{5}$. 

Define $f:\mathbb{R}^n\to\{0,1,\bot\}$ so that 
\[f(\bx)=\begin{cases}
0,\qquad&\text{if } \|\bx\|_p\le (1+2\eps)N\cdot E_{n-t}^p\\
1,\qquad&\text{if } \|\bx\|_p\ge (1+4\eps)N\cdot E_{n-t}^p. \\
\bot,\qquad&\text{otherwise}.
\end{cases}\] 
since it has been shown an $\Omega\left(n^{1-2/p}\eps^{-2/p}\log n\log^{2/p}\frac{1}{\delta}\right)$ dimension lower bound to distinguish the two distributions for the continuous case~\cite{GangulyW18}, by~\thmref{thm:lifting} with sketching matrix $\bA$, function $f$ and covariance matrix $\bS = N^2 I_n$, we get the desired lower bound.

\end{proof}

\subsection{Operator and Ky Fan Norm}\seclab{sec:op-ky}
In this section, we consider the problems of estimating the operator norm and the Ky Fan norm. 
Recall that in both problems, the input is a matrix $\bX$. 
In the operator norm estimation problem, the goal is to estimate $\|\bX\|_{op}=\max_{\bv}\frac{\|\bX\bv\|_2}{\|\bv\|_2}$, while in the Ky Fan norm estimation problem with parameter $s$, the goal is to estimate the sum of the top $s$ singular values, i.e., $\sigma_1(\bX)+\ldots+\sigma_s(\bX)$, where the singular values of $\bX$ satisfy $\sigma_1(\bX)\ge\sigma_2(\bX)\ge\ldots$. 

One difficulty in immediately applying our previous results to prove a lower bound for the operator norm is that the hard distributions are not both Gaussians. 
Rather, one distribution is a Gaussian matrix while the other distribution is a Gaussian matrix plus a ``spike'', consisting of a product of two Gaussian vectors.  
We thus need to first prove a number of structural properties about these distributions when the inputs are discrete Gaussians instead of continuous Gaussians. 

\begin{lemma}[Analog of Lemma 3 in~\cite{LiW16}]
\lemlab{lem:exp:exp:xtay}
Suppose that $x \sim \calD(0, \sigma^2 \cdot I_m)$ and $y \sim \calD(0, \sigma^2 \cdot I_n)$ are independent and $\bA \in \mathbb{R}^{m \times n}$ satisfies $\|\bA \|_F < 1$. Then, it holds that $\mathbb{E}_{x,y} \left[e^{x^{\top} A y/\sigma^2} \right] \leq \frac{1}{\sqrt{1- \|\bA \|_F^2}} \cdot \left( 1 + \frac{1}{\poly(n)} \right) $. 
\end{lemma}
\begin{proof}
 Let us first look at the simple case, when $x$ and $y$ are sampled independently from the continuous Gaussian $\calN(0, \sigma^2)$ and $a < 1$. In the proof of the general statement for $x,y \sim \calD(0, \sigma^2 I_n)$, we will eventually reduce to this case.
\begin{align*}
    \mathbb{E}_{x,y \sim \calN(0, \sigma^2)} \left[e^{axy/\sigma^2} \right] &=   \frac{1}{2\pi \sigma^2} \int_{y} \int_{x}  e^{\frac{axy}{\sigma^2} - \frac{x^2 + y^2}{2\sigma^2}} dx \  dy \\ &= \frac{1}{2\pi\sigma^2} \int_{y} \int_{x} e^{-\frac{1}{2\sigma^2}(x-ay)^2} \cdot e^{-\frac{1}{2\sigma^2}(1-a^2)y^2} dx \ dy \\ &= \frac{1}{\sigma\sqrt{2\pi}}\int_{y} e^{-\frac{1}{2\sigma^2} (1-a^2) y^2} dy \\ &= \frac{1}{\sqrt{1-a^2}}
\end{align*}

Now, let us consider the general case. Without loss of generality, we can assume that $m \geq n$. We consider the singular value decomposition of $\bA$, given by $\bA = U\Sigma V^{\top}$, where $U$ and $V$ are orthogonal matrices of dimension $m$ and $n$, respectively, and $\Sigma = \textrm{diag}(\sigma_1,\ldots, \sigma_n)$, where $\sigma_1,\ldots, \sigma_n$ are the non-zero singular values of $\bA$. Also, since we assumed that $\|\bA \|_F < 1$, we know that $\sigma_i \in [0,1)$ for all $i \in [n]$. Let $x \sim \calD'(0, \sigma^2 \cdot I_n)$ represent a sample from the discrete Gaussian distribution which is generated by first sampling $x \sim \mathcal{N}(0, \sigma^2 \cdot I_n)$ and then rounding $x$ to the nearest integer. Let $q'$ be the probability mass function of $\calD'(0, \sigma^2 I_n)$, and let $p$ be the probability mass function for $\calD(0, \sigma^2 I_n)$. Then, we can apply \lemref{lem:pmf} and compute as follows:
\begin{align*}
    \mathbb{E}_{x,y \sim \calD(0, \sigma^2 \cdot I_n)} \left(e^{x^{\top} \bA y/\sigma^2}\right) &= \mathbb{E}_{x,y \sim \calD(0, \sigma^2 \cdot I_n)} \left(e^{\frac{x^{\top} U \Sigma V^{\top} y}{\sigma^2}}\right) \\ &= \sum_{x \in \mathbb{Z}^n} \sum_{y \in \mathbb{Z}^n} e^{\frac{x^{\top} U\Sigma V^{\top} y}{\sigma^2}} \cdot p(x) \cdot p(y) \\ &\leq \left(1 + \frac{1}{n^{2C}}\right) \sum_{x \in \mathbb{Z}^n}\sum_{y \in \mathbb{Z}^n} \left( e^{\frac{x^{\top}U \Sigma V^{\top} y}{\sigma^2}} \cdot q'(x) \cdot q'(y) \right) \\ &\leq \left( 1 + \frac{1}{n^{2C}} \right) \mathbb{E}_{x,y \sim \calN(0, \sigma^2 I_n)} \left(e^{\frac{x^{\top} U \Sigma V^{\top} y}{\sigma^2} + \frac{1}{n^{2C}}} \right) \\ &= \left( 1 + \frac{1}{\poly(n)}\right) \mathbb{E}_{x,y \sim \calN(0, \sigma^2 I_n)}\left(e^{\frac{x^{\top}U \Sigma V^{\top} y}{\sigma^2}} \right)\\ &= \left( 1 + \frac{1}{\poly(n)} \right) \mathbb{E}_{x,y \sim \mathcal{N}(0, \sigma^2 \cdot I_n)} \left(e^{\frac{x^{\top} \Sigma y}{\sigma^2}}\right) \\  &= \left( 1 + \frac{1}{\poly(n)} \right) \frac{1}{(2\pi \sigma^2)^n} \int_{\mathbb{R}^n \times \mathbb{R}^n} \exp\left[\sum_{i = 1}^n \frac{\sigma_i x_i y_i}{\sigma^2} - \frac{x_i^2 + y_i^2}{2\sigma^2}\right] dx \ dy \\ &= \left( 1 + \frac{1}{\poly(n)} \right) \prod_{i = 1}^n \frac{1}{\sqrt{1 - \sigma_i^2}} \leq \left( 1 + \frac{1}{\poly(n)} \right) \frac{1}{\sqrt{1 - \sum_{i = 1}^n \sigma_i^2}} \\ &= \left( 1 + \frac{1}{\poly(n)} \right) \cdot \frac{1}{\sqrt{1-\|\bA\|_F^2}}
\end{align*}
\end{proof}

Let $N$ be a sufficiently large polynomial in $n$ and $m$. 
Let $\calG(m,n)$ be a distribution on $m\times n$ matrices where each entry is drawn from $\calN(0,N^2)$. 
Let $\calD_1=\calG(m,n)$ and let $\calD_2=\calG(m,n)+\sum_{i=1}^r s_iu^i(v_i)^\top$, where $u_1,\ldots,u_r\sim\calD(0,N\cdot I_m)$ and $v_1,\ldots,v_r\sim\calD(0,N\cdot I_n)$. 
Let $\bB\in\mathbb{R}^{d\times mn}$ be an orthonormal matrix, with $d\le mn$. 
We abuse notation so that each row $B^i$ is viewed as both a vector in $mn$ as well as an $m\times n$ matrix. 
Let $\calB_1$ and $\calB_2$ denote the distributions of $\bB\cdot\bX$ for $\bX\sim\calD_1$ and $\bX\sim\calD_2$, respectively.  The following lemma shows that the distributions $\mathcal{B}_1$ and $\mathcal{B}_2$ are not distinguishable if $d \le c / \norm{s}_2^4$.

\begin{lemma}[Analog of Theorem 4 in~\cite{LiW16}]
\lemlab{lem:tvd:g:spike}
There exists an absolute constant $c>0$ such that for $d\le\frac{c}{\|s\|_2^4}$, we have
\[\TVD(\calB_1,\calB_2)\le\frac{1}{10}.\]
\end{lemma}
\begin{proof}
By rotational invariance of Gaussian distributions, we have that $\bB\cdot\calG(m,n)\sim\calN(0,N^2\cdot\bB\bB^\top)$. 
Since we assume the rows of $\bB$ are orthonormal, then we have $\bB\bB^\top=I_d$. 
Hence, we define $\calL_1=\calN(0,N^2\cdot I_d)$ so that the distribution of the image $\bB\cdot\calG(m,n)$ is $\calL_1$. 

Similarly, let $\calL_2=\calN(0,N^2\cdot I_d)+\nu$, where $\nu$ is the distribution over
\[
\begin{pmatrix}
\sum_{i=1}^r s_i(\bu^i)^\top B^1\bv^i\\
\sum_{i=1}^r s_i(\bu^i)^\top B^2\bv^i\\
\vdots\\
\sum_{i=1}^r s_i(\bu^i)^\top B^d\bv^i
\end{pmatrix}.
\]
Observe that $\calL_2$ is the image of $\bB\cdot\left(\calG(m,n)+\sum_{i=1}^r s_i\bu\bv^\top\right)$. 
Thus our task is to upper bound $\TVD(\calL_1,\calL_2)$. 
To that end, we first define
\[\xi=\sum_{i=1}^d\sum_{j,\ell=1}^r\sum_{a=1}^m\sum_{c=1}^m\sum_{b=1}^n\sum_{d=1}^n s_js_\ell(B^i)_{a,b}(B^i)_{c,d}(\bu^j)_a(\bv^j)_b(\bu^\ell)_c(\bv^\ell)_d.\]
Note that in the expectation of $\xi$, all cross terms vanish due to the symmetry of the probability mass function for discrete Gaussians, except those with $j=\ell$, $a=c$, and $b=d$. 
Therefore,
\[\Ex{\xi}=\sum_{i=1}^d\sum_{j=1}^r\sum_{a=1}^m\sum_{b=1}^ns_j^2(B^i)^2_{a,b}\cdot\Ex{(\bu^j_a)^2}\Ex{(\bv^j_b)^2}=d\cdot N^2\cdot\|s\|_2^2\cdot\sum_{a=1}^m\sum_{b=1}^n(B^i)^2_{a,b}.\]
Since the rows $B^1,\ldots,B^d$ are orthonormal, then we have
\[\Ex{\xi}=d\cdot N^2\cdot\|s\|_2^2.\]
We define $\calE$ to be the event that $\|s\|_2^2\cdot\xi<\frac{1}{2} N^2$ so that by assumption and Markov's inequality, we have $\PPr{\calE}\ge 1-2c$. 
We use $\widetilde{\nu}$ to denote the distribution of $\nu$ conditioned on $\calE$ and we write $\widetilde{\calL}_2=\calN(0,N^2\cdot I_d)+\widetilde{\nu}$. 
By triangle inequality, 
\begin{align*}
\TVD(\calL_1,\calL_2)&\le\TVD(\calL_1,\widetilde{\calL}_2)+\TVD(\widetilde{\calL}_2,\calL_2). 
\end{align*}
Then by \propref{prop:tvd:chi},
\begin{align*}
\TVD(\calL_1,\calL_2)&\le\sqrt{\EEx{\bz_1,\bz_2\sim\widetilde{\nu}}{e^{\langle\bz_1,\bz_2\rangle/N^2}}-1}+\TVD(\widetilde{\nu},\nu).
\end{align*}
By renormalizing $\nu$ to acquire $\widetilde{\nu}$, 
\begin{align*}
\TVD(\calL_1,\calL_2)&\le\sqrt{\frac{1}{\PPr{\calE}}\left(\EEx{\bz_1\sim\nu,\bz_2\sim\widetilde{\nu}}{e^{\langle\bz_1,\bz_2\rangle/N^2}}-1\right)}+\frac{1}{\PPr{\calE}}-1,
\end{align*}
Let $\bx,\by\sim\widetilde{\eta}$. 
Then we have
\begin{align*}
\EEx{\bz_1\sim\nu,\bz_2\sim\widetilde{\nu}}{e^{\langle\bz_1,\bz_2\rangle/N^2}}&=\Ex{\exp\left(\frac{1}{N^2}\sum_{i=1}^d\sum_{a,b,c,d,j,\ell}s_j(B^i)_{a,b}(\bu^j)_a(\bv^j)_b\cdot s_{\ell}(B^i)_{c,d}(\bx^\ell)_c(\by^\ell)_d\right)}.
\end{align*}
We define 
\[Q^{\ell}_{c,d}=s_{\ell}\sum_{i=1}^d\sum_{j,a,b}(B^i)_{a,b}(B^i)_{c,d}s_j(\bu^j)_a(\bv^j)_b,\]
so that
\begin{align*}
\EEx{\bz_1\sim\nu,\bz_2\sim\widetilde{\nu}}{e^{\langle\bz_1,\bz_2\rangle/N^2}}&=\EEx{\substack{\bu^1,\ldots,\bu^r\sim\widetilde{\nu},\\\bv^1,\ldots,\bv^r\sim\widetilde{\nu}}}{\prod_{\ell=1}^r\EEx{\substack{\bx^\ell\sim\calD(0,N\cdot I_m),\\\by^\ell\sim\calD(0,N\cdot I_n)}}{\exp\left(\frac{1}{N^2}\sum_{c,d}Q^\ell_{c,d}(\bx^\ell)_c(\by^\ell)_d\right)}}\\
&=\EEx{\substack{\bu^1,\ldots,\bu^r\sim\widetilde{\nu},\\\bv^1,\ldots,\bv^r\sim\widetilde{\nu}}}{\prod_{\ell=1}^r\EEx{\substack{\bx^\ell\sim\calD(0,N\cdot I_m),\\\by^\ell\sim\calD(0,N\cdot I_n)}}{\exp\left(\frac{1}{N}\sum_{c,d}\frac{1}{N}\cdot Q^\ell_{c,d}(\bx^\ell)_c(\by^\ell)_d\right)}}.
\end{align*}
We first claim that $\|Q^\ell\|_F^2<N^2$ and hence $\cdot\|\frac{1}{N}Q^\ell\|_F^2<1$. 
We have
\begin{align*}
\|Q^\ell\|_F^2&=\sum_{c,d}(Q^\ell)_{c,d}^2\\
&=s_\ell^2\sum_{c,d}\sum_{i,i'}\sum_{j,a,b}\sum_{j',a',b'}s_j(B^i)_{a,b}(B^i)_{c,d}(\bu^j)_a(\bv^j)_b\cdot s_{j'}(B^i)_{a',b'}(B^i)_{c,d}(\bu^{j'})_{a'}(\bv^{j'})_{b'}.
\end{align*}
Since the rows $B^1,\ldots,B^d$ are orthonormal, then for sufficiently large $N$, 
\begin{align*}
\|Q^\ell\|_F^2&\le s_\ell^2\sum_{c,d}\sum_{i}(B^i)_{c,d}^2\sum_{j,a,b}\sum_{j',a',b'}s_j(B^i)_{a,b}(\bu^j)_a(\bv^j)_b\cdot s_{j'}(B^i)_{a',b'}(\bu^{j'})_{a'}(\bv^{j'})_{b'}\\
&\le s_\ell^2\sum_{i}\sum_{j,a,b}\sum_{j',a',b'}s_j(B^i)_{a,b}(\bu^j)_a(\bv^j)_b\cdot s_{j'}(B^i)_{a',b'}(\bu^{j'})_{a'}(\bv^{j'})_{b'}\\
&=s_\ell^2\xi<\frac{1}{2} N^2,
\end{align*}
conditioned on $\calE$. 
Then by \lemref{lem:exp:exp:xtay}, we have that since $\|Q^\ell\|_F^2<1$, then 
\begin{align*}
\EEx{\substack{\bu^1,\ldots,\bu^r\sim\widetilde{\nu},\\\bv^1,\ldots,\bv^r\sim\widetilde{\nu}}}{\prod_{\ell=1}^r\EEx{\substack{\bx^\ell\sim\calD(0,N\cdot I_m),\\\by^\ell\sim\calD(0,N\cdot I_n)}}{\exp\left(\frac{1}{N}\sum_{c,d}\frac{1}{N}\cdot Q^\ell_{c,d}(\bx^\ell)_c(\by^\ell)_d\right)}}&\le\EEx{\substack{\bx^\ell\sim\calD(0,N\cdot I_m),\\\by^\ell\sim\calD(0,N\cdot I_n)}}{\prod_{\ell=1}^r\frac{1}{\sqrt{1-\frac{s_\ell^2\cdot\xi}{N^2}}}}\\
&\le\EEx{\substack{\bx^\ell\sim\calD(0,N\cdot I_m),\\\by^\ell\sim\calD(0,N\cdot I_n)}}{\frac{1}{\sqrt{1-\frac{\|s\|_2^2\cdot\xi}{N^2}}}},
\end{align*}
since $\frac{1}{\sqrt{1-x}}\cdot\frac{1}{\sqrt{1-y}}\le\frac{1}{\sqrt{1-x-y}}$ for $x,y\ge 0$. 
Moreover, we have $\frac{1}{\sqrt{1-x}}\le1+x$ for $x\in\left[0,\frac{1}{2}\right]$. 
Thus,
\begin{align*}
\EEx{\substack{\bu^1,\ldots,\bu^r\sim\widetilde{\nu},\\\bv^1,\ldots,\bv^r\sim\widetilde{\nu}}}{\prod_{\ell=1}^r\EEx{\substack{\bx^\ell\sim\calD(0,N\cdot I_m),\\\by^\ell\sim\calD(0,N\cdot I_n)}}{\exp \left(\frac{1}{N^2}\sum_{c,d}Q^\ell_{c,d}(\bx^\ell)_c(\by^\ell)_d\right)}}&\le1+\frac{\|s\|_2^2\cdot\Ex{\xi}}{N^2}\\
&\le 1+d\cdot\|s\|_2^4.
\end{align*}
Hence, 
\[\TVD(\calL_1,\calL_2)\le\sqrt{\frac{d\cdot\|s\|_2^4}{1-2c}}+\frac{2c}{1-2c}\le\sqrt{\frac{c}{1-2c}}+\frac{2c}{1-2c}\le\frac{1}{20},\]
for sufficiently small $c>0$. 
\end{proof}

We now give a lower bound for $\alpha$-approximation to the operator norm of a matrix $\bX \in \mathbb{Z}^{n \times n}$. 
\begin{corollary}[$\alpha$-approximation to operator norm]
\corlab{cor:op_lb}
Let $\alpha > 1 + c$ be an arbitrarily small constant $c$. 
Then any integer sketch that estimates $\|\bX\|_{op}$ for $\bX \in \mathbb{Z}^{n \times n}$ within a factor of $\alpha$ with error probability $\leq 1/6$ requires sketching dimension $\Omega\left(\frac{n^2}{\alpha^4}\right)$.
\end{corollary}
\begin{proof}
Let $\calG'(m,n)$ be a distribution on $m\times n$ matrices where each entry is drawn from $\mathcal{D}(0,N^2)$. 
Consider the discrete analog of the above two distributions where we let $\calD_1=\calG'(m,n)$ and let $\calD_2=\calG'(m,n)+\sum_{i=1}^r s_iu^i(v_i)^\top$, where $u_1,\ldots,u_r\sim\calD(0,N\cdot I_m)$ and $v_1,\ldots,v_r\sim\calD(0,N\cdot I_n)$. 

Let $m = n$ and let $r = 1$, $s_1 = \gamma_1\alpha/\sqrt{n}$ from some constant $\gamma_1$ large enough in $\mathcal{D}_2$. 
Then, the result follows by applying \lemref{lem:tvd:g:spike} and \thmref{thm:lifting} if we can show an $\alpha$-gap of the operator norm between these two distributions.

We next show that this gap holds with high constant probability. By \lemref{lem:disc:subgauss} and \lemref{lem:two:sided:oper}, there exists an absolute constant $C>0$ such that 
\[\|\bG \|_{op} \leq 3CN\sqrt{n},\]
with probability at least $1-2e^{-n}$. 

On the other hand, since $s_1\bu\bv^\top$ is a rank-one matrix, we have 
\[\|s_1\bu\bv^\top\|_{op}=s_1\|\bu\|_2\|\bv\|_2\ge\frac{\gamma_1}{\sqrt{n}}\cdot\alpha\cdot(Nn),\]
for some absolute constant $\gamma_2>0$. 
Therefore by triangle inequality,
$$\|\bG+s_1\bu\bv^\top\|_{op}\ge\|s_1\bu\bv^\top\|_{op}-\|\bG\|_{op}\ge\frac{\gamma_1\alpha}N\sqrt{n}-3CN\sqrt{n}>3\alpha CN\sqrt{n},$$
for sufficiently large constant $\gamma_1>0$. 
\end{proof}

Next, we show a dimension lower bound for any integer linear sketch which $(1+\eps)$-approximates the operator norm of a matrix $\bX$. 
\begin{lemma}
\lemlab{lem:op:eps:lb}
Given an approximation parameter $\eps\in\left(0,\frac{1}{3}\right)$, any integer sketch that outputs a $(1+\eps)$-approximation for the operator norm of a matrix $\bX\in\mathbb{Z}^{(d/\eps^2)\times d}$ with probability at least $\frac{5}{6}$ uses sketching dimension $\Omega\left(\frac{d^2}{\eps^2}\right)$. 
\end{lemma}
\begin{proof}
Let $m = d/\eps^2$ and $n = d$. 
Let $\bX \in \mathbb{R}^{(d/\eps^2) \times d}$. 
As before, take $r = 1$ and let $s_1 = \alpha \sqrt{\eps/d}$ for some constant $\alpha > 0$ large enough to apply \lemref{lem:tvd:g:spike}.  
Now, it suffices for us to show that $\bG$ and $\bG + \alpha \sqrt{\frac{\eps}{d}}\bu\bv^{\top}$ differ in operator norm by a factor of $1 + \eps$ for some constant $\alpha$, where each entry of $\bG$ is drawn from $\calD(0,N^2)$, and each entry of $\bu\in\mathbb{Z}^{d/\eps^2}$ and $\bv\in\mathbb{Z}^d$ is drawn independently from $\calD(0,N)$, so that $\bu\sim \calD(0, N\cdot I_{d/\eps^2})$, and $\bv \sim \calD(0, N \cdot I_{d})$. 
By \lemref{lem:disc:subgauss} and \lemref{lem:two:sided:oper}, there exists an absolute constant $C>0$ such that 
$$\|\bG \|_{op} \leq CN\left(\frac{\sqrt{d}}{\eps} + 2\sqrt{d}\right) = CN(1 + 2\eps)\cdot\frac{\sqrt{d}}{\eps},$$
with probability at least $1-2e^{-d}$. 
Now, we will also show that 
$$\left \|\bG + \alpha \sqrt{\frac{\eps}{d}}\bu\bv^{\top} \right \|_{op}  \geq CN(1 + 4\eps)\cdot\frac{\sqrt{d}}{\eps}\ge(1+\eps)\cdot\|\bG\|_{op}.$$ 
To see this, we note that 
\begin{align*}
    \left\|\bG + \alpha \sqrt{\frac{\eps}{d}} \bu \bv^{\top} \right\|_{op} &= \sup_{\bx \in \mathbb{S}^{n-1}} \left \| \left(\bG + \alpha \sqrt{\frac{\eps}{d}} \bu \bv^{\top}\right)\bx \right \|_2 \geq \left \| \left(\bG + \alpha \sqrt{\frac{\eps}{d}} \bu \bv^{\top}\right)\frac{\bv}{\|\bv\|_2} \right \|_2 \\ &= \left \|\bG\cdot\frac{\bv}{\|\bv \|_2} + \alpha \sqrt{\frac{\eps}{d}} \bu \|\bv\|_2 \right \|_2.
\end{align*}
By squaring both sides, we have that
$$\left\|\bG + \alpha \sqrt{\frac{\eps}{d}} \bu \bv^{\top} \right\|_{op}^2 \geq \left \|\bG\cdot\frac{\bv}{\|\bv \|_2} + \alpha \sqrt{\frac{\eps}{d}} \bu \|\bv\|_2 \right \|_2^2 = \frac{\|\bG\bv \|_2^2}{\| \bv\|_2^2} + 2 \alpha \sqrt{\frac{\eps}{d}} \langle \bG\bv, \bu\rangle + \frac{\alpha^2 \eps}{d} \|\bv\|_2^2 \|\bu\|_2^2.$$
By \lemref{lem:two:sided:oper}, we have that 
\[\frac{\|\bG\bv \|_2^2}{\| \bv\|_2^2}\ge\left(CN(1-2\eps)\cdot\frac{\sqrt{d}}{\eps}\right)^2,\]
with probability at least $1-2e^{-d}$. 

Now, consider the setting where $\bG\in\mathbb{R}^{d/\eps^2\times d}$ has real-valued entries independently drawn from the continuous distribution $\calN(0, N^2)$. 
Then for any fixing of $\bv\in\mathbb{R}^d$, we have that $\bG\bv~\sim\calN(0,N^2\cdot\|v\|_2^2\cdot I_{d/\eps^2})$. 
Therefore, for any fixing of $\bu\in\mathbb{R}^{d/\eps^2}$, we have that $\langle\bG\bv,\bu\rangle\sim\calN(0,N^2\cdot\|u\|_2^2\cdot\|v\|_2^2)$. 
Hence, there exists a constant $\gamma_5$ such that with high probability, 
\[\langle\bG\bv,\bu\rangle\ge-\gamma_5 N^2\cdot\frac{d}{\eps}.\] 
Moreover, there exists a constant $\gamma_6$ such that with high probability, 
\[\|\bv\|_2^2\cdot\|\bu\|_2^2\ge\gamma_6 N^2\cdot\frac{d^2}{\eps^2}.\] 
Thus if each entry of $\bG$ is drawn from $\calN(0,N^2)$, then we have with high probability,
\[2\alpha\sqrt{\frac{\eps}{d}}\langle\bG\bv,\bu\rangle+\frac{\alpha^2\eps}{d}\|\bv\|_2^2\cdot\|\bu\|_2^2\ge-2\alpha\gamma_5 N^2\sqrt{\frac{d}{\eps}}+\alpha^2\gamma_6 N^2\cdot\frac{d}{\eps}.\] 

Since we draw each entry of $\bG$ from $\calD(0, N^2)$ rather than $\calN(0, N^2)$, then for sufficiently large $N$, we have that with high probability, 
\[\frac{\|\bG\bv \|_2^2}{\| \bv\|_2^2} + 2 \alpha \sqrt{\frac{\eps}{d}} \langle \bG\bv, \bu\rangle + \frac{\alpha^2 \eps}{d} \|\bv\|_2^2 \|\bu\|_2^2\ge\left(CN(1-2\eps)\cdot\frac{\sqrt{d}}{\eps}\right)^2-2\alpha\gamma_5 N^2\sqrt{\frac{d}{\eps}}+\alpha^2\gamma_6 N^2\cdot\frac{d}{\eps}.\]
Therefore, for sufficiently large $\alpha$, we have that
$$\left \|\bG + \alpha \sqrt{\frac{\eps}{d}}\bu\bv^{\top} \right \|_{op}  \geq CN(1 + 4\eps)\cdot\frac{\sqrt{d}}{\eps}\ge(1+\eps)\cdot\|\bG\|_{op}.$$  
\end{proof}
Finally, we show a lower bound for obtaining a constant-factor approximation to the Ky Fan norm of a matrix $\bX$. 
\begin{corollary}[Ky Fan norm lower bound]
    \corlab{cor:KyFan_lb}
    There exists an absolute constant $c > 0$ such that any integer linear sketch that estimates $\norm{\bX}_{F_S}$ for $\bX \in \mathbb{Z}^{n \times n}$ and $s \le\O{\sqrt{n}}$ within a factor of $1 + c$ with error probability $1/6$ requires sketching dimension $\Omega(n^2 / s^2)$. 
\end{corollary}
\begin{proof}
Let $m = n$ and $r = s$ where $s_1 = s_2 = \ldots = s_r = \gamma /\sqrt{n}$ in $\mathcal{D}_2$ for some sufficiently large constant $C$.  
Then, the result follows by applying \lemref{lem:tvd:g:spike} and \thmref{thm:lifting} if we can show a $(1 + c)$-gap of the KyFan $s$ norm between these two distributions.

When $X \sim \mathcal{D}_1$, by \lemref{lem:disc:subgauss} and \lemref{lem:two:sided:oper}, there exists an absolute constant $C>0$ such that 
\[\|\bG \|_{op} \leq CN\sqrt{n},\]
with probability at least $1-2e^{-n}$. 
Thus $\norm{X}_{F_s} \le Cs N\sqrt{n}$ with high probability.

When $X \sim \mathcal{D}_2$, we can write $\bX = \bG + \frac{\gamma}{\sqrt{n}}\bP$ where $\bP = \bu_1 \bv_1^T + \bu_2 \bv_2^T + \ldots + \bu_s \bv_s^T$. 
We claim that with high probability $\norm{\bP}_1 \ge 0.9Nsn$ and thus $\norm{\bX}_{F_s} \ge \frac{C}{\sqrt{n}} \norm{\bP}_{F_s} - \norm{\bG}_{F_s} \ge 0.9 \gamma Ns\sqrt{n} - CNs\sqrt{n}$, showing a multiplicative gap of $\norm{\bX}_{F_s}$ between the two distributions.

Now we prove this claim. 
Since $\bu_i, \bv_i$ are drawn from the discrete Gaussian distribution. With high probability we have $0.99 \sqrt{nN} \le \norm{\bu_i}_2 \le 1.01 \sqrt{nN}$ for all $i$ and $|\sum_{i \ne j} \langle \bu_i, \bu_j\rangle| \le 1.01 s\sqrt{nN}$. 
Then, following the same procedure in the Corollary 8 in~\cite{LiW16}, by the min-max theorem for singular values we have 
\[
\sigma_{\ell}^2 (\bP) \ge 0.99^2 n \sigma_{\ell}^2(\bV) - 1.01^3 s n^{3/2}N^2,
\]
where $\bV$ is a $k \times n $ matrix with rows $\bv_1^T, \ldots, \bv_s^T$. 
Therefore, we have 
\[
\norm{\bP}_1 \ge 0.99\sqrt{n} \norm{\bV}_1 - 1.01^{3/2} s^{1/2} n^{3/4} N \;.
\]
Since $\bV$ is a Gaussian matrix, from \lemref{lem:two:sided:oper} on $\bV^T$ we have $\norm{\bV} \ge C_2 s \sqrt{n}N$ with high probability. 
The claim follows from our assumption.
\end{proof}

\subsection{Eigenvalue Approximation and PSD Testing}
In this section, we consider the problems of eigenvalue approximation and PSD testing. 
In both problems, the input is a matrix $\bM\in\mathbb{Z}^{n\times n}$. 
In the eigenvalue approximation problem, the goal is to estimate each eigenvalue of $\bM$ up to an additive error $\eps\cdot\|\bM\|_F$. 
On the other hand, in the PSD testing problem, the goal is to either output YES if $\bM$ is positive semidefinite, so that $\bv\bM\bv^\top \geq 0$ for all vectors $\bv\in\mathbb{R}^n$, or output NO if $\bM$ has an eigenvalue of value at most $-\eps\cdot\|\bM\|_p$, where $\|\bM\|_p$ denotes the Schatten-$p$ norm of $\bM$. In fact, for both of these problems, the hard distribution is inspired by the hard distribution for approximating the operator norm (See \secref{sec:op-ky}).

We first give our lower bound for the eigenvalue approximation problem. 
\begin{theorem}
\thmlab{thm:eigen:lb}
Given an approximation parameter $\eps\in\left(0,\frac{1}{3}\right)$, any integer sketch that outputs additive $\eps\cdot\|\bM\|_F$ approximations to the eigenvalues of a matrix $\bM \in \mathbb{Z}^{d\times d}$ with probability at least $\frac{3}{4}$ uses sketching dimension $\Omega\left(\frac{1}{\eps^4}\right)$ for $d=\Omega\left(\frac{1}{\eps^2}\right)$.
\end{theorem}
\begin{proof}
We consider two matrices $\bG$ and $\bH$ as possible inputs for $\bM$. 
Let $s\in(0,1)$ be a parameter that we shall set. 
We define $\bG\in\mathbb{R}^{d\times d}$ to be a matrix with entries drawn independently at random from $\calD(0,N^2)$ and $\bH=\bG+s\bu\bv^\top$, where $\bu,\bv\in\mathbb{R}^d$ have entries drawn independently at random from $\calD(0,N)$. 
By \lemref{lem:two:sided:oper}, there exists a universal constant $C_1>0$ such that $\|\bG\|_{op}\le C_1N\sqrt{d}$ with probability at least $1-2e^{-d}$. 

With probability at least $1-2e^{-d}$, we have that $\|\bu\|_2=\Theta(\sqrt{Nd})$. 
Similarly, with probability at least $1-2e^{-d}$, we have that $\|\bv\|_2=\Theta(\sqrt{Nd})$. 
Because $s\bu\bv^\top$ is a rank-one matrix, then there exists a universal constant $C_2>0$ such that
\[\|s\bu\bv^\top\|_{op}=s\|\bu\|_2\|\bv\|_2\ge C_2sNd,\]
with probability at least $1-4e^{-d}$ by a union bound. 
By triangle inequality,
\[\|\bH\|_{op}\ge\|s\bu\bv^\top\|_{op}-\|\bG\|_{op}\ge C_2sNd-C_1N\sqrt{d}.\]

Next, note that by Gaussian concentration, $\|\bG\|_F=\Theta(Nd)$ with high probability for a universal constant $C_3>0$. 
Furthermore, we have 
\[\|\bG+s\bu\bv^\top\|_F\le\|\bG\|_F+\|s\bu\bv^\top\|_F\le s\|\bu\|_2\cdot\|\bv\|_2=\Theta(Nd),\]
with high probability, since $s<1$. 

Hence if $C_2sNd-C_1N\sqrt{d}-C_1N\sqrt{d}\ge\eps\cdot\Theta(Nd)$, then any linear sketch that approximates the eigenvalues of $\bM$ up to additive $\eps\cdot\|\bM\|_F$ can distinguish whether the input matrix is $\bG$ or $\bH$. 
Note that the inequality is satisfied for $s=\O{\eps}$. 
However, by \lemref{lem:tvd:g:spike}, any linear sketch for an input matrix with $\O{d^2}$ rows that can distinguish between $\bG$ and $\bH$ with probability at least $\frac{3}{4}$ requires at least $\Omega\left(\frac{1}{s^4}\right)$ rows. 
Therefore provided $d=\Omega\left(\frac{1}{\eps^2}\right)$, the linear sketch requires $\Omega\left(\frac{1}{\eps^4}\right)$ rows. 
\end{proof}
Next, we give our lower bound for PSD testing. 
\begin{theorem}
\thmlab{thm:psd:lb}
Given a distance parameter $\eps$, any integer sketch that reads $\bM\in\mathbb{Z}^{d\times d}$ and serves as a two-sided tester for whether $\bM$ is PSD or $\eps$-far from PSD in $\ell_p$ distance with probability at least $\frac{3}{4}$ requires:
\begin{enumerate}
\item
$\Omega\left(\frac{1}{\eps^{2p}}\right)$ rows for $p\in[1,2]$. 
\item 
$\Omega\left(\frac{1}{\eps^4}d^{2-4/p}\right)$ rows for $p\in(2,\infty)$.
\item
$\Omega(d^2)$ for $p=\infty$. 
\end{enumerate}
\end{theorem}
\begin{proof}
Let $\bG\in\mathbb{R}^{d\times d}$ be a matrix with entries drawn independently at random from $\calD(0,N^2)$. 
Let $\bH=\bG+s\bu\bv^\top$, where $\bu,\bv\in\mathbb{R}^d$ have entries drawn independently at random from $\calD(0,N)$. 
By \lemref{lem:two:sided:oper}, we have that there exists an absolute constant $C_2>0$ such that $\|\bG\|_{op}\le C_1N\sqrt{d}$ with probability at least $1-2e^{-d}$. 
We define
\[\bG_{\text{sym}}=
\begin{bmatrix}
    0 & \bG \\
    \bG^\top & 0
\end{bmatrix},\qquad
\bH_{\text{sym}}=
\begin{bmatrix}
    0 & \bH \\
    \bH^\top & 0
\end{bmatrix}
.\]
Since the eigenvalues of $\bG_{\text{sym}}$ are exactly $\pm\sigma_i(\bG)$ for $i\in[d]$, where $\sigma_i(\bG)$ denotes the singular values of $\bG$. 
Hence, the matrix $\bG_{\text{sym}}+CN\sqrt{d}\cdot I_{2d}$ is PSD with probability at least $1-2e^{-d}$ for any $C\ge C_1$. 

Since $s\bu\bv^\top$ is a rank-one matrix, then there exists an absolute constant $C_2>0$ such that
\[\|s\bu\bv^\top\|_{op}=s\|\bu\|_2\|\bv\|_2\ge C_2sNd,\]
with probability at least $1-4e^{-d}$ by union bounding over the events that $\|\bu\|_2=\Omega(\sqrt{Nd})$ and $\|\bv\|_2=\Omega(\sqrt{Nd})$. 
By triangle inequality,
\[\|\bH\|_{op}\ge\|s\bu\bv^\top\|_{op}-\|\bG\|_{op}\ge C_2sNd-C_1N\sqrt{d}.\]
Because the eigenvalues of $\bH_{\text{sym}}$ are exactly $\pm\sigma_i(\bH)$ for $i\in[d]$, where $\sigma_i(\bG)$ denotes the singular values of $\bG$, then the matrix $\bH_{\text{sym}}+CN\sqrt{d}\cdot I_{2d}$ has a negative eigenvalue with value at most $-(C_2sNd-C_1N\sqrt{d})+CN\sqrt{d}$, i.e., magnitude at least $C_2sNd-2CN\sqrt{d}$ for $C\ge C_1$. 

Moreover, we have the Schatten-$p$ norm satisfies $\|\bG_{\text{sym}}\|_p\le C_1Nd^{1/2+1/p}$ with probability at least $1-2e^{-d}$, since the operator norm of $\bG$ is at most $C_1N\sqrt{d}$ with probability at least $1-2e^{-d}$. 
Therefore, for
\begin{align}
\label{eqn:psd:cond}
\eps\le\frac{C_2sNd-2CN\sqrt{d}}{C_1Nd^{1/2+1/p}},
\end{align}
then a two-sided PSD-tester can distinguish between $\bG$ and $\bH$ with high probability. 
However, by \lemref{lem:tvd:g:spike}, any sketch that uses $\O{d^2}$ rows and can distinguish between $\bG$ and $\bH$ with probability at least $\frac{3}{4}$ requires at least $\Omega\left(\frac{1}{s^4}\right)$ rows. 
In particular, for $p\in[1,2]$, we can confirm that the setting $d=\Theta\left(\frac{1}{\eps^p}\right)$ and $s=\O{\eps^{p/2}}$ satisfies the inequality in (\ref{eqn:psd:cond}). 
Therefore, $\Omega\left(\frac{1}{\eps^{2p}}\right)$ sketching dimension is necessary for $p\in[1,2]$. 
Similarly, it can be confirmed that $d=\Omega\left(\frac{1}{\eps^p}\right)$ and $s=\O{\eps d^{1/p-1/2}}$ satisfies the inequality in (\ref{eqn:psd:cond}) for $p\in(2,\infty)$, yielding a lower bound of $\Omega\left(\frac{1}{\eps^4}d^{2-4/p}\right)$. 
Finally, the setting of any arbitrary $d$ and $s=\O{\frac{1}{\sqrt{d}}}$ gives a lower bound of $\Omega(d^2)$ for $p=\infty$. 
\end{proof}

\subsection{Compressed Sensing} 
In the $\ell_2/\ell_2$-sparse recovery problem, we want to find a sparse vector $\hat{\bx}$ such that $\norm{\hat{\bx} - \bx}_2 \le (1 + \eps) \min_{\text{$k$-sparse } \tilde{\bx}} \norm{\tilde{\bx} -  \bx}_2$. 
We consider the following discrete analogue of the construction of the input distribution in~\cite{PriceW11}:
Let $\mathcal{F} \subset \{S \subset [n] \mid \abs{S} = k\}$ be a family of $k$-sparse supports such that:
\begin{itemize}
\item $\abs{S \Delta S'} \geq k$ for $S \neq S' \in \mathcal{F}$,
\item $\Pr_{S \in \mathcal{F}} [i \in S] = k/n$ for all $i \in [n]$, and
\item $\log \abs{\mathcal{F}} = \Omega(k \log (n/k))$.
\end{itemize}

Let $X = \{\bx \in \{0, \pm \sqrt{N}\}^n \mid \supp(\bx) \in \mathcal{F}\}$ for some large enough $N = \poly(n)$. 
Let $\bw \sim \mathcal{D}(0, \eps N\frac{k}{n} I_n)$ be i.i.d discrete Gaussian in each coordinate. 
The procedure is defined as follows. 
Given an input $\bx + \bw$, let $\bx' = g(\bA(\bx + \bw))$ for a sketching matrix $\bA$ and a post-processing function $g$. 
Then we round $\bx'$ to $\hat{\bx}$ by $\hat{\bx} = \argmin_{\hat{\bx} \in X} \norm{\hat{\bx} - \bx'}_2$. 
This gives a Markov chain $S \to \bx \to \bx' \to S'$. 
In this section, we will focus the case that $\bx'$ has sparsity $\O{\frac{\eps n}{\log n}}$.

Given an $S \in \mathcal{F}$, define the distribution $\mathcal{L}_S: \bz + \bw$ where $\bz \in \{\{0, \pm \sqrt{N}\}^n \ | \ \mathrm{supp}(\bx) = S\}$. 
We show the following claim: suppose there is an algorithm $\calA$ that returns a $(1 + \eps)$-sparse recovery of $\bx$ with high constant probability and uses an integer linear sketch $\bA$ with estimator $g$, where from $\bA \bx, \bx \sim \mathcal{L}_S$. Then, we will construct another algorithm $\calA'$ which uses an estimator $h$ (which is related to $g$) such that $\calA'$ can identify the support $S$ of $\bx$, where $\bx \sim \mathcal{L}'_S: {\bz + \bw}$ and $\bw \sim \mathcal{N}(0, \eps N\frac{k}{n} I_n)$ is the continuous Gaussian noise. Note that a lower bound of $\Omega(\eps^{-1} k \log (n / k))$ was proved for the latter distribution in ~\cite{PriceW11}.

Suppose that $\bx' = g(\bA\bx)$ is a $(1 + \eps)$-sparse recovery of $\bx$. Then, we first show that for $\bx \sim \mathcal{L}_S$, we can recover the support $S$ from $\bx'$ with high probability.

\begin{lemma}
Suppose that $\bA \in \mathbb{Z}^{r \times n}$ is the sketching matrix and $\bx' = g(\bA (\bx + \bw))$ has at most $c \eps n / \log n$ non-zero coordinates for some constant $c$. Then with high probability we have from $\bx'$ we can recover the support $S$.  
\end{lemma}
\begin{proof}
Let $S'$ be the support of $\hat{\bx}$ and $T = S \cup S'$.
Following the same argument in Lemma~4.3 in~\cite{PriceW11} we can have that if $\bx'$ is a $(1 + \eps)$-sparse recovery of $\bx$, then we have that either $\norm{\bx'_T - \bx}_2^2 \le 9 \eps \norm{\bw}_2^2$ or $\norm{\bx'_{\bar{T}} - \bw}_2^2 \le (1 - 2\eps) \norm{\bw}_2^2$ and for the first case we can further argue that $S' = S$. 
Hence, we only need to show that with high constant probability, the second case will not happen. 

To show this, let $S''$ be the support of $\bx'$, note that since $\bx'$ has $c\eps n /\log n$ sparsity and $\bw$ is drawn from the discrete Gaussian distribution $\mathcal{D} (0, \eps N\frac{k}{n} I_n)$, we have with high probability, $\norm{\bw}_2^2 = \Theta(\eps N k)$ and $\norm{\bw_{S''}}  \le \O{\frac{\eps Nk \log n}{n} \cdot\frac{\eps n}{\log n}} = \O{\eps^2 N k}$, which means that for a sufficiently constant $c$ we have $\norm{\bx'_{\bar{T}} - \bw}_2^2 \ge (1 - 2\eps) \norm{\bw}_2^2$ holds with high probability, which is what we need.
\end{proof}
Now, consider the continuous case where $\bx \sim \mathcal{L}'_S: {\bz + \bw}$ and $\bw \sim \mathcal{N}(0, \eps N\frac{k}{n} I_n)$. 
It has been shown in~\cite{PriceW11} that in this case, with high probability, we can also recover the support $S$. 
Recall that $\bz$ is sampled from an integer distribution. 
Hence, from \thmref{thm:lifting} with covariance matrix $\bS = \eps N \frac{k}{k} I_n$, we can modify $g$ to another function $h$ with sketching matrix $\bA'$ where from $h(\bA' \bx)$ we can recover $S$ with high probability. 
From this we get the following lower bound.

\begin{lemma}
\lemlab{lem:compress_sensing_lb}
Suppose that $\eps > \sqrt{\frac{k\log n}{n}}$. 
Any integer sketch that outputs an $\O{\frac{\eps n}{\log n}}$-sparse vector $\bx'$ such that $\norm{\bx - \bx'}_2^2 \le (1 + \eps) \min_{\text{$k$-sparse } \tilde{\bx}} \norm{\tilde{\bx} -  \bx}_2$ with high constant probability requires sketching dimension $\Omega\left(\frac{k}{\eps}\log\frac{n}{k}\right)$.
\end{lemma} 

\section*{Acknowledgements}
Elena Gribelyuk and Huacheng Yu are supported in part by an NSF CAREER award CCF-2339942. 
Honghao Lin was supported in part by a Simons Investigator Award, NSF CCF-2335412, and a CMU Paul and James Wang Sercomm Presidential Graduate Fellowship. 
David P. Woodruff was supported in part by a Simons Investigator Award and NSF CCF-2335412. 
Samson Zhou is supported in part by NSF CCF-2335411. 
The work was conducted in part while David P. Woodruff and Samson Zhou were visiting the Simons Institute for the Theory of Computing as part of the Sublinear Algorithms program.

\def\shortbib{0}
\bibliographystyle{alpha}
\bibliography{references}

\end{document}